\documentclass[fleqn]{lmcs}
\pdfoutput=1

\usepackage{lastpage}
\lmcsdoi{17}{2}{6}
\lmcsheading{}{\pageref{LastPage}}{}{}%
{Jun.~16,~2020}{Apr.~14,~2021}{}

\usepackage[british]{babel}
\usepackage[utf8]{inputenc}
\usepackage[T1]{fontenc}
\usepackage{amsfonts,amssymb,mathdefs,microtype,array,mathrsfs,amsthm,graphicx,relsize}

\usepackage[section]{thmdefs}

\let\sc\scshape

\makeatletter
\newcommand\m@thsm@ller[2]{\mbox{\relscale{0.91}$\m@th#1#2$}}
\let\smaller\undefined
\DeclareRobustCommand\smaller[1]{\relax\ifmmode{\mathpalette\m@thsm@ller{#1}}\else{\relscale{0.91}#1}\fi}
\makeatother

\newlist{enuma}{enumerate}{10}
\setlist[enuma]{label={\normalfont(\alph*)}}
\newlist{enumr}{enumerate}{10}
\setlist[enumr]{label={\normalfont(\roman*)}}
\newlist{enum1}{enumerate}{10}
\setlist[enum1]{label={\normalfont(\arabic*)}}

\DeclareRobustCommand*{\dom}{\qopname\relax o{dom}}
\DeclareRobustCommand*{\rng}{\qopname\relax o{rng}}

\newcommand*{\id}{\mathrm{id}}

\newcommand*{\Pos}{\mathsf{Pos}}
\newcommand*{\Set}{\mathsf{Set}}
\newcommand*{\Alg}{\mathsf{Alg}}
\newcommand*{\FAlg}{\mathsf{FAlg}}

\newcommand*{\PAlg}{\mathsf{PAlg}}

\newcommand*{\PSp}{\mathsf{PSp}}
\newcommand*{\Dist}{\mathsf{Dist}}

\newcommand*{\Log}{\mathsf{Log}}
\newcommand*{\Alp}{\mathsf{Alph}}

\newcommand*{\val}{\mathrm{val}}

\newcommand*{\Flat}{\mathrm{flat}}
\newcommand*{\sing}{\mathrm{sing}}

\newcommand*{\reg}{\mathrm{reg}}

\newcommand*{\op}{\mathrm{op}}

\newcommand*{\MSO}{\smaller{\mathrm{MSO}}}

\newcommand*{\FO}{\smaller{\mathrm{FO}}}

\newcommand*{\Th}{\mathrm{Th}}
\newcommand*{\Mod}{\mathrm{Mod}}
\newcommand*{\Id}{\mathrm{Id}}

\newcommand*{\Syn}{\mathrm{Syn}}
\newcommand*{\syn}{\mathrm{syn}}

\newcommand*{\?}{\kern .08em}

\DeclareRobustCommand*{\Aboveseg}{\mathord\Uparrow}

\newcommand\upqed{\vskip-\baselineskip\vskip-\belowdisplayskip}
\newcommand\markenddef{\hfill$\lrcorner$}

\keywords{algebraic language theory, monads, Eilenberg--Moore algebras}

\begin{document}
\title[Algebraic Language Theory for Eilenberg--Moore Algebras]{Algebraic Language Theory\texorpdfstring{\\}{ }for Eilenberg--Moore Algebras}
\author{Achim Blumensath}
\address{Masaryk University Brno}
\email{blumens@fi.muni.cz}
\thanks{Work partially supported by the Czech Science Foundation, grant No.~GA17-01035S}
\maketitle

\begin{abstract}\noindent
We develop an algebraic language theory based on the notion of an
Eilenberg--Moore algebra. In comparison to previous such frameworks the main
contribution is the support for algebras with infinitely many sorts
and the connection to logic in form of so-called `definable algebras'.
\end{abstract}

\tableofcontents

\section{Introduction}   

There are various approaches to formal language theory, each having
its own strengths and weaknesses.
We are here interested in the algebraic approach -- in particular,
in its use in characterising subclasses of regular languages,
like the class of first-order definable languages.

The initial algebraic theory was developed for languages of finite words.
It has subsequently been generalised, first to infinite words (see, e.g.,~\cite{PerrinPin04})
and then to finite trees (e.g.,~\cite{BojanczykWalukiewicz07}).
More recently, also a framework for dealing with infinite trees was developed
\cite{BojanczykId09,Blumensath11c,BojanczykIdSk13,Blumensath13a,Blumensath20,BlumensathZZ}.
Each of these four theories comes in several different variants, depending on
which notion of a language or a logic they were designed for.
As usual when such a wealth of slightly different settings has been developed,
people have started to consolidate and unify them.
One well-known proposal of this kind, based on the formalism of Eilenberg--Moore algebras,
was put forward by Boja\'nczyk~\cite{Bojanczyk15}.
Later on it was generalised by Salamanca~\cite{Salamanca17} and by
Ad\'amek, Chen, Milius and Urbat~\cite{UrbatChAdMi2017,MiliusUrbat19}.
The present paper is built on these works.

The motivation for such generalisation initiatives originates in several places.
First of all, setting up a new algebraic framework for language theory
entails a lot of grunt work and usually results in papers that are quite long
and in large parts not very deep.
Having most of the common parts extracted into a general framework reduces
a lot of this work and allows one to focus on the parts that contain the ideas which are new.

Apart from potentially saving a lot of work, such a program can also lead to novel insights.
When proving a general result one is usually forced to isolate
the key properties and notions that are required for the proof
(such as \emph{denseness} and \emph{$\bbM$-compositionality,} which we will introduce below).
This in turn provides insight into
how far the methods used can be extended and where their limits are.

Finally, in a concrete case there are often several possible variations
of the definitions that more-or-less work equally well.
Knowing which of them generalises helps one to evaluate their respective merits.

Besides hopefully improving upon the presentation,
the main contributions of the present article lie in two areas.
Firstly, we present the first framework that does support algebras with
infinitely many sorts, which is required when one wants to cover the
concrete frameworks that have been introduced for languages of infinite trees.
While it turns out that many results and proofs go through for infinitely many sorts with
nearly no changes, there are also a few places below where we are forced to make
substantial adjustments.
In particular, we introduce the notion of a \emph{dense} morphism of monads
in Section~\ref{Sect: synt. alg.} to prove the existence of syntactic algebras,
and we have to modify the definition of a pseudo-variety in Section~\ref{Sect: varieties}
by adding closure under so-called \emph{sort-accumulation points.}

Secondly, the existing frameworks concentrate on the algebraic and
language-theoretic side of things, while mostly ignoring the connections to logic.
This is rather unfortunate, as logic is one of the main application areas
for algebraic language theory. We will therefore devote a substantial part
of the article to the connection between the algebraic theory and the study of logics.

\smallskip
The overview of the article is as follows.
In Section~\ref{Sect:algebras} we set up the basic toolkit of monads and
Eilenberg--Moore algebras which our algebraic framework is based on.
Our preparations continue in Section~\ref{Sect:quotients} with the development of
a theory of quotients and congruences for such algebras.

Our algebraic framework is set up in Sections \ref{Sect: synt. alg.}--\ref{Sect:axioms}.
The central notions of a syntactic congruence and a syntactic algebra are introduced
in Section~\ref{Sect: synt. alg.}.
Equipped with these tools, we study pseudo-varieties in Section~\ref{Sect: varieties}
and derive our version of the Variety Theorem.
Section~\ref{Sect:profinitary} contains technical material on profinitary monads,
which is needed to prove a version of Reiterman's Theorem in Section~\ref{Sect:axioms}.

The second part of the article consists of Sections \ref{Sect:logic}--\ref{Sect:applications}.
We start by collecting a few basic notions from logic in Section~\ref{Sect:logic}.
Section~\ref{Sect:definable} contains the connection to language theory
in terms of algebras whose products are definable in a certain sense.
Finally in Section~\ref{Sect:applications},
we show how one can apply our framework to study monadic second-order logic
and first-order logic over infinite trees.

\section{Monads and algebras}   
\label{Sect:algebras}

We assume that the reader is reasonably familiar with basic notions of category theory.
But in order to make the article more accessible to readers from other fields,
we have tried not to rely on any concepts that are not covered by the usual introductory
text books.
As a consequence we will explicitly define any of the more specialised notions needed below
-- such as that of a monad or a copresentable object.

Let me also make a philosophical remark.
In this article I~have tried to strike a balance between the level of generality
of the framework and the technical overhead entailed by it.
For this reason, many of the results below will not be stated in the most general form possible.
Instead, I~have adopted a level of generality that covers (most~of) the intended applications
while not obscuring the proofs by pointless technicalities.
In particular, the framework below is not presented in a purely category-theoretical language,
but in a mixture of set theory and category theory.

\smallskip
In formal language theory one studies sets of labelled objects like words, trees, traces,
pictures, (hyper-)graphs, and so on.
To capture all these various settings we start by introducing
an operation~$\bbM$ mapping a given set~$A$ of labels to the set~$\bbM A$
of all $A$-labelled objects. A~language in this context is then simply a subset $K \subseteq \bbM A$.
For instance, for languages of finite words we can define $\bbM A := A^+$.
To accommodate more complicated settings like trees, it will be convenient to work
not with plain sets but with \emph{many-sorted} ones.
For a given set~$\Xi$ of \emph{sorts,} an \emph{$\Xi$-sorted set} is a family
$A = (A_\xi)_{\xi \in \Xi}$ of plain sets.
Then $\bbM$~maps a $\Xi$-sorted set~$A$ of \emph{labels} to a $\Xi$-sorted set
$\bbM A = (\bbM_\xi A)_{\xi \in \Xi}$ of \emph{$A$-labelled objects.}
For instance, when working with infinite words it is convenient to use two sorts
$\Xi = \{1,\infty\}$ where sort~$1$ represents the `finite' elements and sort~$\infty$
the `infinite' ones. The operation~$\bbM$ maps $A = \langle A_1,A_\infty\rangle$ to
$\bbM A = \langle\bbM_1 A,\bbM_\infty A\rangle$ where
\begin{align*}
  \bbM_1 A := A_1^+
  \qtextq{and}
  \bbM_\infty A := A_1^+A_\infty \cup A_1^\omega\,.
\end{align*}
(So a finite word is a finite sequence of finite elements, while an infinite word
can either be a finite sequence of finite elements followed by a single infinite element,
or an infinite sequence of finite elements.)
Our intended applications consist in deriving characterisation results for various logics.
To be able to handle logics that are not closed under negation, it will turn out to be
necessary to be slightly more general and consider \emph{ordered many-sorted sets,}
that is, $\Xi$-sorted sets $A = (A_\xi)_{\xi \in \Xi}$ where each sort~$A_\xi$
is equipped with a partial order. Such sets form a category $\Pos^\Xi$ if we take
as morphisms the \emph{order-preserving $\Xi$-sorted functions,} that is,
a morphism $f : A \to B$ consists of a family $f = (f_\xi)_{\xi \in \Xi}$ of functions
where each component $f_\xi : A_\xi \to B_\xi$ is order-preserving.
We will frequently identify a sorted set $A = (A_\xi)_{\xi \in \Xi}$ with its
disjoint union $A = \bigcupdot_{\xi \in \Xi} A_\xi$.
Using this point of view, a morphism $f : A \to B$ corresponds to a sort-preserving
and order-preserving function between the corresponding disjoint unions.

Before continuing let us introduce a bit of terminology.
From this point on, we will use the terms `set' and `function' as a short-hand for
'ordered $\Xi$-sorted set' and `order-preserving $\Xi$-sorted function'.
If we mean any other kind of set or function, we will mention this explicitly.
We call a set $A \in \Pos^\Xi$ \emph{unordered} if its ordering is trivial,
i.e., any two distinct elements are incomparable.
For a property~$P$, we say that $A$~is \emph{sort-wise $P$} if each set~$A_\xi$ has property~$P$.
In particular, \emph{sort-wise finite} means that every~$A_\xi$ is finite.

Of course, the operation~$\bbM$ alone does not provide sufficient structure to build
a meaningful theory. Usually, the objects in a formal language are subject to
various composition operations, like concatenation of words, substitution for terms, etc..
To capture such operations we will employ the category-theoretical notion of a \emph{monad.}
Note that, in the cases of interest where $\bbM A$ is a set of $A$-labelled objects
of some kind, every function $f : A \to B$ induces an operation $\bbM f : \bbM A \to \bbM B$
which applies the function~$f$ to each label.
This turns~$\bbM$ into a functor $\Pos^\Xi \to \Pos^\Xi$.

There are two other ingredients we will need.
Firstly, the concatenation operation in question is often of the form
$\Flat : \bbM\bbM A \to \bbM A$, that is, it takes an $\bbM A$-labelled object $s \in \bbM\bbM A$
and assembles the appearing labels into a single large object.
We call $\Flat(s)$ the \emph{flattening} of~$s$.
Secondly, there is usually a \emph{singleton operation} $\sing : A \to \bbM A$ that takes a label
$a \in A$ and produces an object with a single position which is labelled by~$a$.
For instance, in the case of words $\Flat : (A^+)^+ \to A^+$ is simply the concatenation
operation and $\sing : A \to A^+$ produces $1$-letter words.
\begin{alignat*}{-1}
  \Flat(\langle w_0,\dots,w_n\rangle) &:= w_0\dots w_n\,,
  &&\quad\text{for } w_0,\dots,w_n \in A^+, \\
  \sing(a) &:= \langle a\rangle\,, &&\quad\text{for } a \in A\,.
\end{alignat*}
Usually, the flattening operation is associative, which makes the functor~$\bbM$ into a \emph{monad.}
\begin{defi}
A \emph{monad} consists of a functor $\bbM : \Pos^\Xi \to \Pos^\Xi$ that is equipped with
two natural transformations $\Flat : \bbM\circ\bbM \Rightarrow \bbM$ and
$\sing : \Id \Rightarrow \bbM$
(where $\Id$~is the identity functor) satisfying the following equations.
\begin{align*}
  \Flat \circ \sing = \id\,, \qquad
  \Flat \circ \bbM\sing = \id\,, \qquad
  \Flat \circ \Flat = \Flat \circ \bbM\Flat\,.
\end{align*}

\medskip
{\centering
\includegraphics{Abstract-1.mps}
%
%
%
%
\par}

\vskip-1.2em

\noindent
\markenddef
\end{defi}

In algebraic language theory one equips the sets $\bbM A$ with an algebraic structure
of some kind and then uses homomorphisms $\bbM A \to \frakB$ into some other algebra~$\frakB$
to describe languages $K \subseteq \bbM A$.
If $\bbM$~is a monad, there is a canonical way to define this algebraic structure\?:
we can equip a set~$A$ with a \emph{product} operation of the form $\pi : \bbM A \to A$.
For instance, for words this product takes the form $\pi : A^+ \to A$,
i.e., it multiplies a sequence of elements into a single element.
Hence, $\pi$~can be seen as a semigroup product of variable arity.
But note that not every operation $\pi : A^+ \to A$ is of the form
\begin{align*}
  \pi(\langle a_0,\dots,a_m\rangle) = a_0 \cdot a_1 \cdot\cdots\cdot a_m
\end{align*}
for some semigroup product ${}\cdot{} : A \times A \to A$.
If we want to exactly capture the notion of a semigroup,
we have to impose additional conditions on~$\pi$.
It turns out, there are two such conditions\?:
associativity requires that
\begin{align*}
  \pi(\pi(w_0),\dots,\pi(w_m)) = \pi(w_0\dots w_m)\,,
  \quad\text{for all } w_0,\dots,w_m \in A^+,
\end{align*}
and the fact that the product of a single element should return that element again
requires that
\begin{align*}
  \pi(\langle a\rangle) = a\,, \quad\text{for } a \in A\,.
\end{align*}
These two conditions can be phrased more concisely as
\begin{align*}
  \pi \circ \bbM\pi = \pi \circ \Flat
  \qtextq{and}
  \pi \circ \sing = \id\,.
\end{align*}
This leads us to the following definition.
\begin{defi}
Let $\bbM : \Pos^\Xi \to \Pos^\Xi$ be a monad.

(a)
An \emph{Eilenberg-Moore algebra} for~$\bbM$, or \emph{$\bbM$-algebra} for short,
is a pair $\frakA = \langle A,\pi\rangle$ consisting of a set~$A$ and a function
$\pi : \bbM A \to A$ satisfying

\vskip-\bigskipamount
\noindent
\begin{minipage}[t]{0.5\textwidth}
\vspace*{-1ex}%
\begin{align*}
  \pi \circ \bbM\pi &= \pi \circ \Flat\,, \\
  \pi \circ \sing &= \id\,.
\end{align*}
The first of these equations is called the \emph{associative law} for~$\pi$,
the second one the \emph{unit law.}\strut
\end{minipage}%
\begin{minipage}[t]{0.5\textwidth}
\vspace*{1ex}%
\centering
\includegraphics{Abstract-2.mps}
%
%
%
%
%
\end{minipage}

\smallskip
(b) A \emph{morphism} $\varphi : \frakA \to \frakB$ of $\bbM$-algebras is a function
$\varphi : A \to B$ commuting with the respective products in the sense that

\vskip-\bigskipamount
\noindent
\begin{minipage}[t]{0.5\textwidth}
\vspace*{-1ex}%
\begin{align*}
  \varphi \circ \pi = \pi \circ \bbM\varphi\,.
\end{align*}
\end{minipage}%
\begin{minipage}[t]{0.5\textwidth}
\vspace*{1ex}%
\centering
\includegraphics{Abstract-3.mps}
%
%
%
%
%
\end{minipage}

\smallskip
(c) We denote the category of all $\bbM$-algebras and their morphisms by $\Alg(\bbM)$.

(d) An algebra~$\frakA$ is \emph{finitary} if its universe~$A$ is sort-wise finite
and $\frakA$~is finitely generated, i.e., there exists a finite set $C \subseteq A$
such that every element $a \in A$ can be written as $a = \pi(s)$, for some $s \in \bbM C$.
\markenddef
\end{defi}

As a further example, let us take a look at the functor
\begin{align*}
  \bbM\langle A_1,A_\infty\rangle := \langle A_1^+,\ A_1^+A_\infty \cup A_1^\omega\rangle
\end{align*}
for infinite words.
In this case an $\bbM$-algebra has two product functions
\begin{align*}
  \pi_1 : A_1^+ \to A_1
  \qtextq{and}
  \pi_\infty : A_1^+A_\infty \cup A_1^\omega \to A_\infty\,.
\end{align*}
The laws of an $\bbM$-algebra ensure that $\pi_1$~corresponds to a semigroup product
$A_1 \times A_1 \to A_1$ and $\pi_\infty$~correspond to the additional products
$A_1 \times A_\infty \to A_\infty$ and $A_1^\omega \to A_\infty$ of an $\omega$-semigroup.
Hence, in this case $\bbM$-algebras are nothing but $\omega$-semigroups.

There is a natural way to turn a set of the form $\bbM A$ into an $\bbM$-algebra\?:
we can chose the function $\Flat : \bbM\bbM A \to \bbM A$ as the product.
It turns out that algebras of this form are exactly the \emph{free algebras.}
\begin{prop}\label{Prop: existence of free algebra}
For each ranked set~$A$, there exists a free $\bbM$-algebra over~$A$.
It has the form $\langle\bbM A,\Flat\rangle$.
\end{prop}
\begin{proof}
The fact that $\Flat : \bbM\bbM A \to \bbM A$
is the free $\bbM$-algebra is a standard result in category theory.
As the functor~$\bbM$ is a monad, it is left adjoint to the
forgetful functor $\Alg(\bbM) \to \Pos^\Xi$ which maps a $\bbM$-algebra~$\frakB$
to its universe~$B$ (see, e.g., Proposition~4.1.4 of~\cite{Borceux94b}).
Consequently, for every $\bbM$-algebra~$\frakB$ and every function $f : A \to B$,
there exists a unique morphism $\varphi : \bbM A \to \frakB$ such that $\varphi \circ \sing = f$.
\end{proof}

In order to obtain non-trivial results we have to put some mild restrictions
on the kind of monad~$\bbM$ we consider.
In the applications we have in mind, $\bbM$~is always a polynomial functor of the form
\begin{align*}
  \bbM A = \sum_{i < \lambda} A^{D_i},
\end{align*}
for some cardinal~$\lambda$ and unordered sets $D_i \in \Set^\Xi$. For instance,
for the word functor $\bbM A = A^+$, we can take $\lambda = \aleph_0$
and $D_i = \{0,\dots,i\}$.
Similarly, if we consider languages of trees, we can fix an enumeration
$(t_i)_{i<\lambda}$ of all unlabelled trees and choose for~$D_i$ the
set of vertices of~$t_i$.

For the results in this article, we do not need to assume that $\bbM$~is polynomial.
A~few weaker properties suffice. To state these, we have to introduce a bit of terminology.
\begin{defi}
Let $\bbM : \Pos^\Xi \to \Pos^\Xi$ be a functor.

(a) The \emph{lift} of a relation $R \subseteq A \times B$ is the relation
$R^\bbM \subseteq \bbM A \times \bbM B$ consisting of all pairs $\langle s,t\rangle$
such that
\begin{align*}
  s = \bbM p(u) \qtextq{and} t = \bbM q(u)\,,
  \quad\text{for some } u \in \bbM R\,,
\end{align*}
where $p : A \times B \to A$ and $q : A \times B \to B$ are the two projections.

(b)
We say that $\bbM$ \emph{uses the standard ordering} if the ordering of~$\bbM A$
is the lift~$\leq^\bbM$ of the ordering~$\leq$ of~$A$.

(c)
We say that $\bbM$~\emph{preserves injectivity/surjectivity/bijectivity} if
it is the case that $\bbM$ maps injective/surjective/bijective functions to functions of the same kind.

(d)
$\bbM$~\emph{preserves preimages} if, for every function $f : A \to B$
and every subset $P \subseteq B$, we have
\begin{align*}
  \bbM(f^{-1}[P]) = (\bbM f)^{-1}[\bbM P]\,.
\end{align*}
\upqed
\markenddef
\end{defi}
For instance, the standard ordering for the word functor $\bbM A = A^+$ is
\begin{align*}
  \langle a_0,\dots,a_m\rangle \leq \langle b_0,\dots,b_n\rangle
  \quad\iff\quad
  m = n \text{ and } a_i \leq b_i \text{ for all } i \leq m\,.
\end{align*}

For our framework we require the following properties of~$\bbM$,
which are clearly shared by every polynomial functor.
\begin{VarLem}[Convention]
In the following we will always tacitly assume that\/ $\bbM : \Pos^\Xi \to \Pos^\Xi$
is a monad which preserves injectivity, surjectivity, bijectivity, and preimages, and
that it uses the standard ordering.
\end{VarLem}

An example of a monad that does not fit into this framework would be the functor~$\bbD$
mapping a set~$A$ to its set of downward closed sets, ordered by inclusion.
The multiplication of this monad maps a set of sets to its union,
and the singleton function maps an element $a \in A$ to the set
$\set{ b \in A }{ b \leq a }$.
For a function $f : A \to B$ we set
\begin{align*}
  \bbD f(X) = \set{ b \in B }{ b \leq f(a) \text{ for some } a \in X }\,.
\end{align*}
This monad preserves surjectivity, but neither injectivity, bijectivity, nor preimages,
and it does not use the standard ordering.

Let us derive a first consequence of our assumptions.
A~frequent problem we will have to deal with is the fact that in $\Pos^\Xi$
not every surjective morphism has a right inverse.
Therefore, we will sometimes be forced to make a detour through the category $\Set^\Xi$
by ignoring the order of the sets involved.
For simplicity, we will treat $\Set^\Xi$ as a full subcategory of $\Pos^\Xi$
via the following embedding.
\begin{defi}\label{Def: bbV and iota}
(a) Let $\bbV : \Pos^\Xi \to \Pos^\Xi$ be the functor mapping a set~$A$ with order~$\leq$
to the same set, but with the trivial order~$=$, and let
$\iota : \bbV \Rightarrow \Id$ the natural transformation induced
by the identity maps.

(b) A functor $\bbM : \Pos^\Xi \to \Pos^\Xi$ is \emph{order agnostic} if
there exists a natural isomorphism $\delta : \bbM\circ\bbV \Rightarrow \bbV\circ\bbM$
satisfying

\vskip-1ex
\noindent
\begin{minipage}[t]{0.5\textwidth}
\vspace*{0pt}%
\begin{align*}
  \bbM\iota &= \iota \circ \delta \\
\prefixtext{and}
  \bbV\sing &= \delta \circ \sing\,.
\end{align*}
\end{minipage}%
\begin{minipage}[t]{0.5\textwidth}
\centering
\vspace*{0pt}%
\includegraphics{Abstract-4.mps}
%
%
%
%
\end{minipage}

\vskip-1em
\markenddef
\end{defi}
Intuitively, being order agnostic means that $\bbM$~does not make essential use
of the ordering of a set~$A$ when producing~$\bbM A$. The ordering of~$A$ has no
influence on which elements~$\bbM A$ contains, only on the ordering between them.

\begin{lem}
If\/ $\bbM$~satisfies the above assumption, it is order agnostic.
\end{lem}
\begin{proof}
We start by showing that each set of the form $\bbM\bbV A$ has the trivial order.
Hence, suppose that $s,t \in \bbM\bbV A$ with $s \leq t$.
As $\bbM$~uses the standard ordering, we can find some $u \in \bbM\Delta$
with $s = \bbM p(u)$ and $t = \bbM q(u)$, where $\Delta \subseteq A \times A$
is the ordering of~$\bbV A$ and $p,q : A \times A \to A$ are the two projections.
Note that $\Delta = \set{ \langle a,a\rangle }{ a \in A }$ is the diagonal.
Consequently, we have $p(d) = q(d)$, for all $d \in \Delta$, which implies that
$s = \bbM p(u) = \bbM q(u) = t$.

It follows that $\bbV\bbM\bbV A = \bbM\bbV A$ and the respective identity maps
provide morphisms $i : \bbM\bbV A \to \bbV\bbM\bbV A$ and $j : \bbV A \to \bbV\bbV A$
that are inverse to the functions $\iota : \bbV\bbM\bbV A \to \bbM\bbV A$
and $\iota : \bbV\bbV A \to \bbV A$.
We claim that the morphism $\delta := \bbV\bbM\iota \circ i : \bbM\bbV A \to \bbV\bbM A$
is the desired natural transformation.
First, note that $\delta$~is bijective, as $i$ and~$\iota$ are bijective and
both $\bbV$~and~$\bbM$ preserve bijectivity.
Since the domain~$\bbM\bbV A$ and the codomain~$\bbV\bbM A$ both use the trivial order,
$\delta$~therefore has an inverse.
To conclude the proof it is hence sufficient to show that the following diagram commutes.
\begin{center}
\includegraphics{Abstract-5.mps}
\end{center}
Since $\bbV\iota = \iota$ it follows that
\begin{alignat*}{-1}
  \iota \circ \delta
    &= \rlap{\iota \circ \bbV\bbM\iota \circ i
     = \bbM\iota \circ \iota \circ i
     = \bbM\iota} \\[1ex]
\prefixtext{and}
  \delta \circ \sing
    &= \bbV\bbM\iota \circ i \circ \sing
   &&= \bbV\bbM\iota \circ i \circ \sing \circ \iota \circ j \\
  &&&= \bbV\bbM\iota \circ i \circ \iota \circ \bbV\sing \circ j \\
  &&&= \bbV\bbM\iota \circ \bbV\sing \circ j \\
  &&&= \bbV(\sing \circ \iota) \circ j \\
  &&&= \bbV\sing \circ (\bbV\iota \circ j)
     = \bbV\sing \circ (\iota \circ j)
     = \bbV\sing\,.
\end{alignat*}
\upqed
\end{proof}

\section{Congruences and quotients}   
\label{Sect:quotients}

We start by developing a theory of congruences for $\bbM$-algebras.
Most of the arguments in this section are quite standard,
but we did not find them worked out for Eilenberg--Moore algebras
anywhere in the literature.
We begin by looking at quotients of ordered sets. Then we will turn to $\bbM$-algebras.
\begin{defi}
Let $A$~be an ordered set and ${\sqsubseteq} \subseteq A \times A$ a preorder with
${\leq} \subseteq {\sqsubseteq}$.

(a) The \emph{kernel} of a function $f : A \to B$ is the relation
\begin{align*}
  \ker f := \set{ \langle a,a'\rangle \in A \times A }{ f(a) \leq f(a') }\,.
\end{align*}

(b) For $a \in A$ and $X \subseteq A$, we set
\begin{align*}
  \Aboveseg a := \set{ b \in A }{ b \geq a }
  \qtextq{and}
  \Aboveseg X := \bigcup_{a \in X} \Aboveseg a\,.
\end{align*}

(c) The set of \emph{$\sqsubseteq$-classes} is
\begin{align*}
  A/{\sqsubseteq} := \set{ [a]_{\sqsubseteq} }{ a \in A }
  \qtextq{where}
  [a]_{\sqsubseteq} := \set{ b \in A }{ b \sqsubseteq a \text{ and } a \sqsubseteq b }\,.
\end{align*}
We equip it with the ordering
\begin{align*}
  [a]_{\sqsubseteq} \leq [b]_{\sqsubseteq} \quad\defiff\quad a \sqsubseteq b\,.
\end{align*}

(d) The \emph{quotient map} $q : A \to A/{\sqsubseteq}$ maps $a \in A$ to~$[a]_\sqsubseteq$.

(e) We say that $\sqsubseteq$~has \emph{finitary index} if
the quotient $A/{\sqsubseteq}$ is sort-wise finite.
\markenddef
\end{defi}

A very useful tool to construct quotients is the following lemma from universal algebra.
\begin{lem}[Factorisation Lemma]
Let $f : A \to B$ and $g : A \to C$ be functions and assume that $f$~is surjective.
Then $g = h \circ f$, for some $h : B \to C$, if and only if
\begin{align*}
  \ker f \subseteq \ker g\,.
\end{align*}
Moreover, the function $h$~is unique, if it exists.
\end{lem}
\begin{proof}
The uniqueness of~$h$ follows from the surjectivity of~$f$,
since surjective functions are epimorphisms\?:
$h \circ f = g = h' \circ f$ implies $h = h'$.
Hence, it remains to consider existence.

$(\Rightarrow)$ If $g = h \circ f$, then
\begin{align*}
  f(a) \leq f(b)
  \qtextq{implies}
  g(a) = h(f(a)) \leq h(f(b)) = g(b)\,.
\end{align*}

$(\Leftarrow)$
Suppose that $\ker f \subseteq \ker g$.
As $f$~is surjective, it has a right inverse~$r$ (in~$\Set^\Xi$, $r$~might not be monotone).
We claim that $h := g \circ r$ is the desired function.

For monotonicity, suppose that $a \leq b$ in~$B$.
Then
\begin{align*}
  f(r(a)) = a \leq b = f(r(b))
  \qtextq{implies}
  \langle r(a),r(b)\rangle \in \ker f \subseteq \ker g\,.
\end{align*}
Consequently,
\begin{align*}
  h(a) = g(r(a)) \leq g(r(b)) = h(b)\,.
\end{align*}

To show that $g = h \circ f$, set $e := r \circ f$.
For $a \in A$, it follows that
\begin{align*}
  f(e(a)) = (f \circ r \circ f)(a) =  f(a)\,.
\end{align*}
Hence, $\langle a,e(a)\rangle, \langle e(a),a\rangle \in \ker f \subseteq \ker g$,
which implies that $g(a) = g(e(a))$. Thus
$g = g \circ e = g \circ r \circ f = h \circ f$.
\end{proof}

In order to lift the statement of the Factorisation Lemma from functions to morphisms,
it is sufficient to prove that, if $f$~and~$g$ are morphisms of $\bbM$-algebras, so is~$h$.
\begin{lem}\label{Lem: factorisation is morphism}
Let $f : \frakA \to \frakB$ and $g : \frakA \to \frakC$ be morphisms of\/ $\bbM$-algebras
and $h : B \to C$ a function such that $g = h \circ f$.
If $f$~is surjective, then $h$~is also a morphism of\/ $\bbM$-algebras.
\end{lem}
\begin{proof}
Note that
\begin{align*}
  h \circ \pi \circ \bbM f
  = h \circ f \circ \pi
  = g \circ \pi
  = \pi \circ \bbM g
  = \pi \circ \bbM h \circ \bbM f\,.
\end{align*}
Since $f$~is surjective, so is~$\bbM f$.
Therefore, the above equation implies that $h \circ \pi = \pi \circ \bbM h$,
i.e., that $h$~is a morphism of $\bbM$-algebras.
\end{proof}

In $\Set^\Xi$ there also exists a dual to the statement of the Factorisation Lemma,
but in $\Pos^\Xi$~this version only holds in special cases as, in general,
surjective functions do not have right inverses.
Let us record the following version for algebras of the form~$\bbM X$ where $X$~is unordered.
In category-theoretical terminology
it states that such algebras are \emph{projective}.
\begin{lem}\label{Lem: FX projective}
Let $\varphi : \bbM X \to \frakB$ and $\psi : \frakA \to \frakB$ be morphisms of\/ $\bbM$-algebras
where $X$~is an unordered set.
If $\psi$~is surjective, there exists some morphism $\hat\varphi : \bbM X \to \frakA$
such that $\varphi = \psi \circ \hat\varphi$.

\medskip
\centering
\includegraphics{Abstract-6.mps}
%
%
%
%
\end{lem}
\begin{proof}
If $\psi$~is surjective, we can pick, for every $x \in X$ some element
$f(x) \in \psi^{-1}(\varphi(\sing(x)))$.
This defines a function $f : X \to A$ with $\psi \circ f = \varphi \circ \sing$
(which is trivially monotone as $X$~is unordered).
As $\bbM X$ is freely generated by the range of~$\sing$,
we can extend~$f$ to a unique morphism $\hat\varphi : \bbM X \to \frakA$ with
$\hat\varphi \circ \sing = f$. It follows that
\begin{align*}
  \psi \circ \hat\varphi \circ \sing
  = \psi \circ f
  = \varphi \circ \sing\,.
\end{align*}
As the range of~$\sing$ generates~$\bbM X$, this implies that $\psi \circ \hat\varphi = \varphi$.
\end{proof}

Next, let us define quotients for algebras instead of sets.
\begin{defi}
Let $\frakA$~be an $\bbM$-algebra and ${\sqsubseteq} \subseteq A \times A$
a preorder with ${\leq} \subseteq {\sqsubseteq}$.

(a) For $s,t \in \bbM A$, we set
\begin{align*}
  s \sqsubseteq_\bbM t \quad\defiff\quad \bbM q(s) \leq \bbM q(t)\,,
\end{align*}
where $q : A \to A/{\sqsubseteq}$ is the quotient map.

(b) Let $\frakA$~be an $\bbM$-algebra.
The preorder~$\sqsubseteq$ is a \emph{congruence ordering} on~$\frakA$ if
\begin{align*}
  s \sqsubseteq_\bbM t \qtextq{implies} \pi(s) \sqsubseteq \pi(t)\,.
\end{align*}

(c)
If $\sqsubseteq$~is a congruence ordering on~$\frakA$,
we define the \emph{quotient}~$\frakA/{\sqsubseteq}$ as the algebra with universe~$A/{\sqsubseteq}$
where the product $\pi : \bbM(A/{\sqsubseteq}) \to A/{\sqsubseteq}$ is the unique function such that
\begin{align*}
  \pi \circ \bbM q = q \circ \pi\,,
\end{align*}
where $q : A \to A/{\sqsubseteq}$ denotes the quotient map.
\markenddef
\end{defi}
\begin{Rem}
It is straightforward to show that ${\sqsubseteq^\bbM} \subseteq {\sqsubseteq_\bbM}$.
For most monads~$\bbM$, these two relations are actually equal,
but our assumptions on~$\bbM$ are not quite strong enough to prove this in general.
\end{Rem}

\begin{prop}\label{Prop: quotient map is morphism}
Let\/ $\sqsubseteq$~be a congruence ordering on an\/ $\bbM$-algebra\/~$\frakA$.
The quotient\/ $\frakA/{\sqsubseteq}$ is a well-defined\/ $\bbM$-algebra and
the quotient map $q : \frakA \to \frakA/{\sqsubseteq}$ is a morphism of\/ $\bbM$-algebras.
\end{prop}
\begin{proof}
We have to check several properties.

(a) To see that $q$~is monotone, note that
\begin{align*}
  a \leq b \quad\Rightarrow\quad a \sqsubseteq b \quad\Rightarrow\quad q(a) \leq q(b)\,.
\end{align*}

(b) To show that the product of $\frakA/{\sqsubseteq}$ is well-defined and monotone
we apply the Factorisation Lemma. Note that, for $s,t \in \bbM A$,
\begin{align*}
  \bbM q(s) \leq \bbM q(t)
  \quad\Rightarrow\quad
  s \sqsubseteq_\bbM t
  \quad\Rightarrow\quad
  \pi(s) \sqsubseteq \pi(t)
  \quad\Rightarrow\quad
  q(\pi(s)) \leq q(\pi(t))\,,
\end{align*}
where the second step follows from the fact that $\sqsubseteq$~is a congruence ordering.
Thus, $\ker \bbM q \subseteq \ker {(q \circ \pi)}$ and
the Factorisation Lemma implies that there is a unique function
$\pi : \bbM(A/{\sqsubseteq}) \to A/{\sqsubseteq}$ with
$\pi \circ \bbM q = q \circ \pi$.

(c) It remains to check the two axioms of an $\bbM$-algebra.
For the unit law, note that $\sing : \Id \Rightarrow \bbM$ is a natural transformation.
Hence,
\begin{align*}
  \pi \circ \sing \circ q = \pi \circ \bbM q \circ \sing
                          = q \circ \pi \circ \sing
                          = q\,.
\end{align*}
As $q$~is surjective, this implies that $\pi \circ \sing = \id$.
For the associative law, note that
\begin{align*}
  \pi \circ \bbM\pi \circ \bbM\bbM q
  &= \pi \circ \bbM q \circ \bbM\pi \\
  &= q \circ \pi \circ \bbM\pi \\
  &= q \circ \pi \circ \Flat \\
  &= \pi \circ \bbM q \circ \Flat
   = \pi \circ \Flat \circ \bbM\bbM q\,.
\end{align*}
Hence, surjectivity of $\bbM\bbM q$ implies that $\pi \circ \bbM\pi= \pi \circ \Flat$.
\end{proof}

As usual, we have defined our notion of a congruence such that congruences correspond to
kernels of morphisms.
We will establish this correspondence in Proposition~\ref{Prop: characterisations of congruences}
below. But before doing so, let us take a closer look at the auxiliary relation $\sqsubseteq_\bbM$.
\begin{lem}\label{Lem: standard ordering extends to all preorders}
Let $\sqsubseteq$~be a preorder on~$A$ with ${\leq} \subseteq {\sqsubseteq}$. Then
\begin{align*}
  \bbM{\sqsubseteq} = \langle\bbM p_0,\bbM p_1\rangle^{-1}[{\sqsubseteq_\bbM}]\,,
\end{align*}
where $p_0,p_1 : A \times A \to A$ are the two projections.
\end{lem}
\begin{proof}
Let $p'_0,p'_1 : A/{\sqsubseteq} \times A/{\sqsubseteq} \to A/{\sqsubseteq}$ be the two projections,
$q : A \to A/{\sqsubseteq}$ the quotient map, and
let $\leq$~be the ordering on $A/{\sqsubseteq}$.
As $\bbM$~uses the standard ordering it follows that $\bbM(A/{\sqsubseteq})$
is ordered by the lift~$\leq^\bbM$.
We consider the following diagram
\begin{center}
\includegraphics{Abstract-7.mps}
\end{center}
where the vertical arrows denote the respective inclusion maps.
Let us first explain why this diagram commutes.
Since in each square the vertical maps are inclusions and the top map is a
restriction of the bottom one,
it is sufficient to show that the maps in the upper trapezium
map the first of the given subsets to the second one.
That is, we have to show that
\begin{align*}
  \bbM(q \times q)[\bbM{\sqsubseteq}]           &\subseteq \bbM{\leq}\,, \\
  \langle\bbM p'_0,\bbM p'_1\rangle[\bbM{\leq}] &\subseteq {\leq}^\bbM, \\
  (\bbM q \times \bbM q)[{\sqsubseteq}_\bbM]    &\subseteq {\leq}^\bbM.
\end{align*}
The first inclusion follows from the fact that $q \times q$ maps~$\sqsubseteq$ to~$\leq$
by simply applying the functor~$\bbM$\?;
the second one holds by definition of the lift~$\leq^\bbM$\?; and
the last inclusion follows immediately from the definition of~$\sqsubseteq_\bbM$.

To conclude the proof, note that we have even the stronger statements
\begin{align*}
  \bbM{\sqsubseteq}  &= \bbM(q \times q)^{-1}[\bbM{\leq}]\,, \\
  \bbM{\leq}         &= \langle\bbM p'_0,\bbM p'_1\rangle^{-1}[{\leq}^\bbM]\,, \\
  {\sqsubseteq}_\bbM &= (\bbM q \times \bbM q)^{-1}[{\leq}^\bbM]\,.
\end{align*}
The first equation follows from the fact that $\bbM$~preserves preimages,
while the second one holds by definition of the lift~$\leq^\bbM$.
To check the third equation, suppose that $\bbM q(s) \leq^\bbM \bbM q(t)$.
Considering the sets
\begin{align*}
  X := \set{ u \in \bbM(A/{\sqsubseteq}) }{ \bbM q(s) \leq^\bbM u }
  \qtextq{and}
  Y := \set{ v \in \bbM A }{ s \sqsubseteq_\bbM v }\,,
\end{align*}
preservation of preimages implies that $Y = (\bbM q)^{-1}[X]$.
In particular, $t \in Y$ and $s \sqsubseteq_\bbM t$.

It now follows that
\begin{align*}
  \bbM{\sqsubseteq}
  &= \bbM(q \times q)^{-1}\bigl[\langle\bbM p'_0,\bbM p'_1\rangle^{-1}[{\leq}^\bbM]\bigr] \\
  &= \langle\bbM p_0,\bbM p_1\rangle^{-1}\bigl[(\bbM q \times \bbM q)^{-1}[{\leq}^\bbM]\bigr]
   = \langle\bbM p_0,\bbM p_1\rangle^{-1}[{\sqsubseteq}_\bbM]\,.
\end{align*}
\upqed
\end{proof}

We obtain the following characterisation of congruence orderings.
\begin{prop}\label{Prop: characterisations of congruences}
Let $\frakA$ be an $\bbM$-algebra and ${\sqsubseteq} \subseteq A \times A$ a preorder that
contains the ordering of~$A$. Let $p_0,p_1 : A \times A \to A$ be the two projections.
The following conditions are equivalent.
\begin{enum1}
\item $\sqsubseteq$~is a congruence ordering on~$\frakA$.
\item ${\sqsubseteq} = \ker \varphi$, for some morphism $\varphi : \frakA \to \frakB$.
\item $u \in \bbM{\sqsubseteq}$\quad implies\quad $\pi(\bbM p_0(u)) \sqsubseteq \pi(\bbM p_1(u))$\,.
\item $\sqsubseteq$~induces a subalgebra of $\frakA \times \frakA$.
\end{enum1}
\end{prop}
\begin{proof}
(1)~$\Rightarrow$~(2)
The quotient map $q : \frakA \to \frakA/{\sqsubseteq}$ has kernel~$\sqsubseteq$.

(2)~$\Rightarrow$~(1)
Clearly, $a \leq b$ implies $\varphi(a) \leq \varphi(b)$.
Thus, ${\leq} \subseteq {\sqsubseteq}$.
For the other condition, consider two elements $s,t \in \bbM A$ with $s \sqsubseteq_\bbM t$.
By definition, this means that $\bbM q(s) \leq \bbM q(t)$ where $q : \frakA \to \frakA/{\sqsubseteq}$
is the quotient map.
As $q$~is surjective, we can use the Factorisation Lemma to find a function
$f : A/{\sqsubseteq} \to B$ with $\varphi = f \circ q$.
By monotonicity of $f$~and~$\pi$, it follows that
\begin{align*}
  \varphi(\pi(s))
   = \pi(\bbM\varphi(s))
  &= \pi(\bbM f(\bbM q(s))) \\
  &\leq \pi(\bbM f(\bbM q(t)))
   = \pi(\bbM\varphi(t))
   = \varphi(\pi(t))\,.
\end{align*}
Consequently, $\langle \pi(s),\pi(t)\rangle \in \ker \varphi = {\sqsubseteq}$.

(4)~$\Rightarrow$~(3)
Let $u \in \bbM{\sqsubseteq}$. Then
\begin{align*}
  \bigl\langle \pi(\bbM p_0(u)),\ \pi(\bbM p_1(u))\bigr\rangle
  = \bigl\langle p_0(\pi(u)),\ p_1(\pi(u))\bigr\rangle
  = \pi(u) \in {\sqsubseteq}\,.
\end{align*}
Hence, $\pi(\bbM p_0(u)) \sqsubseteq \pi(\bbM p_1(u))$.

(3)~$\Rightarrow$~(4)
Let $u \in \bbM{\sqsubseteq}$. Then $\pi(\bbM p_0(u)) \sqsubseteq \pi(\bbM p_1(u))$ implies that
\begin{align*}
  \pi(u) = \bigl\langle p_0(\pi(u)),\ p_1(\pi(u))\bigr\rangle
         = \bigl\langle \pi(\bbM p_0(u)),\ \pi(\bbM p_1(u))\bigr\rangle \in {\sqsubseteq}\,.
\end{align*}

(3)~$\Rightarrow$~(1) To show that ${\sqsubseteq}$~is a congruence ordering,
suppose that $s \sqsubseteq_\bbM t$.
By Lemma~\ref{Lem: standard ordering extends to all preorders}, there is some
$u \in \bbM{\sqsubseteq}$ with $s = \bbM p_0(u)$ and $t = \bbM p_1(u)$.
Hence it follows by~(3) that
\begin{align*}
  \pi(s) = \pi(\bbM p_0(u)) \sqsubseteq \pi(\bbM p_1(u)) = \pi(t)\,.
\end{align*}

(1)~$\Rightarrow$~(3)
Given $u \in \bbM{\sqsubseteq}$,
Lemma~\ref{Lem: standard ordering extends to all preorders} implies that
$\bbM p_0(u) \sqsubseteq_\bbM \bbM p_1(u)$.
Hence, $\bbM(q \circ p_0)(u) \leq \bbM(q \circ p_1)(u)$,
where $q : \frakA \to \frakA/{\sqsubseteq}$ is the quotient map.
As the product~$\pi$ is monotone, it follows that
\begin{align*}
  \pi(\bbM(q \circ p_0)(u)) \leq \pi(\bbM(q \circ p_1)(u))\,.
\end{align*}
Hence,
\begin{align*}
  q(\pi(\bbM p_0(u)))
  = \pi(\bbM(q \circ p_0)(u))
  \leq \pi(\bbM(q \circ p_1)(u))
  = q(\pi(\bbM p_1(u)))\,,
\end{align*}
which implies that $\pi(\bbM p_0(u)) \sqsubseteq \pi(\bbM p_1(u))$.
\end{proof}

\section{Languages and syntactic algebras}   
\label{Sect: synt. alg.}

Our main point of interest is to determine which
sets $K \subseteq \bbM_\xi\Sigma$ are definable in a given logic.
We start by collecting the needed notions from language theory.

\begin{defi}
(a) An \emph{alphabet} is a finite \emph{unordered} set $\Sigma \in \Pos^\Xi$.
We denote by $\Alp$ the category of all alphabets with functions
as morphisms.

(b) A \emph{language} over the alphabet~$\Sigma$ is a subset
$K \subseteq \bbM_\xi\Sigma$, for some sort~$\xi$.

(c)
A \emph{family of languages} is a function~$\calK$ mapping each alphabet~$\Sigma$
to a class $\calK[\Sigma]$ of languages over~$\Sigma$.

(d)
A function $f : \bbM\Sigma \to A$ \emph{recognises}
a language $K \subseteq \bbM_\xi\Sigma$ if
$K = f^{-1}[P]$, for some \emph{upwards closed} set $P \subseteq A_\xi$.

(e)
Let $f : \Sigma \to \Gamma$ be a morphism of $\Alp$.
We call a morphism of the form $\bbM f : \bbM\Sigma \to \bbM\Gamma$ a~\emph{relabelling}
and, for a language $K \subseteq \bbM_\xi\Gamma$, we call the set
\begin{align*}
  (\bbM f)^{-1}[K] := \set{ s \in \bbM_\xi\Sigma }{ \bbM f(s) \in K }
\end{align*}
an \emph{inverse relabelling} of~$K$.
\markenddef
\end{defi}

Note that we always assume alphabets to be unordered.
This is required for the variety theorem in the next section.
But sometimes it is useful to also work with languages over ordered alphabets.
We do so by simply forgetting the order.
This leads to the following extension of the notion of a family of languages.
\begin{defi}
Let $\calK$~be a family of languages. For a finite \emph{ordered} set~$C$, we define
\begin{align*}
  \calK[C] := \bigset{ \bbM\iota[K] }{ K \in \calK[\bbV C] }\,,
\end{align*}
where $\bbV$~and~$\iota$ are the operations from Definition~\ref{Def: bbV and iota}.
\markenddef
\end{defi}

One of our main tools will be the following relation associated with a language.
\begin{defi}
Let $\frakA$~be an $\bbM$-algebra.

(a) A \emph{context} is an element of $\bbM(A + \Box)$,
where $\Box$~is considered as some special symbol of an arbitrary, but fixed sort~$\zeta$.
For a context $p \in \bbM_\xi(A + \Box)$ and an element $a \in A_\zeta$, we define
\begin{align*}
  p[a] := \sigma_a(p) \in A_\xi
\end{align*}
where $\sigma_a : \bbM(A + \Box) \to \frakA$ is the unique morphism that extends
the function $s_a : A + \Box \to A$ given by
\begin{align*}
  s_a(\Box) := a
  \qtextq{and}
  s_a(c) := c\,,
  \quad\text{for } c \in A\,.
\end{align*}
In the case where $\frakA = \bbM\Sigma$ is a free $\bbM$-algebra,
we will also consider elements $p \in \bbM(\Sigma + \Box)$ as contexts,
by identifying them with their image under $\bbM(\sing + 1)$
(the function $\sing + 1$ maps $a \in A$ to $\sing(a)$ and $\Box$ to~$\Box$).

(b) The \emph{composition} of two contexts $p,q \in \bbM(A + \Box)$ is the context
\begin{align*}
  pq := \hat p[q] \in \bbM(A + \Box)\,,
\end{align*}
where $\hat p := \bbM(\sing + 1)(p)$ and $\hat p[q]$ is evaluated in the $\bbM$-algebra
$\bbM(A + \Box)$.

(c) A \emph{derivative} of a subset $K \subseteq A_\xi$ is a set of the form
\begin{align*}
  p^{-1}[K] := \set{ a \in A_\zeta }{ p[a] \in K }\,,
  \quad\text{where } p \in \bbM_\xi(A + \Box) \text{ is a context.}
\end{align*}

(d) The \emph{syntactic congruence} of an upwards closed set $K \subseteq A_\xi$
is the relation
\begin{align*}
  a \preceq_K b \quad\defiff\quad
  (p[a] \in K \Rightarrow p[b] \in K)\,,\quad \text{for all } p \in \bbM_\xi(A+\Box)\,,
\end{align*}
for $a,b \in A$.

(e) We call the quotient
$\Syn(K) := \frakA/{\preceq_K}$ the \emph{syntactic algebra} of~$K$ and the quotient map
$\syn_K : \frakA \to \frakA/{\preceq_K}$ the \emph{syntactic morphism} of~$K$.

(f) We say that a language~$K$ \emph{has a syntactic algebra} if
$\preceq_K$~is a congruence ordering with finitary index.
\markenddef
\end{defi}
Note that, in general, the syntactic congruence does not need to be a congruence ordering,
the syntactic algebra not an $\bbM$-algebra, and the syntactic morphism not a morphism of
$\bbM$-algebras, but we will mainly be interested in the case where they are.
Hence the terminology.

\begin{lem}\label{Lem: applying context is monotonous}
Let\/ $\frakA$~be an\/ $\bbM$-algebra, $K \subseteq A_\xi$ upwards closed, $a,b \in A$, and
$p \in \bbM(A + \Box)$.
\begin{enuma}
\item $a \leq b \qtextq{implies} p[a] \leq p[b]\,.$
\item $a \leq b \qtextq{implies} a \preceq_K b\,.$
\item $a \preceq_K b \qtextq{implies} p[a] \preceq_K p[b] \qtextq{and} a \preceq_{p^{-1}[K]} b\,.$
\end{enuma}
\end{lem}
\begin{proof}
(a)
Let $g : A + \Box \to A \times A$ be the function where
\begin{align*}
  g(\Box) := \langle a,b\rangle
  \qtextq{and}
  g(c) = \langle c,c\rangle\,, \quad\text{for } c \in A\,,
\end{align*}
let $q,q' : A \times A \to A$ be the two projections, and set $u := \bbM g(p)$.
Then $u \in \bbM{\leq}$ (where $\leq$~is the ordering of~$A$) and
\begin{align*}
  p[a] = \pi(\bbM q(u))
  \qtextq{and}
  p[b] = \pi(\bbM q'(u))\,.
\end{align*}
Since $\bbM$~uses the standard ordering, this implies that $p[a] \leq p[b]$.

(b) Suppose that $a \leq b$ and let $p \in \bbM(A + \Box)$ be a context.
By~(a), we have $p[a] \leq p[b]$.
Consequently, $p[a] \in K$ implies $p[b] \in K$.

(c)
Suppose that $a \preceq_K b$. To show that $p[a] \preceq_K p[b]$, consider
a context~$q$ with $q[p[a]] \in K$.
Then
$a \preceq_K b$
and
$q[p[a]] = (qp)[a] \in K$
implies that $q[p[b]] = (qp)[b] \in K$.

To show that $a \preceq_{p^{-1}[K]} b$, consider a context~$q$ with $q[a] \in p^{-1}[K]$.
Then
$a \preceq_K b$
and
$(pq)[a] = p[q[a]] \in K$
implies that $(pq)[b] = p[q[b]] \in K$. Thus $q[b] \in p^{-1}[K]$.
\end{proof}

One consequence of this lemma is that
the quotient $A/{\preceq_K}$ does exist at least as a set.
\begin{cor}
Let\/ $\frakA$~be an\/ $\bbM$-algebra and $K \subseteq A_\xi$ upwards closed.
Then\/ $\preceq_K$~is a preorder with\/ ${\leq} \subseteq {\preceq_K}$.
\end{cor}
\begin{proof}
Reflexivity and transitivity of $\preceq_K$ follow immediately from the definition.
The fact that $\preceq_K$ contains~$\leq$ is part~(a) of the preceding lemma.
\end{proof}

\begin{lem}\label{Lem: morphisms recognising K}
A morphism $\varphi : \bbM\Sigma \to \frakA$ of\/ $\bbM$-algebras recognises a
language $K \subseteq \bbM_\xi\Sigma$ if, and only if,\/ $\ker \varphi \subseteq {\preceq_K}$.
\end{lem}
\begin{proof}
$(\Leftarrow)$
We claim that $K = \varphi^{-1}[P]$ where $P := \Aboveseg\varphi[K]$.
Clearly, $\varphi(t) \in P$, for all $t \in K$. Conversely,
\begin{alignat*}{-1}
  \varphi(t) \in P
  &\quad\Rightarrow\quad
  \varphi(s) \leq \varphi(t)\,, &&\quad\text{for some } s \in K\,, \\
  &\quad\Rightarrow\quad
  s \preceq_K t\,, &&\quad\text{for some } s \in K\,, \\
  &\quad\Rightarrow\quad
  t \in K\,.
\end{alignat*}

$(\Rightarrow)$
Suppose that $K = \varphi^{-1}[P]$, for an upwards closed set~$P$,
and let $\varphi(s) \leq \varphi(t)$.
To show that $s \preceq_K t$, consider a context $p \in \bbM(\Sigma+\Box)$
with $p[s] \in K$.
Set
\begin{align*}
  \hat p := \bbM(\varphi \circ \sing + 1)(p) \in \bbM(A + \Box)\,.
\end{align*}
Then $\hat p[\varphi(s)] = \varphi(p[s]) \in P$.
According to Lemma~\ref{Lem: applying context is monotonous}\,(a),
we further have
$\hat p[\varphi(s)] \leq \hat p[\varphi(t)]$.
Together, it follows that $\hat p[\varphi(t)] = \varphi(p[t]) \in P$.
Hence, $p[t] \in K$.
\end{proof}

A~noteworthy consequence of this lemma is that the syntactic morphism of a language~$K$ is
the terminal object in the category of all morphisms recognising~$K$.
\begin{thm}\label{Thm: Syn(K) terminal}
Let $K \subseteq \bbM_\xi\Sigma$ be a language such that $\preceq_K$~is a congruence ordering.
For every surjective morphism $\varphi : \bbM\Sigma \to \frakA$ recognising~$K$,
there exists a unique morphism $\varrho : \frakA \to \Syn(K)$ such that
$\syn_K = \varrho \circ \varphi$.
\end{thm}
\begin{proof}
Suppose that $\varphi$~recognises~$K$. By Lemma~\ref{Lem: morphisms recognising K}
we have
\begin{align*}
  \ker \varphi \subseteq {\preceq_K} = \ker \syn_K\,.
\end{align*}
Therefore, we can use the Factorisation Lemma to find a unique function $\varrho : A \to \Syn(K)$
with $\syn_K = \varrho \circ \varphi$.
According to Lemma~\ref{Lem: factorisation is morphism}, this function~$\varrho$ is a morphism.
\end{proof}

Let us take a look at what kind of languages are recognised by a syntactic algebra.
\begin{prop}\label{Prop: comparing syntactic congruences}
Let $K \subseteq \bbM_\xi\Sigma$ and $L \subseteq \bbM_\zeta\Sigma$ be languages such that
$K$~has a syntactic algebra.
The following statements are equivalent.
\begin{enum1}
\item ${\preceq_K} \subseteq {\preceq_L}$
\item $\syn_K : \bbM\Sigma \to \Syn(K)$ recognises~$L$.
\item Every morphism recognising~$K$ also recognises~$L$.
\item $L$~has the form
  \begin{align*}
    \bigcup_{i<m} \bigcap_{k < n_i} p_{ik}^{-1}[K]\,,
    \quad\text{for suitable } m,n_i < \omega \text{ and } p_{ik} \in \bbM(\Sigma+\Box)\,.
  \end{align*}
\end{enum1}
\end{prop}
\begin{proof}
(1)~$\Leftrightarrow$~(2)
follows directly by Lemma~\ref{Lem: morphisms recognising K}
since ${\preceq_K} = \ker \syn_K$.

(3)~$\Rightarrow$~(2) is trivial as $\syn_K$~recognises~$K$.

(1)~$\Rightarrow$~(3)
Suppose that $\varphi : \bbM\Sigma \to \frakA$ recognises~$K$.
By Lemma~\ref{Lem: morphisms recognising K}, it follows that
$\ker \varphi \subseteq {\preceq_K} \subseteq {\preceq_L}$,
which, by the same lemma, implies that $\varphi$~recognises~$L$.

(4)~$\Rightarrow$~(1)
We have shown in Lemma~\ref{Lem: applying context is monotonous}\,(b) that
\begin{align*}
  {\preceq_K} \subseteq {\preceq_{p^{-1}[K]}}\,,
  \quad\text{for every context } p\,.
\end{align*}
To conclude the proof it is therefore sufficient to show that, for all languages $K,L_0,L_1$,
\begin{align*}
  {\preceq_K} \subseteq {\preceq_{L_0}}
  \qtextq{and}
  {\preceq_K} \subseteq {\preceq_{L_1}}
  \qtextq{implies}
  {\preceq_K} \subseteq {\preceq_{L_0 \cup L_1}}
  \qtextq{and}
  {\preceq_K} \subseteq {\preceq_{L_1 \cap L_1}}\,.
\end{align*}
For the first inclusion, suppose that $s \preceq_K t$ and let $p$~be a context such that
$p[s] \in L_0 \cup L_1$.
Then there is some $i < 2$ such that $p[s] \in L_i$.
Hence, $s \preceq_{L_i} t$ implies $p[t] \in L_i \subseteq L_0 \cup L_1$.

Similarly, if $s \preceq_K t$ and $p$~is a context with $p[s] \in L_0 \cap L_1$, then
\begin{alignat*}{-1}
  p[s] &\in L_0 &&\qtextq{and} &s &\preceq_{L_0} t &&\qtextq{implies} &p[t] &\in L_0\,, \\
  p[s] &\in L_1 &&\qtextq{and} &s &\preceq_{L_1} t &&\qtextq{implies} &p[t] &\in L_1\,.
\end{alignat*}
Thus $p[t] \in L_0 \cap L_1$.

(2)~$\Rightarrow$~(4)
By definition of~$\preceq_K$, for every pair of elements $a,b \in \Syn_\xi(K)$
with $a \nleq b$, we can fix some context~$p_{ab}$ such that
\begin{align*}
  p_{ab}[s] \in K \qtextq{and} p_{ab}[t] \notin K\,,
  \quad\text{for } s \in \syn_K^{-1}(a) \text{ and } t \in \syn_K^{-1}(b)\,.
\end{align*}
Set $P := \syn_K[L]$ and let $Q := \Syn_\xi(K) \setminus P$ be the complement.
For $t \in \syn_K^{-1}[P]$ and $u \in \syn_K^{-1}[Q]$, it follows that
\begin{align*}
  t \npreceq_K u \qtextq{implies} p_{ab}[t] \in K
  \quad\text{where } a := \syn_K(t) \text{ and } b := \syn_K(u)\,.
\end{align*}
Similarly,
for $t \in \syn_K^{-1}[Q]$ and $s \in \syn_K^{-1}[P]$, we have
\begin{align*}
  s \npreceq_K t \qtextq{implies} p_{ab}[t] \notin K
  \quad\text{where } a := \syn_K(s) \text{ and } b := \syn_K(t)\,.
\end{align*}
Taken together it follows that
\begin{align*}
  t \in \syn_K^{-1}[P]
  \quad\iff\quad
  \text{there is some } a \in P \text{ with }
  p_{ab}[t] \in K \text{ for all } b \in Q\,.
\end{align*}
Thus,
\begin{align*}
  L & = \syn_K^{-1}[P] = \bigcup_{a \in P} \bigcap_{b \in Q} p_{ab}^{-1}[K]\,.
  \qedhere
\end{align*}
\end{proof}

The next proposition describes languages recognised by syntactic algebras via arbitrary morphisms.
\begin{prop}\label{Prop: languages recognised by Syn(K)}
Let $K \subseteq \bbM_\xi\Sigma$ be a language with a syntactic algebra.
A~language $L \subseteq \bbM_\zeta\Gamma$ is recognised by $\Syn(K)$ if, and only if,
\begin{align*}
  L = \varphi^{-1}\bigl[\bigcup_{i<m} \bigcap_{k < n_i} p_{ik}^{-1}[K]\bigr]\,,
\end{align*}
for a suitable morphism $\varphi : \bbM\Gamma \to \bbM\Sigma$, numbers $m,n_i < \omega$, and
contexts $p_{ik} \in \bbM(\Sigma+\Box)$.
\end{prop}
\begin{proof}
$(\Rightarrow)$
Suppose that $L = \psi^{-1}[P]$ for some $P \subseteq \Syn(K)$.
By Lemma~\ref{Lem: FX projective}, there exists a morphism $\varphi : \bbM\Gamma \to \bbM\Sigma$
with $\syn_K \circ \varphi = \psi$.
Consequently, it follows by Proposition~\ref{Prop: comparing syntactic congruences} that
\begin{align*}
  L = \psi^{-1}[P]
    = (\syn_K \circ \varphi)^{-1}[P]
    = \varphi^{-1}[\syn_K^{-1}[P]]
    = \varphi^{-1}\bigl[\bigcup_{i<m} \bigcap_{k < n_i} p_{ik}^{-1}[K]\bigr]\,,
\end{align*}
for suitable contexts~$p_{ik}$.

$(\Leftarrow)$
By Proposition~\ref{Prop: comparing syntactic congruences},
the morphism $\syn_K : \bbM\Sigma \to \Syn(K)$ recognises the language
\begin{align*}
  M := \bigcup_{i<m} \bigcap_{k < n_i} p_{ik}^{-1}[K]\,.
\end{align*}
Consequently, $\syn_K \circ \varphi : \bbM\Gamma \to \Syn(K)$ recognises $\varphi^{-1}[M] = L$.
\end{proof}

In general, there is no reason why the syntactic congruence~$\preceq_K$
should be a congruence ordering. Let us isolate one property of the functor~$\bbM$
ensuring that this is in fact the case.
\begin{defi}
A functor $\bbM : \Pos^\Xi \to \Pos^\Xi$ is \emph{finitary} if
it commutes with directed colimits, that is, if
\begin{align*}
  \bbM(\varinjlim D) = \varinjlim (\bbM \circ D)\,,
  \quad\text{for every directed diagram } D : I \to \Pos^\Xi.
\end{align*}
\upqed
\markenddef
\end{defi}
\begin{Rem}
(a)
In particular, if $\bbM$~is finitary then $\bbM A$~is equal to the directed colimit
of the diagram consisting of~$\bbM C$, for all finite $C \subseteq A$.

(b)
The word functor $\bbM A := A^+$ is finitary as every finite word uses only finitely many labels.
The functor
\begin{align*}
  \bbM\langle A_1,A_\infty\rangle := \langle A_1^+,\ A_1^+A_\infty \cup A_1^\omega\rangle
\end{align*}
for infinite words, on the other hand, is not finitary as
an infinite word can contain infinitely many different labels. Thus, in general
$A^\omega \neq \bigcup {\set{ C^\omega }{ C \subseteq A \text{ finite} }}$.
\end{Rem}

\begin{prop}\label{Prop: syntactic congruence for finitary functors}
Let\/ $\frakA$~be a finitary\/ $\bbM$-algebra and $K \subseteq A_\xi$ a set.
If\/ $\bbM$~is finitary, then $\preceq_K$~is a congruence ordering on\/~$\frakA$.
\end{prop}
\begin{proof}
We use the characterisation from Proposition~\ref{Prop: characterisations of congruences}\,(3).
Hence, fix $u \in \bbM{\preceq_K}$.
We have to show that
\begin{align*}
  \pi(\bbM q_0(u)) \preceq_K \pi(\bbM q_1(u))\,,
\end{align*}
where $q_0,q_1 : A \times A \to A$ are the two projections.
As $\bbM$~is finitary, there exists a finite relation $R \subseteq {\preceq_K}$
such that $u \in \bbM R$.
Let $\langle a_0,b_0\rangle,\dots,\langle a_{m-1},b_{m-1}\rangle$ be an
enumeration of~$R$ and set
\begin{alignat*}{-1}
  t_k &:= \bbM p_k(u)\,,
  &&\qtextq{where}
  &p_k(\langle a_i,b_i\rangle) &:=
    \begin{cases}
      a_i &\text{if } i \geq k\,, \\
      b_i &\text{if } i < k\,.
    \end{cases} \\
  r_k &:= \bbM p'_k(u)\,,
  &&\qtextq{where}
  &p'_k(\langle a_i,b_i\rangle) &:=
    \begin{cases}
      a_i  &\text{if } i > k\,, \\
      \Box &\text{if } i = k\,, \\
      b_i  &\text{if } i < k\,.
    \end{cases}
\end{alignat*}
Then
$\pi(t_k) = r_k[a_k]$ and $\pi(t_{k+1}) = r_k[b_k]$,
and it follows by Lemma~\ref{Lem: applying context is monotonous} that
\begin{align*}
  a_k \preceq_K b_k
  \qtextq{implies}
  \pi(t_k) = r_k[a_k] \preceq_K r_k[b_k] = \pi(t_{k+1})\,.
\end{align*}
Consequently,
$\pi(\bbM q_0(u)) =
 \pi(t_0) \preceq_K \dots \preceq_K \pi(t_m)
 = \pi(\bbM q_1(u))$,
as desired.
\end{proof}

Unfortunately, not all the monads~$\bbM$ used in applications are finitary.
In particular those needed for languages of infinite words or infinite trees are not.
Therefore, we have to extend the preceding proposition to a larger class of functors.
It turns out that,
in all the known examples of a non-finitary functors where syntactic algebras exists,
the functor in question is `ruled' in a certain sense by a subfunctor which \emph{is} finitary.
The precise definitions are as follows.
\begin{defi}
Let $\langle\bbM_0,\mu_0,\varepsilon_0\rangle$ and
$\langle\bbM_1,\mu_1,\varepsilon_1\rangle$ be monads.

(a) A natural transformation $\varrho : \bbM_0 \Rightarrow \bbM_1$ is a
\emph{morphism of monads} if
\begin{align*}
  \varepsilon_1 = \varrho \circ \varepsilon_0
  \qtextq{and}
  \mu_1 \circ (\varrho \circ \bbM_0\varrho) = \varrho \circ \mu_0\,.
\end{align*}
In this case we say that $\bbM_0$~is a \emph{reduct} of~$\bbM_1$.

(b) Let $\varrho : \bbM_0 \Rightarrow \bbM_1$ be a morphism of monads
and $\frakA = \langle A,\pi\rangle$ an $\bbM_1$-algebra.
The \emph{$\varrho$-reduct} of~$\frakA$ is the $\bbM_0$-algebra
$\langle A,\pi \circ \varrho\rangle$.
If $\varrho$~is understood, we also speak of an \emph{$\bbM_0$-reduct} of~$\frakA$.

(c) A morphism $\varrho : \bbM^\circ \Rightarrow \bbM$ of monads is
\emph{dense} over a class~$\calC$ of $\bbM$-algebras if,
for all $\frakA \in \calC$, $C \subseteq A$, and $s \in \bbM C$,
there exists $s^\circ \in \bbM^\circ C$ with $\pi(s^\circ) = \pi(s)$.

(d) We say that a monad~$\bbM$ is \emph{essentially finitary} over a class~$\calC$
if there exists a morphism $\varrho : \bbM^\circ \Rightarrow \bbM$ such that
$\bbM^\circ$~is finitary and $\varrho$~is dense over the closure of~$\calC$ under binary products.
\markenddef
\end{defi}

Let us again consider the functor
\begin{align*}
  \bbM\langle A_1,A_\infty\rangle := \langle A_1^+,\ A_1^+A_\infty \cup A_1^\omega\rangle
\end{align*}
for infinite words and let
\begin{align*}
  \bbM^\circ\langle A_1,A_\infty\rangle :=
    \langle A_1^+,\ A_1^+A_\infty \cup A_1^{\mathrm{up}}\rangle\,,
\end{align*}
where $A_1^{\mathrm{up}}$ denotes the set of all ultimately periodic words in~$A_1^\omega$.
Then the inclusion map $\bbM^\circ \Rightarrow \bbM$ is dense over the class of all finite
$\omega$-semigroups since the infinite product of a finite $\omega$-semigroup is completely
determined by its restriction to all ultimately periodic words.
The case of infinite trees is similar and will be treated in detail in
Section~\ref{Sect:applications}.

If $\bbM^\circ \Rightarrow \bbM$ is dense over~$\calC$, every $\bbM$-algebra in~$\calC$
is uniquely determined by its $\bbM^\circ$-reduct.
This will be used below to prove the existence of syntactic algebras for essentially finitary
monads.
\begin{lem}\label{Lem: dense reducts}
Let $\varrho : \bbM^\circ \Rightarrow \bbM$ be dense over a class~$\calC$ that is closed
under binary products.
\begin{enuma}
\item Any two algebras in~$\calC$ with the same\/ $\bbM^\circ$-reduct are isomorphic.
\item Let $\varphi : \frakA^\circ \to \frakB^\circ$ be a morphism of\/ $\bbM^\circ$-algebras
  and assume that\/ $\frakA^\circ$~and\/~$\frakB^\circ$ are the\/ $\bbM^\circ$-reducts of two
  $\bbM$-algebras\/ $\frakA,\frakB \in \calC$.
  Then $\varphi$~is also a morphism\/ $\frakA \to \frakB$ of\/ $\bbM$-algebras.
\item A~relation\/~$\sqsubseteq$ is a congruence ordering on an\/ $\bbM$-algebra\/ $\frakA \in \calC$
  if, and only if, it is a congruence ordering on the\/ $\bbM^\circ$-reduct\/~$\frakA^\circ$
  of\/~$\frakA$.
\end{enuma}
\end{lem}
\begin{proof}
(a)
Suppose that $\calC$~contains two $\bbM$-algebras $\frakA = \langle A,\pi\rangle$ and
$\frakA' = \langle A,\pi'\rangle$ with the same $\bbM^\circ$-reduct
$\frakA^\circ = \langle A,\pi^\circ\rangle$.
To show that $\pi = \pi'$, fix an element $s \in \bbM A$.
Set $t := \bbM d(s) \in \bbM\Delta$ where
$\Delta := \set{ \langle a,a\rangle }{ a \in A }$ is the diagonal of $A \times A$ and
$d : A \to \Delta$ is the diagonal map.
By assumption, the product $\frakA \times \frakA'$ belongs to~$\calC$.
As~$\varrho$~is dense, we can find some $t^\circ \in \bbM^\circ\Delta$ with
$\pi^\circ(t^\circ) = \pi(t)$.
Note that $t^\circ \in \bbM^\circ\Delta$ implies that
$\bbM^\circ p(t^\circ) = \bbM^\circ q(t^\circ)$
where $p,q : A \times A \to A$ are the two projections.
Consequently,
\begin{align*}
  \pi(s)
   = \pi(\bbM p(t))
  &= p(\pi(t)) \\
  &= p(\pi^\circ(t^\circ)) \\
  &= \pi^\circ(\bbM^\circ p(t^\circ)) \\
  &= \pi^\circ(\bbM^\circ q(t^\circ)) \\
  &= q(\pi^\circ(t^\circ)) \\
  &= q(\pi(t))
   = \pi'(\bbM q(t))
   = \pi'(s)\,.
\end{align*}

(b) Fix $s \in \bbM A$. To show that $\pi(\bbM\varphi(s)) = \varphi(\pi(s))$,
we consider the graph
\begin{align*}
  G := \set{ \langle a,\varphi(a)\rangle }{ a \in A }
\end{align*}
of~$\varphi$.
Let $i := \langle\id,\varphi\rangle: A \to G$ be the natural bijection and
set $t := \bbM i(s) \in \bbM G$.
Since $\frakA \times \frakB \in \calC$ and $\varrho$~is dense,
we can find some $t^\circ \in \bbM^\circ G$ with $\pi(t^\circ) = \pi(t)$.
Let $p : A \times B \to A$ and $q : A \times B \to B$ be the two projections.
Note that
\begin{align*}
  \varphi = q \circ i
  \qtextq{and}
  q(g) = \varphi(p(g))\,, \quad\text{for } g \in G\,,
\end{align*}
which implies that $\bbM^\circ q(t^\circ) = \bbM^\circ(\varphi \circ p)(t^\circ)$.
Therefore,
\begin{align*}
  \pi(\bbM\varphi(s))
   = \pi(\bbM q(t))
  &= q(\pi(t)) \\
  &= q(\pi(t^\circ)) \\
  &= \pi(\bbM^\circ q(t^\circ)) \\
  &= \pi(\bbM^\circ(\varphi\circ p)(t^\circ)) \\
  &= \varphi(p(\pi(t^\circ))) \\
  &= \varphi(p(\pi(t)))
   = \varphi(\pi(\bbM p(t)))
   = \varphi(\pi(s))\,.
\end{align*}

(c)
Clearly, every congruence ordering of~$\frakA$ is also one of~$\frakA^\circ$.
Conversely, suppose that $\sqsubseteq$~is a congruence ordering of~$\frakA^\circ$.
We use the characterisation from Proposition~\ref{Prop: characterisations of congruences}\,(3).
Thus, let $u \in \bbM{\sqsubseteq}$.
As the product $\frakA \times \frakA$ belongs to~$\calC$ and $\varrho$~is dense over~$\calC$,
we can find some $u^\circ \in \bbM^\circ{\sqsubseteq}$ with $\pi(u^\circ) = \pi(u)$.
By assumption, we have $\pi(\bbM^\circ p_0(u^\circ)) \sqsubseteq \pi(\bbM^\circ p_1(u^\circ))$,
where $p_0,p_1 : A \times A \to A$ are the two projections.
Since
\begin{align*}
  \pi(\bbM p_i(u)) = p_i(\pi(u)) = p_i(\pi(u^\circ)) = \pi(\bbM^\circ p_i(u^\circ))\,,
\end{align*}
this implies that $\pi(\bbM p_0(u)) \sqsubseteq \pi(\bbM p_1(u))$, as desired.
\end{proof}

\begin{thm}\label{Thm: essentially finitary monads have syntactic algebras}
Suppose that $\bbM$~is essentially finitary over~$\calC$.
If a language $K \subseteq \bbM_\xi\Sigma$ is recognised by some morphism
$\varphi : \bbM\Sigma \to \frakA$ with $\frakA \in \calC$,
then $\preceq_K$~is a congruence ordering on~$\bbM\Sigma$.
\end{thm}
\begin{proof}
Suppose that $K = \varphi^{-1}[P]$ for $P \subseteq A_\xi$.
Let $\frakB \subseteq \frakA$ be the subalgebra induced by $\rng \varphi$.
By Proposition~\ref{Prop: syntactic congruence for finitary functors},
$\preceq_P$~is a congruence ordering on the $\bbM^\circ$-reduct $\frakA^\circ$ of~$\frakA$.
Hence, Lemma~\ref{Lem: dense reducts}\,(c) implies that it is also a congruence ordering on~$\frakA$.
Consequently, its restriction is one on~$\frakB$.
(If $\preceq_P$~is the kernel of some morphism~$q$, then its restriction to~$B$
is the kernel of $q \circ i$, where $i : \frakB \to \frakA$ is the inclusion morphism.)
Thus, ${\preceq_P} = \ker \syn_P$ where $\syn_P : \frakB \to \Syn(P) = \frakB/{\preceq_P}$
is the quotient morphism. We will show that
\begin{align*}
  {\preceq_K} = (\varphi \times \varphi)^{-1}[{\preceq_P}]\,.
\end{align*}
It then follows that ${\preceq_K} = \ker (\syn_P \circ \varphi)$
is also the kernel of a morphism and, thus, a congruence ordering.
Hence, it remains to prove the claim.

$(\supseteq)$
Let $f := \varphi \circ \sing + 1: \Sigma + \Box \to B + \Box$ be the function mapping
$c \in \Sigma$ to $\varphi(\sing(c))$ and $\Box$ to~$\Box$.
For $s \in \bbM\Sigma$ and $p \in \bbM(\Sigma + \Box)$, we have
\begin{align*}
  p[s] \in K \quad\iff\quad \varphi(p[s]) \in P
             \quad\iff\quad (\bbM f(p))[\varphi(s)] \in P\,.
\end{align*}
If $\varphi(s) \preceq_P \varphi(t)$ it therefore follows that
\begin{align*}
  p[s] \in K
  \quad\Rightarrow\quad
  (\bbM f(p))[\varphi(s)] \in P
  \quad\Rightarrow\quad
  (\bbM f(p))[\varphi(t)] \in P
  \quad\Rightarrow\quad
  p[t] \in K\,,
\end{align*}
for all $p \in \bbM(\Sigma + \Box)$.
Consequently, $s \preceq_K t$.

$(\subseteq)$
Suppose that $s \preceq_K t$ and fix $p \in \bbM(B + \Box)$.
The morphism $\bbM(\varphi+1) : \bbM(\bbM\Sigma + \Box) \to \bbM(B + \Box)$
is surjective since
$\varphi$~is surjective and $\bbM$~preserves surjectivity.
Thus, we can fix some context $\hat p \in (\bbM(\varphi+1))^{-1}(p)$.
It follows that
\begin{align*}
  p[\varphi(s)] \in P
  &\quad\Rightarrow\quad
  \varphi(\hat p[s]) \in P \\
  &\quad\Rightarrow\quad
  \hat p[s] \in K \\
  &\quad\Rightarrow\quad
  \hat p[t] \in K \\
  &\quad\Rightarrow\quad
  \varphi(\hat p[t]) \in P
  \quad\Rightarrow\quad
  p[\varphi(t)] \in P\,.
\end{align*}
Consequently, $\varphi(s) \preceq_P \varphi(t)$.
\end{proof}

\section{Varieties}   
\label{Sect: varieties}

After these preparations, we come to the first of the central theorems of algebraic
language theory\?: the Variety Theorem. This theorem characterises which kind of language
families are amenable to our algebraic tools by
establishing a correspondence between language families and
the classes of algebras recognising them.

Before we can formally define the families and classes involved, we need to introduce a bit of
notation that allows us to transfer problems to a setting with only finitely many sorts.
\begin{defi}
Let $\Delta \subseteq \Xi$ be a set of sorts and $A$~a $\Xi$-sorted set.

(a) We denote by~$A|_\Delta$ the subset of~$A$ containing only the elements with a sort in~$\Delta$.
(Depending on the circumstances, we will treat~$A|_\Delta$ either as a set in $\Pos^\Delta$, or as
a set in $\Pos^\Xi$ that just happens to have no element with a sort in $\Xi \setminus \Delta$.)

(b) For a function $f : A \to B$ we denote the induced function $A|_\Delta \to B|_\Delta$
by~$f|_\Delta$.

(c) The corresponding restriction of the functor~$\bbM$ is defined by
\begin{align*}
  \bbM|_\Delta A := (\bbM(A|_\Delta))|_\Delta
  \qtextq{and}
  \bbM|_\Delta f := (\bbM(f|_\Delta))|_\Delta\,.
\end{align*}
(Again, depending on what is convenient, we will consider~$\bbM|_\Delta$
either as a functor $\Pos^\Delta \to \Pos^\Delta$ or $\Pos^\Xi \to \Pos^\Xi$.)

(d) For an $\bbM$-algebra~$\frakA$ we denote by~$\frakA|_\Delta$ the $\bbM|_\Delta$-algebra
with domain~$A|_\Delta$ and product $\pi \restriction \bbM|_\Delta A$.
For a class~$\calC$ of $\bbM$-algebras we set
$\calC|_\Delta := \set{ \frakA|_\Delta }{ \frakA \in \calC }$.

(e) An $\bbM$-algebra~$\frakB$ is a \emph{sort-accumulation point} of a class~$\calA$
of $\bbM$-algebras if, for every finite subset $\Delta \subseteq \Xi$,
there is some algebra $\frakA \in \calA$ such that
$\frakA|_\Delta$~is a quotient of~$\frakB|_\Delta$ (as $\bbM|_\Delta$-algebras).
\markenddef
\end{defi}
Let us check that these notions are well-behaved.
\begin{lem}
\begin{enuma}
\item $\langle \bbM|_\Delta,\Flat|_\Delta,\sing|_\Delta\rangle$ is a monad
  on $\Pos^\Delta$.
\item If\/ $\frakA$~is an $\bbM$-algebra,\/ $\frakA|_\Delta$~is an $\bbM|_\Delta$-algebra.
\end{enuma}
\end{lem}
\begin{proof}
(a)
It is sufficient to show that
\begin{alignat*}{-1}
  \bbM|_\Delta A = (\bbM A)|_\Delta
  \qtextq{and}
  \bbM|_\Delta f = (\bbM f)|_\Delta\,, \quad\text{for } f : A|_\Delta \to B|_\Delta\,.
\end{alignat*}
Then it follows that every equation satisfied by~$\bbM$ is also satisfied by~$\bbM|_\Delta$.

The second of these equations follows immediately from the definition as
$f : A|_\Delta \to B|_\Delta$ implies that $f|_\Delta = f$. Consequently,
\begin{align*}
  \bbM|_\Delta f = (\bbM(f|_\Delta))|_\Delta = (\bbM f)|_\Delta\,.
\end{align*}

For the first equation, let $\mathsf{1}$~be the set containing exactly one element of each
sort~$\xi$ with $A_\xi \neq \emptyset$,
and let $f : A \to \mathsf{1}$ be the unique function.
Since $\bbM$~preserves preimages, it follows that
\begin{align*}
  \bbM(A|_\Delta) = (\bbM f)^{-1}[\bbM(\mathsf{1}|_\Delta)] = (\bbM A)|_\Delta\,.
\end{align*}
Consequently, $\bbM|_\Delta A = (\bbM(A|_\Delta))|_\Delta = (\bbM A)|_\Delta$.

(b) similar.
\end{proof}

We will show below that there is a precise correspondence between the following families of
languages and classes of algebras.
\begin{defi}
(a) A \emph{(positive) variety of languages} is a family~$\calK$ of languages that is closed under
(i)~finite unions and intersections, (ii)~inverse morphisms, and (iii)~derivatives.

(b)
A class~$\calV$ of finitary $\bbM$-algebras is a \emph{pseudo-variety} if it is closed under
(i)~quotients, (ii)~finitary subalgebras of finite products, and
(iii)~sort-accumulation points.
\markenddef
\end{defi}

\begin{Rem}
(a) In the definition of a pseudo-variety, closure under quotients is superfluous
as it is implied by closure under sort-accumulation points. We have left it as
a requirement in the definition to emphasis the analogy to the usual definition
in the setting with finitely many sorts.

(b) The reason why we combine the operations of taking subalgebras and forming products into
a single one is that, in general, the product of two finitary algebras need not be
finitely generated (see~\cite{Blumensath20} for a counterexample).

(c) If the set of sorts~$\Xi$ is finite, our notion of a pseudo-variety coincides with
the standard one.
\end{Rem}

\begin{lem}\label{Lem: basic properties of pseudo-varieties}
Let $\calV$~be a class of finitary $\bbM$-algebras and let
$\sfC$,~$\sfH$, $\sfS_\omega$, and~$\sfP_\omega$ denote the operations of taking, respectively,
sort-accumulation points, quotients, finitary subalgebras, and finite products.
\begin{enuma}
\item $(\sfC\sfH\sfS_\omega\sfP_\omega)^2(\calV) = \sfC\sfH\sfS_\omega\sfP_\omega(\calV)\,.$
\item The following conditions are equivalent.
  \begin{enum1}
  \item $\calV$ is a pseudo-variety.
  \item $\calV = \sfC\sfH\sfS_\omega\sfP_\omega(\calV)$
  \item $\calV$~satisfies the following to statements.
    \begin{itemize}
    \item For every finite set $\Delta \subseteq \Xi$ of sorts,
      the reduct~$\calV|_\Delta$ is a pseudo-variety.
    \item $\calV$~is the closure of the reducts~$\calV|_\Delta$ in the sense that
      \mathindent=7em
      \begin{align*}
        \frakA \in \calV \quad\iff\quad \frakA|_\Delta \in \calV|_\Delta\,,
        \quad\text{for all finite } \Delta \subseteq \Xi\,.
      \end{align*}
    \end{itemize}
  \end{enum1}
\end{enuma}
\end{lem}
\begin{proof}
(b)
(1)~$\Rightarrow$~(2) immediately follows from the closure properties of a pseudo-variety
and (2)~$\Rightarrow$~(1) follows once we have proved~(a) below.

(3)~$\Rightarrow$~(2)
If $\frakA \in \sfC\sfH\sfS_\omega\sfP_\omega(\calV)$,
then $\frakA|_\Delta \in \sfC\sfH\sfS_\omega\sfP_\omega(\calV|_\Delta) = \calV|_\Delta$,
for all finite $\Delta \subseteq \Xi$. This implies that $\frakA \in \calV$.

(2)~$\Rightarrow$~(3)
If $\frakA$~is an algebra with
$\frakA|_\Delta \in \calV|_\Delta$, for all finite $\Delta \subseteq \Xi$.
Then $\sfC(\calV) = \calV$ implies that $\frakA \in \calV$.
For the other claim, note that
\begin{align*}
  \calV = \sfC\sfH\sfS_\omega\sfP_\omega(\calV)
  \qtextq{implies}
  \calV|_\Delta = \sfC\sfH\sfS_\omega\sfP_\omega(\calV|_\Delta)\,,
\end{align*}
since the reduct operation~$|_\Delta$ commutes with $\sfC$, $\sfH$, $\sfS_\omega$,
and $\sfP_\omega$.
Hence, the claim follows by the implication (2)~$\Rightarrow$~(1) we have already proved above.

(a)
We have to show that $\sfC\sfH\sfS_\omega\sfP_\omega(\calV)$ is closed under all four operations.
First, note that $\sfC$,~$\sfH$, $\sfS_\omega$, and~$\sfP_\omega$ are closure operators.

Furthermore, we have $\sfH \circ \sfC = \sfC = \sfC \circ \sfH$
since the operation of taking a sort-accumulation point already performs a quotient.
Thus, in particular, $\sfC\sfH\sfS_\omega\sfP_\omega(\calV) = \sfC\sfS_\omega\sfP_\omega(\calV)$,
which reduces the number of cases we have to consider.
To conclude the proof, it is sufficient to establish the following three statements.
\begin{enumr}
\item $\sfS_\omega\sfP_\omega\sfS_\omega\sfP_\omega(\calC) \subseteq \sfS_\omega\sfP_\omega(\calC)$
\item $\sfP_\omega\sfC(\calC) \subseteq \sfC\sfP_\omega(\calC)$
\item $\sfS_\omega\sfC\sfP_\omega(\calC) \subseteq \sfC\sfS_\omega\sfP_\omega(\calC)$
\end{enumr}
where $\calC$~is a class of finitary $\bbM$-algebras.

(i)
Let $\frakA \subseteq \prod_{i < n} \frakB_i$ be finitary
where $\frakB_i \subseteq \frakC_i$ with $\frakC_i \in \sfP_\omega(\calC)$.
Then $\frakA \subseteq \prod_i \frakB_i \subseteq \prod_i \frakC_i$
and $\frakA \in \sfS_\omega\sfP_\omega(\calC)$.

(ii)
Let $\frakA = \prod_{i < n} \frakB_i$ where, for every $i < n$ and every
finite $\Delta \subseteq \Xi$, there exists a quotient map
$q^i_\Delta : \frakC_\Delta^i|_\Delta \to \frakB_i|_\Delta$ with $\frakC_\Delta^i \in \calC$.
Setting $\frakD_\Delta := \prod_{i < n} \frakC_\Delta^i$ we obtain an algebras in
$\sfP_\omega(\calC)$ with quotient maps
\begin{align*}
  \prod_{i < n} q^i_\Delta : \frakD_\Delta|_\Delta \to \frakA|_\Delta\,.
\end{align*}
Hence, $\frakA \in \sfC\sfP_\omega(\calC)$.

(iii)
Let $\frakA$~be a finitary subalgebra of~$\frakB$ and
fix quotient maps $q_\Delta : \frakC_\Delta|_\Delta \to \frakB|_\Delta$ with
$\frakC_\Delta \in \sfP_\omega(\calC)$, for every finite $\Delta \subseteq \Xi$.
We have to find a finitary subalgebra $\frakD_\Delta \subseteq \frakC_\Delta$
such that $\frakA|_\Delta$~is a quotient of~$\frakD_\Delta|_\Delta$.
Given $\Delta \subseteq \Xi$,
let $\frakD_\Delta$~be the subalgebra of~$\frakC_\Delta$ generated by the set
$X_\Delta := q_\Delta^{-1}[A|_\Delta]$.
Then $\frakD_\Delta$~is finitely generated and, hence, finitary.
Thus, $\frakD_\Delta \in \sfS_\omega\sfP_\omega(\calC)$.

It remains to prove that $\frakA|_\Delta$~is a quotient of~$\frakD_\Delta|_\Delta$.
For this it is sufficient to show that $D_\Delta|_\Delta = X_\Delta$ as that means that
$q_\Delta$~restricts to a surjective map $\frakD_\Delta|_\Delta \to \frakA|_\Delta$.
Hence, fix $a \in D_\Delta|_\Delta$. As~$D_\Delta$~is generated by~$X_\Delta$,
we can find some $s \in \bbM X_\Delta$ with $a = \pi(s)$. Then $s \in \bbM|_\Delta X_\Delta$,
which implies that $a = \pi(s) \in X_\Delta$ by the fact that preimages of subalgebras
induce subalgebras.
\end{proof}

The aim of this section is to establish a one-to-one correspondence between
varieties of languages and pseudo-varieties of $\bbM$-algebras.
The arguments are mostly standard, except for some adjustments needed to support
infinitely many sorts.
We start with the following observation.
\begin{lem}\label{Lem: calV contains Syn(K)}
Let $\calV$~be a pseudo-variety and $K \subseteq \bbM_\xi\Sigma$ a language with a syntactic algebra.
Then $K$~is recognised by some algebra\/ $\frakA \in \calV$ if, and only if,\/ $\Syn(K) \in \calV$.
\end{lem}
\begin{proof}
$(\Leftarrow)$ is trivial since $\Syn(K)$ recognises~$K$.
For $(\Rightarrow)$, consider a morphism $\varphi : \bbM\Sigma \to \frakA$ recognising~$K$
with $\frakA \in \calV$.
As~$\calV$ is closed under finitary subalgebras, we may assume that $\varphi$~is surjective.
We can therefore use Theorem~\ref{Thm: Syn(K) terminal}
to find a morphism $\varrho : \frakA \to \Syn(K)$ with $\syn_K = \varrho \circ \varphi$.
As $\syn_K$~is surjective, so is~$\varrho$.
By closure of~$\calV$ under quotients, it follows that $\Syn(K) \in \calV$.
\end{proof}

The first step in correlating varieties of languages and pseudo-varieties of algebras consists
in following the fact.
\begin{prop}
If $\calK$~is the family of languages recognised by the algebras of a pseudo-variety~$\calV$
of\/ $\bbM$-algebras, then $\calK$~is a variety of languages.
\end{prop}
\begin{proof}
We have to prove three closure properties.

(1) We start with inverse morphisms.
Suppose that $K = \varphi^{-1}[P]$ for a morphism $\varphi : \bbM\Gamma \to \frakA$ with
$\frakA \in \calV$ and $P \subseteq A_\xi$.
Let $\psi : \bbM\Sigma \to \bbM\Gamma$ be a morphism. Then
\begin{align*}
  \psi^{-1}[K] = \psi^{-1}[\varphi^{-1}[P]] = (\varphi \circ \psi)^{-1}[P]
\end{align*}
is recognised by the morphism $\varphi \circ \psi$ to $\frakA \in \calV$.

(2) Next, we consider closure under derivatives.
Let $K \in \calK$ and fix a context~$p$. By assumption,
there is a morphism $\varphi : \bbM\Sigma \to \frakA$ recognising~$K$ with $\frakA \in \calV$.
By Proposition~\ref{Prop: comparing syntactic congruences}, $\varphi$~also recognises~$p^{-1}[K]$.
Hence, $p^{-1}[K] \in \calK$.

(3) It remains to prove closure under finite intersections and unions.
Clearly, the empty union~$\emptyset$ and the empty intersection~$\bbM\Sigma$ are
recognised by any morphism. Thus, it is sufficient to consider binary unions and intersections.
Consider two morphisms $\varphi : \bbM\Sigma \to \frakA$ and $\psi : \bbM\Sigma \to \frakB$
with $\frakA,\frakB \in \calV$,
and set $K := \varphi^{-1}[P]$ and $L := \psi^{-1}[Q]$, for upwards closed $P \subseteq A_\xi$ and
$Q \subseteq B_\xi$.
Then
\begin{align*}
  K \cap L &= \langle\varphi,\psi\rangle^{-1}[P \times Q]\,, \\
  K \cup L &= \langle\varphi,\psi\rangle^{-1}[(P \times B_\xi) \cup (A_\xi \times Q)]
\end{align*}
are recognised by $\langle\varphi,\psi\rangle : \bbM\Sigma \to \frakA \times \frakB$.
Let $\frakC$~be the subalgebra of $\frakA \times \frakB$ induced by the
range of~$\langle\varphi,\psi\rangle$.
Then $\frakC$~is finitary and $\frakC \in \calV$.
\end{proof}

It remains to prove the converse direction of the correspondence.
We start with two lemmas.
\begin{lem}\label{Lem: languages recognised by quotients}
Let $q : \frakA \to \frakB$ be a surjective morphism.
Every language recognised by\/~$\frakB$ is also recognised by\/~$\frakA$.
\end{lem}
\begin{proof}
Suppose that $L = \psi^{-1}[P]$ where $\psi : \bbM\Sigma \to \frakB$ and $P \subseteq B_m$ is
upwards closed.
By Lemma~\ref{Lem: FX projective}, there exists a morphism $\varphi : \bbM\Sigma \to \frakA$
such that $q \circ \varphi = \psi$. Setting $Q := q^{-1}[P]$, it follows that
\begin{align*}
  \varphi^{-1}[Q]
  = \varphi^{-1}[q^{-1}[P]]
  = (q \circ \varphi)^{-1}[P]
  = \psi^{-1}[P]
  = L\,.
\end{align*}
\upqed
\end{proof}
\begin{lem}\label{Lem: languages recognised by infinite products}
Every language $K \subseteq \bbM_\xi\Sigma$ that is recognised by a finitary subalgebra\/
$\frakC \subseteq \prod_{i \in I} \frakA^i$ of a product of finitary\/ $\bbM$-algebras\/~$\frakA^i$
can be written as a finite positive boolean combination of languages recognised by the factors\/~$\frakA^i$.
\end{lem}
\begin{proof}
As the statement also holds for infinite products and the proof of that case is not more complicated,
we prove the more general statement.
Let $\varphi : \bbM\Sigma \to \frakC$ be a morphism such that
$K = \varphi^{-1}[P]$ for some upwards closed $P \subseteq C \subseteq \prod_i A^i_\xi$.
Let $p_k : \prod_i \frakA^i \to \frakA^k$ be the projection.
By the way the ordering of the product~$\prod_i A^i$ is defined, we can pick,
for every pair $a,b \in \prod_i A^i_\xi$ of elements with $a \nleq b$,
some index $h \in I$ such that $p_h(a) \nleq p_h(b)$.
Let $H \subseteq I$ be the finite set of such indices~$h$ that correspond to pairs $a,b \in C_\xi$.
For $s \in \bbM_\xi\Sigma$ and $a \in C_\xi$, it follows that
\begin{align*}
  a \nleq \varphi(s) \quad\iff\quad p_h(a) \nleq p_h(\varphi(s))\,, \quad\text{for some } h \in H\,,
\end{align*}
or, equivalently,
\begin{align*}
  \varphi(s) \geq a \quad\iff\quad p_h(\varphi(s)) \geq p_h(a)\,, \quad\text{for all } h \in H\,.
\end{align*}
Consequently,
\begin{align*}
  K
  = \varphi^{-1}[P]
  = \bigcup_{a \in P} \varphi^{-1}[\Aboveseg a]
  = \bigcup_{a \in P} \bigcap_{h \in H} (p_h \circ \varphi)^{-1}[\Aboveseg p_h(a)]\,.
\end{align*}
As the languages $(p_h \circ \varphi)^{-1}[\Aboveseg p_h(a)]$ are recognised
by the morphism $p_h \circ \varphi : \bbM\Sigma \to \frakA^h$,
the claim follows.
\end{proof}

\begin{thm}\label{Thm: first variety theorem}
Let $\calK$~be a variety of languages such that every language in~$\calK$ has a syntactic algebra.
A language~$K$ belongs to~$\calK$ if, and only if,
it is recognised by some algebra from the pseudo-variety~$\calV$ generated by the set
$\calS := \set{ \Syn(K) }{ K \in \calK }$.
\end{thm}
\begin{proof}
$(\Rightarrow)$ Every language $K \in \calK$ is recognised by $\Syn(K)$, which belongs to~$\calV$.

$(\Leftarrow)$
It follows from Lemma~\ref{Lem: basic properties of pseudo-varieties} that
$\calV = \sfC\sfH\sfS_\omega\sfP_\omega(\calS)$. We proceed in several steps.
By Proposition~\ref{Prop: languages recognised by Syn(K)}, every language recognised by
a syntactic algebra $\Syn(K)$ with $K \in \calK$ belongs to~$\calK$.
As $\calK$~is a variety of languages, it follows by Lemmas
\ref{Lem: languages recognised by quotients}~and~\ref{Lem: languages recognised by infinite products}
that every language~$L$ recognised by an algebra in $\sfH\sfS_\omega\sfP_\omega(\calS)$
also belongs to~$\calK$.

Finally, suppose that $\frakB$~is a sort-accumulation point of $\sfH\sfS_\omega\sfP_\omega(\calS)$
and let $\varphi : \bbM\Sigma \to \frakB$ be a morphism recognising $K = \varphi^{-1}[P]$
with $P \subseteq B_\xi$.
Let $\Delta \subseteq \Xi$ be the set consisting of~$\xi$ and all sorts appearing in the
alphabet~$\Sigma$.
By assumption, we can find an algebra $\frakA \in \sfH\sfS_\omega\sfP_\omega(\calS)$
and a surjective morphism $q : \frakA|_\Delta \to \frakB|_\Delta$.
Let $\hat\varphi : \bbM\Sigma \to \frakA$ be the unique morphism with
$\hat\varphi(\sing(c)) := q(\varphi(\sing(c)))$, for $c \in \Sigma$.
For $s \in \bbM_\xi\Sigma$ we then have
\begin{align*}
  \hat\varphi(s) = \pi(\bbM(q \circ \varphi \circ \sing)(s))
                 = q((\pi \circ \bbM\varphi \circ \bbM\sing)(s))
                 = q(\varphi(s))\,,
\end{align*}
where the first step follows from the fact that $q$~is a morphism of $\bbM|_\Delta$-algebras
and $\xi \in \Delta$.
Consequently, $\hat\varphi^{-1}[q[P]] = \varphi^{-1}[P] = K$ and
$K$~is already recognised by an algebra in~$\sfH\sfS_\omega\sfP_\omega(\calS)$.
By what we have already shown above, this implies that $K \in \calK$.
\end{proof}

As we have just seen, every pseudo-variety of algebras is associated with a variety of languages
and every variety of languages is associated with a pseudo-variety of algebras.
We conclude this section by proving that this correspondence is one-to-one.
As usual we start with a lemma.
\begin{lem}\label{Lem: algebras in a pseudo-variety}
Let\/ $\frakA$~be a finitary\/ $\bbM$-algebra such that every language recognised by\/~$\frakA$
has a syntactic algebra.
Then\/ $\frakA$~belongs to a pseudo-variety~$\calV$ if, and only if,\/
$\Syn(K) \in \calV$, for every language~$K$ recognised by\/~$\frakA$.
\end{lem}
\begin{proof}
$(\Rightarrow)$ If $K$~is recognised by $\frakA \in \calV$,
it follows by Lemma~\ref{Lem: calV contains Syn(K)} that $\Syn(K) \in \calV$.

$(\Leftarrow)$
Suppose that $\Syn(K) \in \calV$, for every language~$K$ recognised by~$\frakA$.
As $\calV$~is closed under sort-accumulation points,
it is sufficient to show that, for every finite set $\Delta \subseteq \Xi$,
there is some surjective morphism
$\mu : \frakB|_\Delta \to \frakA|_\Delta$ with $\frakB \in \calV$.
Hence, fix $\Delta \subseteq \Xi$. W.l.o.g.\ we may assume that~$A|_\Delta$
generates~$\frakA$. Let
\begin{align*}
  K_a &:= \pi^{-1}(\Aboveseg a) \cap \bbM(A|_\Delta)\,,
          \qquad\text{for } a \in A|_\Delta\,,\\[1ex]
  q &:= \langle\syn_{K_a}\rangle_{a \in A|_\Delta}
                : \bbM(A|_\Delta) \to \prod_{a \in A|_\Delta} \Syn(K_a)\,,
\end{align*}
and let $\frakB$~be the subalgebra of $\prod_a \Syn(K_a)$ induced by $\rng q$.
Note that $\frakB$~is finitely generated and therefore belongs to~$\calV$.

For $s,t \in \bbM_\xi(A|_\Delta)$ with $\xi \in \Delta$, we have
\begin{alignat*}{-1}
                & \langle s,t\rangle \in \ker q \\
\Rightarrow\quad& s \preceq_{K_a} t\,, &&\quad\text{for all } a \in A|_\Delta\,, \\
\Rightarrow\quad& [s \in K_a\ \Rightarrow\ t \in K_a]\,,
                  &&\quad\text{for all } a \in A_\xi\,, \\
\Rightarrow\quad& [a \leq \pi(s) \ \Rightarrow\ a \leq \pi(t)]\,,
                  &&\quad\text{for all } a \in A_\xi\,, \\
\Rightarrow\quad& \pi(s) \leq \pi(t)\,.
\end{alignat*}
Consequently, the Factorisation Lemma provides a function $\mu : \frakB|_\Delta \to \frakA|_\Delta$
such that
$\mu \circ q \restriction \bbM|_\Delta A = \pi \restriction \bbM|_\Delta A$.
Note that the restriction $\pi \restriction \bbM|_\Delta A$ is surjective
(as a morphism $\bbM|_\Delta A \to \frakA|_\Delta$),
since $A|_\Delta$~generates~$\frakA$.
Hence, so is~$\mu$.
\end{proof}

\begin{thm}[Variety Theorem]\label{Thm: variety theorem}
Let $\calV$~be a pseudo-variety of\/ $\bbM$-algebras such that every language recognised
by an algebra in~$\calV$ has a syntactic algebra, and let~$\calK$ be a variety of languages
such that every language in~$\calK$ has a syntactic algebra.
The following statements are equivalent.
\begin{enum1}
\item $\calK$~consists of those languages that are recognised by some algebra in~$\calV$.
\item $\calK$~consists of all languages~$K$ with $\Syn(K) \in \calV$.
\item $\calV$~consists of those algebras that only recognise languages in~$\calK$.
\item $\calV$~is the pseudo-variety generated by the set $\set{ \Syn(K) }{ K \in \calK }$.
\end{enum1}
\end{thm}
\begin{proof}
(1)~$\Leftrightarrow$~(2) follows by Lemma~\ref{Lem: calV contains Syn(K)} and
(4)~$\Rightarrow$~(1) by Theorem~\ref{Thm: first variety theorem}.

(2)~$\Rightarrow$~(3)
If $\frakA \in \calV$ and $K$~is recognised by~$\frakA$, it follows by
Lemma~\ref{Lem: algebras in a pseudo-variety} that $\Syn(K) \in \calV$.
By~(2), this implies that $K \in \calK$.
Conversely, if $\frakA$~only recognises language in~$\calK$, (2)~implies that
$\Syn(K) \in \calV$ for all languages~$K$ recognised by~$\frakA$.
By Lemma~\ref{Lem: algebras in a pseudo-variety} it follows that $\frakA \in \calV$.

(3)~$\Rightarrow$~(4)
Let $\calV_0$~be the pseudo-variety generated by $\set{ \Syn(K) }{ K \in \calK }$.
For each $K \in \calK$,
it follows by Proposition~\ref{Prop: languages recognised by Syn(K)} that
all languages recognised by $\Syn(K)$ belong to~$\calK$.
By assumption, this implies that $\Syn(K) \in \calV$.
Consequently, we have $\calV_0 \subseteq \calV$.
Conversely, let $\frakA \in \calV$.
By assumption, every language recognised by~$\frakA$ belongs to~$\calK$.
(In particular, each such language has a syntactic algebra.)
Therefore, Lemma~\ref{Lem: algebras in a pseudo-variety} implies that $\frakA \in \calV_0$.
\end{proof}

\section{The profinitary term monad}   
\label{Sect:profinitary}

The goal of this section and the next one is to derive an axiomatisation of pseudo-varieties
in terms of systems of inequalities.
We start by defining the kind of terms allowed in our axioms.
The actual axiomatisation will then be presented in Section~\ref{Sect:axioms} below.
A~natural choice for the terms would be to take the elements of~$\bbM X$,
for some set~$X$ of `variables'.
But it turns out that this does not work.
To capture the restriction to \emph{finitary} $\bbM$-algebras, we have to use
a more general notion of a term.
The classic result by Reiterman~\cite{Reiterman82} characterises the pseudo-varieties
of finite semigroups as exactly those axiomatisable by a set of \emph{profinite} equations.
Analogously, we have to define \emph{profinitary $\bbM$-terms} for our version of this
theorem. For this we follow the material in~\cite{ChenAdamekMiliusUrbat16,UrbatChAdMi2017},
but with some adjustments that are needed to support infinitely many sorts.

To explain how we arrive at the definition below, let us collect our requirements on this
set of terms. We are looking for a functor~$\widehat\bbM$
mapping an (unordered) set~$X$ of `variables' to some set $\widehat\bbM X$ of `terms'.
These terms should generalise the ordinary terms from~$\bbM X$, i.e.,
we need an embedding $\iota : \bbM X \to \widehat\bbM X$.
Furthermore, we should be able to `evaluate' a term $t \in \widehat\bbM X$ in a given finitary
$\bbM$-algebra~$\frakA$ with respect to a given `variable assignment' $\beta : X \to A$.
Let us denote the resulting value by $\val(t;\beta)$.
For ordinary terms $t \in \bbM X$, this value should of course correspond to
the value of~$t$ in~$\frakA$. Thus,
\begin{align*}
  \val(\iota(t);\beta) = \pi(\bbM\beta(t))\,,
\end{align*}
where $\pi(\bbM\beta(t))$ is the canonical extension of $\beta : X \to A$ to $\bbM X \to A$.
Furthermore, $\val(t;\beta)$ should be compatible with morphisms of $\bbM$-algebras.
That is,
\begin{align*}
  \val(t;\varphi \circ \beta) = \varphi(\val(t;\beta))\,,
  \quad\text{for every morphism } \varphi : \frakA \to \frakB\,.
\end{align*}

This leads to the following construction.
We work in the category of all morphisms $\bbM X \to \frakA$.
In this category we consider the diagram of all $\beta : \bbM X \to \frakA$
where $\frakA$~is finitary and we take for $\iota : \bbM X \to \widehat\bbM X$
the limit. The morphisms $\widehat\bbM X \to \frakA$ of the corresponding limiting cone
can then be taken as our evaluation maps.
The formal construction is as follows.
\begin{defi}
Let $\calA \subseteq \Alg(\bbM)$ be a subcategory of finitary $\bbM$-algebras and $X$~a set.
We denote the comma category $(\bbM X \downarrow \Alg(\bbM))$ by~$\calC$,
the subcategory $(\bbM X \downarrow \calA)$ by~$\calC_0$, and the inclusion diagram by
$D : \calC_0 \to\nobreak \calC$.

(a)
We denote by $\iota_\calA : \bbM X \to \widehat\bbM_\calA X$ the limit $\iota_\calA := \lim D$ of~$D$,
and the limiting cone by $(\val_\calA({-};\beta))_{\beta \in \calC_0}$.
If $\calA$~is the category of all finitary $\bbM$-algebras, we drop the subscript
and simply write $\widehat\bbM$,~$\iota$, and~$\val({-};\beta)$.

(b) We turn $\widehat\bbM_\calA$~into a functor as follows. Given $f : X \to Y$,
the family $(\val({-};\beta \circ \bbM f))_\beta$ (where $\beta$ ranges over all
morphisms $\beta : \bbM Y \to \frakA \in \calA$) forms a cone from $\widehat\bbM X$ to~$D$.
As the cone $(\val({-};\beta))_\beta$ is limiting, there exists a unique function
$f' : \widehat\bbM X \to \widehat\bbM Y$ such that
\begin{align*}
  \val({-};\beta \circ \bbM f) = \val({-};\beta) \circ f'\,,
  \quad\text{for all } \beta : \bbM Y \to \frakA \in \calA\,.
\end{align*}
We set $\widehat\bbM f := f'$.
\markenddef
\end{defi}
\begin{Rem}
Another, more concise way to define~$\widehat\bbM$ is as the codensity monad
of the forgetful functor $\FAlg(\bbM) \to \Pos^\Xi$ which maps a finitary $\bbM$-algebra
to its universe, see~\cite{ChenAdamekMiliusUrbat16,UrbatChAdMi2017} for details.
\end{Rem}

Let us start by checking that $\widehat\bbM_\calA$~is well-defined and reasonably behaved.
\begin{lem}
The limit $\iota_\calA : \bbM X \to \widehat\bbM_\calA X$ exists.
\end{lem}
\begin{proof}
First, note that the category $\Pos^\Xi$ is complete.
By Proposition~4.3.1 of~\cite{Borceux94b}, this implies that so is $\Alg(\bbM)$.
Now, let $D : \calC_0 \to \calC$ be the diagram defining $\iota_\calA : \bbM X \to \widehat\bbM_\calA X$
and let $U : \calC \to \Alg(\bbM)$ be the forgetful functor mapping
$\beta : \bbM X \to \frakA$ to the codomain~$\frakA$.
As $\Alg(\bbM)$~is complete and $\calC_0$~is essentially small, $U \circ D$ has a limit~$\frakT$.
Let $(\lambda_\beta)_\beta$ the corresponding limiting cone.
As $(\beta)_\beta$~is a cone from~$\bbM X$ to $U \circ D$,
we obtain a unique morphism $\varphi : \bbM X \to \frakT$ such that
$\lambda_\beta \circ \varphi = \beta$, for all~$\beta$.
It is now straightforward to check that $\varphi : \bbM X \to \frakT$
is the limit of~$D$ and $(\lambda_\beta)_\beta$~is the corresponding limiting cone.
\end{proof}

We collect a few basic facts about the evaluation morphisms that will be useful
in the proofs below.
\begin{lem}\label{Lem: basic properties of val}
Let $\calA$~be a class of finitary\/ $\bbM$-algebras,\/ $\frakA,\frakB \in \calA$,
$\beta : \bbM X \to \frakA$, $\varphi : \frakA \to \frakB$, and $f : Y \to X$ morphisms,
and $s,t \in \widehat\bbM_\calA X$.
\begin{enuma}
\item $\val_\calA({-};\beta) \circ \iota_\calA = \beta$
\item $\varphi \circ \val_\calA({-};\beta) = \val_\calA({-};\varphi \circ \beta)$
\item $\val_\calA({-};\beta) \circ \widehat\bbM_\calA f = \val_\calA({-};\beta \circ \bbM f)$
\item If $\calA$~is closed under subalgebras then, for every $\hat s \in \widehat\bbM_\calA X$,
  there is some $s \in \bbM X$ with $\val_\calA(\hat s;\beta) = \beta(s)$.
\item $s \leq t \quad\iff\quad
  \val_\calA(s;\alpha) \leq \val_\calA(t;\alpha)\,,
  \quad\text{for all } \alpha : \bbM X \to \frakC \in \calA\,.$
\end{enuma}
\end{lem}
\begin{proof}
(a) By the definition of a cone, $\val_\calA({-};\beta)$ is a morphism from
$\iota_\calA : \bbM X \to \widehat\bbM_\calA X$ to $\beta : \bbM X \to \frakA$.
This is equivalent to~(a).

(b) In the comma category, $\varphi : \frakA \to \frakB$ corresponds to a morphism from
$\beta : \bbM X \to \frakA$ to $\varphi \circ \beta : \bbM X \to \frakB$.
Hence, (b)~holds again by definition of a cone.

(c) holds be definition of $\widehat\bbM_\calA f$.

(d) Let $\frakA_0$~be the subalgebra of~$\frakA$ induced by the range of~$\beta$,
let $i : \frakA_0 \to \frakA$ be the inclusion morphism,
and let $\beta_0 : \bbM X \to \frakA_0$ be the morphism such that $\beta = i \circ \beta_0$.
Note that $\frakA_0 \in \calA$ since $\calA$~is closed under subalgebras.
Fix $\hat s \in \widehat\bbM_\calA X$.
By~(a), we have $\rng \val_\calA({-};\beta_0) \supseteq \rng \beta_0$ which,
by surjectivity of~$\beta_0$, implies that the two ranges are in fact equal.
Hence, there is some $s \in \bbM X$ with $\beta_0(s) = \val_\calA(\hat s;\beta_0)$.
By~(b), it follows that
\begin{align*}
  \beta(s)
  = i(\beta_0(s))
  = i(\val_\calA(\hat s;\beta_0))
  = \val_\calA(\hat s;i\circ \beta_0)
  = \val_\calA(\hat s;\beta)\,.
\end{align*}

(e) One explicit way to define the limit $\widehat\bbM_\calA X$ is to take all
sequences $(a_\beta)_\beta$ indexed by morphisms $\beta : \bbM X \to \frakA$
satisfying
\begin{align*}
  a_\gamma = \varphi(a_\beta)\,,
  \quad\text{for all } \varphi : \frakA \to \frakB \text{ with }
  \gamma = \varphi \circ \beta\,.
\end{align*}
Then the function $\val_\calA({-};\beta)$ is simply the projection to component~$a_\beta$.
The ordering of~$\widehat\bbM_\calA X$ is taken to be largest relation such that all
projections $\val_\calA({-};\beta)$ are still monotone. That means that
\begin{align*}
  (a_\beta)_\beta \leq (b_\beta)_\beta
  \quad\iff\quad
  a_\beta \leq b_\beta\,, \quad\text{for all } \beta\,.
\end{align*}
\upqed
\end{proof}

\begin{cor}\label{Lem: val maps jointly monomorphic}
Let $X$~be a set and $f,g : C \to \widehat\bbM X$ functions.
\begin{align*}
  f = g \quad\iff\quad
  \val({-};\beta) \circ f = \val({-};\beta) \circ g\,,
  \quad\text{for all } \beta : \bbM X \to A\,.
\end{align*}
\end{cor}
\begin{proof}
This statement holds generally for all limits (see, e.g., Proposition~2.6.4 of~\cite{Borceux94a}).
For our special case, we can give a simple proof using Lemma~\ref{Lem: basic properties of val}\,(e).
By this lemma it follows that, for every $c \in C$,
\begin{align*}
  f(c) = g(c)
  \quad\iff\quad
  \val(f(c);\beta) \leq \val(f(c);\beta)\,,
  \quad\text{for all } \beta : \bbM X \to A\,.
\end{align*}
\upqed
\end{proof}

\begin{lem}
$\widehat\bbM_\calA$ is a functor and $\iota_\calA : \bbM \Rightarrow \widehat\bbM_\calA$
a natural transformation.
\end{lem}
\begin{proof}
To see that $\widehat\bbM_\calA$~is a functor, note that
the uniqueness of the function~$f'$ in the definition of~$\widehat\bbM_\calA f$
implies that $\widehat\bbM_\calA(f \circ g) = \widehat\bbM_\calA f \circ \widehat\bbM_\calA g$.

For the second claim, fix a function $f : X \to Y$.
For every $\beta : \bbM Y \to \frakA \in \calA$,
Lemma~\ref{Lem: basic properties of val}\,(c) implies that
\begin{align*}
  \val({-};\beta) \circ \widehat\bbM f \circ \iota
   = \val({-};\beta \circ \bbM f) \circ \iota
   = \beta \circ \bbM f
   = \val({-};\beta) \circ \iota \circ \bbM f\,.
\end{align*}
Consequently, it follows by Corollary~\ref{Lem: val maps jointly monomorphic} that
$\widehat\bbM f \circ \iota = \iota \circ \bbM f$
\end{proof}

\begin{lem}
$\widehat\bbM_\calA$~forms a monad where the unit map is $\varepsilon := \iota_\calA \circ \sing$
and the multiplication $\mu : \widehat\bbM_\calA \circ \widehat\bbM_\calA \Rightarrow \widehat\bbM_\calA$
is uniquely determined by the equations
\begin{align*}
  \val({-};\beta) \circ \mu = \val\bigl({-}; \pi \circ \bbM\val({-};\beta)\bigr)\,,
  \quad\text{for all } \beta\,.
\end{align*}
\end{lem}
\begin{proof}
To simplify notation, let us drop the subscript~$\calA$.
We define the multiplication $\mu : \widehat\bbM \circ \widehat\bbM \Rightarrow \widehat\bbM$ as follows.
For every morphism $\beta : \bbM X \to \frakA$ with $\frakA \in \calA$, we have
\begin{align*}
  \beta
  &= \beta \circ \pi \circ \sing \\
  &= \pi \circ \bbM\beta \circ \sing \\
  &= \pi \circ \bbM\val({-};\beta) \circ \bbM\iota \circ \sing \\
  &= \val\bigl({-};\pi \circ \bbM\val({-};\beta)\bigr) \circ \iota \circ \bbM\iota \circ \sing\,.
\end{align*}
Furthermore, for two such morphisms $\alpha : \bbM X \to \frakA$ and $\beta : \bbM X \to \frakB$
and a morphism $\varphi : \frakA \to \frakB$ with $\beta = \varphi \circ \alpha$, we have
\begin{align*}
  \varphi \circ \val\bigl({-};\pi \circ \bbM\val({-};\alpha)\bigr)
  &= \val\bigl({-};\varphi \circ \pi \circ \bbM\val({-};\alpha)\bigr) \\
  &= \val\bigl({-};\pi \circ \bbM\varphi \circ \bbM\val({-};\alpha)\bigr) \\
  &= \val\bigl({-};\pi \circ \bbM\val({-};\varphi \circ \alpha)\bigr) \\
  &= \val\bigl({-};\pi \circ \bbM\val({-};\beta)\bigr)\,.
\end{align*}
Consequently, the morphisms $\bigl(\val\bigl({-};\pi \circ \bbM\val({-};\beta)\bigr)\bigr)_\beta$
form a cone from
\begin{align*}
  \iota \circ \bbM\iota \circ \sing : \bbM X \to \widehat\bbM\widehat\bbM X
\end{align*}
to the diagram $(\bbM X \downarrow \calA)$.
As $\iota : \bbM X \to \widehat\bbM X$ is the limit of this cone, there exists a unique map
$\mu : \widehat\bbM\widehat\bbM X \to \widehat\bbM X$ such that
\begin{align*}
  &\mu \circ \iota \circ \bbM\iota \circ \sing = \iota \\
\prefixtext{and}
  &\val({-};\beta) \circ \mu = \val\bigl({-}; \pi \circ \bbM\val({-};\beta)\bigr)\,,
  \quad\text{for all } \beta\,.
\end{align*}
Note that the first of these equations follows from the second one since, for every~$\beta$,
\begin{align*}
  \val({-};\beta) \circ \mu \circ \iota \circ \bbM\iota \circ \sing
  &= \val\bigl({-}; \pi \circ \bbM\val({-};\beta)\bigr) \circ \iota \circ \bbM\iota \circ \sing \\
  &= \pi \circ \bbM\val({-};\beta) \circ \bbM\iota \circ \sing \\
  &= \pi \circ \bbM\beta \circ \sing \\
  &= \beta \circ \pi \circ \sing \\
  &= \beta \\
  &= \val({-};\beta) \circ \iota\,,
\end{align*}
which, by Corollary~\ref{Lem: val maps jointly monomorphic}, implies that
$\mu \circ \iota \circ \bbM\iota \circ \sing = \iota$.

Let us start by showing that these morphisms~$\mu$ form a natural transformation.
Hence, fix a function $f : X \to Y$. For every $\beta : \bbM Y \to \frakA$, we have
\begin{align*}
  \val({-};\beta) \circ \mu \circ \widehat\bbM\widehat\bbM f
  &= \val\bigl({-}; \pi \circ \bbM\val({-};\beta)\bigr) \circ \widehat\bbM\widehat\bbM f \\
  &= \val\bigl({-}; \pi \circ \bbM\val({-};\beta) \circ \bbM\widehat\bbM f\bigr) \\
  &= \val\bigl({-}; \pi \circ \bbM\val({-};\beta \circ \bbM f)\bigr) \\
  &= \val({-};\beta \circ \bbM f) \circ \mu \\
  &= \val({-};\beta) \circ \widehat\bbM f \circ \mu\,.
\end{align*}
By Corollary~\ref{Lem: val maps jointly monomorphic}, this implies that
$\mu \circ \widehat\bbM\widehat\bbM f = \widehat\bbM f \circ \mu$.

The fact that $\varepsilon := \iota \circ \sing$ is a natural transformation follows
immediately from the facts that $\iota$~and~$\sing$ are natural transformations.
It therefore remains to check the three axioms of a monad.
For every $\beta : \bbM X \to \frakA$, we have
\begin{align*}
  \val({-};\beta) \circ \mu \circ \varepsilon
  &= \val\bigl({-};\pi \circ \bbM\val({-};\beta)\bigr) \circ \iota \circ \sing \\
  &= \pi \circ \bbM\val({-};\beta) \circ \sing \\
  &= \val({-};\beta) \circ \pi \circ \sing \\
  &= \val({-};\beta)\,, \displaybreak[0]\\
  \val({-};\beta) \circ \mu \circ \widehat\bbM\varepsilon
  &= \val\bigl({-};\pi \circ \bbM\val({-};\beta)\bigr) \circ \widehat\bbM\varepsilon \\
  &= \val\bigl({-};\pi \circ \bbM\val({-};\beta) \circ \bbM\varepsilon\bigr) \\
  &= \val\bigl({-};\pi \circ \bbM\bigl(\val({-};\beta) \circ \iota \circ \sing\bigr)\bigr) \\
  &= \val\bigl({-};\pi \circ \bbM\bigl(\beta \circ \sing\bigr)\bigr) \\
  &= \val\bigl({-};\beta \circ \pi \circ \bbM\sing\bigr) \\
  &= \val({-};\beta)\,, \displaybreak[0]\\
\prefixtext{and}
  \val({-};\beta) \circ \mu \circ \widehat\bbM \mu
  &= \val\bigl({-};\pi \circ \bbM\val({-};\beta)\bigr) \circ \widehat\bbM \mu \\
  &= \val\bigl({-};\pi \circ \bbM\val({-};\beta) \circ \bbM\mu\bigr) \\
  &= \val\bigl({-};\pi \circ \bbM\bigl(\val({-};\beta) \circ \mu\bigr)\bigr) \\
  &= \val\bigl({-};\pi \circ \bbM\val\bigl({-};\pi \circ \bbM\val({-};\beta)\bigr)\bigr) \\
  &= \val\bigl({-};\pi \circ \bbM\val({-};\beta)\bigr) \circ \mu \\
  &= \val({-};\beta) \circ \mu \circ \mu\,.
\end{align*}
By Corollary~\ref{Lem: val maps jointly monomorphic}, this implies that
\begin{align*}
  \mu \circ \varepsilon = \id\,,\quad
  \mu \circ \widehat\bbM\varepsilon = \id\,,
  \qtextq{and}
  \mu \circ \widehat\bbM \mu = \mu \circ \mu\,.
\end{align*}
\upqed
\end{proof}

The next lemma states that, without loss of generality, we may assume that the morphisms
$\beta : \bbM X \to \frakA$ are all surjective. This will be convenient in some situations.
\begin{lem}\label{Lem: diagram cofiltered and surjective}
Let $X$~be a finite set and $\calA$~a pseudo-variety.
\begin{enuma}
\item $\calC_0 = (\bbM X \downarrow \calA)$ is cofiltered.
\item In the definition of\/~$\widehat\bbM_\calA X$, we can restrict the category~$\calC_0$
  to the surjective morphisms without changing the result.
\end{enuma}
\end{lem}
\begin{proof}
(a) There are two axioms to check.
First, let $\alpha : \bbM X \to \frakA$ and $\beta : \bbM X \to \frakB$ be two objects of~$\calC_0$.
We have to find some $\gamma : \bbM X \to \frakC$ and morphisms $\varphi : \gamma \to \alpha$
and $\psi : \gamma \to \beta$.
Set $\gamma := \langle\alpha,\beta\rangle : \bbM X \to \frakA \times \frakB$,
let $\frakC \subseteq \frakA \times \frakB$ be the subalgebra induced by $\rng \gamma$,
and let $\gamma_0 : \bbM X \to \frakC$ be the corestriction of~$\gamma$.
Note that $\frakC$~is finitely generated (by the image of~$X$).
Furthermore, for each sort $\xi \in \Xi$, the set $C_\xi \subseteq A_\xi \times B_\xi$ is finite.
Hence, $\frakC \in \calA$, $\gamma_0 \in \calC_0$, and we have morphisms
$p : \gamma_0 \to \alpha$ and $q : \gamma_0 \to \beta$, where $p : C \to A$ and $q : C \to B$
are the two projections.

For the second axiom, consider two morphisms $\varphi,\psi : \alpha \to \beta$ with
$\alpha : \bbM X \to \frakA$ and $\beta : \bbM X \to \frakB$ in~$\calC_0$.
The set
\begin{align*}
  C := \set{ a \in A }{ \varphi(a) = \psi(a) }
\end{align*}
induces a subalgebra of~$\frakA$ since, for $s \in \bbM C$, we have
\begin{align*}
  \varphi(\pi(s)) = \pi(\bbM\varphi(s)) = \pi(\bbM\psi(s)) = \psi(\pi(s))\,.
\end{align*}
For $x \in X$, we have
\begin{align*}
  \varphi(\alpha(x)) = \beta(x) = \psi(\alpha(x))\,,
\end{align*}
which implies that $\alpha[X] \subseteq C$. Hence, $\rng \alpha \subseteq C$.
Let $\frakD \subseteq \frakC$ be the subalgebra induced by $\rng \alpha$,
let $\alpha_0 : \bbM X \to \frakD$ be the corresponding corestriction of~$\alpha$,
and let $i : \frakC \to \frakA$ be the inclusion morphism.
Since $\frakD$~is finitely generated (by $\alpha_0[X]$), we have $\frakD \in \calA$.
Furthermore, $i : \alpha_0 \to \alpha$ satisfies $\varphi \circ i = \psi \circ i$.

(b)
Let $\calC_{00}$~be the full subcategory of $\calC_0 = (\bbM X \downarrow \calA)$
consisting of all morphisms $\beta : \bbM X \to \frakA$ that are surjective.
By Lemma~2.11.2 of~\cite{Borceux94a}, it is sufficient to prove the following two properties.
\begin{enumr}
\item Every $\beta \in \calC_0$ factorises through some $\beta_0 \in \calC_{00}$.
\item For all $\alpha,\alpha' \in \calC_{00}$, $\beta \in \calC_0$, and all morphisms
  $\varphi : \alpha \to \beta$ and $\varphi' : \alpha' \to \beta$, there is some
  $\gamma \in \calC_{00}$ with morphisms $\psi : \gamma \to \alpha$ and
  $\psi' : \gamma \to \alpha'$ such that $\varphi \circ \psi = \varphi' \circ \psi'$.
\end{enumr}

(i)
Given $\beta : \bbM X \to \frakA$, let $\frakA_0$~be the subalgebra of~$\frakA$
induced by $\rng \beta$, let $i : \frakA_0 \to \frakA$ be the inclusion function,
and $\beta_0 : \bbM X \to \frakA_0$ be the corestriction of~$\beta$.
Then $\beta = i \circ \beta_0$.
Since $\calA$~is closed under finitary subalgebras, we have $\frakA_0 \in \calA$ and
$\beta_0 \in \calC_0$.

(ii) Consider
$\alpha : \bbM X \to \frakA$, $\alpha' : \bbM X \to \frakA'$ in~$\calC_{00}$,
$\beta : \bbM X \to \frakB$ in~$\calC_0$, and
$\varphi : \alpha \to \beta$ and $\varphi' : \alpha' \to \beta$.
Let $\frakC$~be the subalgebra of $\frakA \times \frakA'$ induced by the range of
$\gamma := \langle\alpha,\alpha'\rangle : \bbM X \to \frakA \times \frakA'$.
Then $\frakC \in \calA$ and $\gamma \in \calC_{00}$.
The two projections $p : \frakC \to \frakA$ and $p' : \frakC \to \frakA'$
are morphisms of~$\calC_{00}$ satisfying $\varphi \circ p = \varphi' \circ p'$.
\end{proof}

To continue our investigation of the monad~$\widehat\bbM_\calA$,
we require some tools from topology.
We start with a variant of Stone duality for ordered topological spaces,
see, e.g., Chapter~11 of~\cite{DaveyPriestley02}.
\begin{defi}
(a) A \emph{Priestley space} consists of an ordered set~$A$ (one-sorted)
equipped with a topology that is compact and has the following separation property\?:
for every pair of elements $a,b \in A$ with $a \nleq b$, there exists a clopen set $C \subseteq A$
which is upwards-closed and contains~$a$, but not~$b$.
A \emph{morphism of Priestley spaces} is a function $f : A \to B$ that is monotone and continuous.
We denote the category of all Priestley spaces and their morphisms by $\PSp$.

(b) We denote by $\Dist$ the category of all distributive lattices (with top and bottom elements)
and all lattice homomorphisms (preserving top and bottom).
\markenddef
\end{defi}
\begin{Rem}
Every Priestley space is a Stone space, i.e., compact, Hausdorff, and totally disconnected.
\end{Rem}
\begin{thm}[Priestley]
The category $\PSp$ is equivalent to $\Dist^\op$.
\end{thm}
\noindent
To translate between these two categories we can map a Priestley space to the
lattices of its upwards-closed clopen subsets, and a distributive lattice
to the set of its prime filters (with a suitable topology).

We start by showing how to compute limits in $\PSp^\Xi$.
\begin{defi}
(a)
Let $(\mu_i)_{i \in I}$ be a cone where $\mu_i : A \to B_i$ and each $B_i$~is a topological space.
The \emph{cone topology} induced by $(\mu_i)_i$ is the topology on~$A$
which has a closed subbasis consisting of all sets of the form $\mu_i^{-1}[K]$
with $i \in I$ and $K \subseteq B_i$ closed.
If~$A$~is the limit of a diagram $D : I \to \Pos^\Xi$ and we do not specify a cone explicitly,
we will always consider the cone topology induced by the corresponding limiting cone.

(b)
For a functor $\bbM : \Pos^\Xi \to \Pos^\Xi$ for which we have defined a lifting to
$\PSp^\Xi \to \PSp^\Xi$,
we write $\PAlg(\bbM)$ for the category of $\bbM$-algebras in~$\PSp^\Xi$.
\markenddef
\end{defi}

\begin{Rem}
Let $X$~be a set and $\calA$~a pseudo-variety.
When we equip each $\frakA \in \calA$ with the discrete topology,
we can turn $\bbM X$ and $\widehat\bbM_\calA X$ into topological spaces where the topology
is induced by the cones $(\beta)_\beta$ and $(\val({-};\beta))_\beta$, respectively.
Then it follows by Lemma~\ref{Lem: basic properties of val}\,(d)
that the embedding $\iota_\calA : \bbM X \to \widehat\bbM_\calA X$
is dense with respect to these topologies.
In fact, the space $\widehat\bbM_\calA X$ can be seen as the topological completion of~$\bbM X$.
In particular, every element of $\widehat\bbM_\calA X$ is the limit of a suitable sequence
in $\bbM X$.
In the semigroup case, for instance, $\widehat\bbM_\calA X$ contains the
\emph{idempotent power}~$x^\pi$ which is the limit of the sequence $(x^{n!})_{n<\omega}$.
\end{Rem}

\begin{lem}\label{Lem: limit is in KHaus}
The forgetful functor\/ $\bbU : \PSp^\Xi \to \Pos^\Xi$ reflects limits.
More precisely, the limit\/ $\lim D$ of a diagram $D : I \to \PSp^\Xi$ is the space
obtained by equipping the set\/ $\lim {(\bbU \circ D)}$ with the cone topology.
\end{lem}
\begin{proof}
Let $A := \lim D$ and $B := \lim {(\bbU \circ D)}$ and let $(\lambda_i)_i$~and~$(\mu_i)_i$
be the corresponding limiting cones.
We start by showing that the cone topology on~$B$ is sort-wise Priestley.
Note that $B_\xi$~is the subset of $\prod_{i \in I} D_\xi(i)$ consisting of all
families $(a_i)_i$ such that $a_l = Df(a_k)$, for all $I$-morphisms $f : k \to l$.
Hence, $B_\xi = \bigcap_f H_f$ where
\begin{align*}
  H_f := \bigset{ \textstyle (a_i)_i \in \prod_i D_\xi(i) }{ Df(a_k) = a_l }\,,
  \quad\text{for } f : k \to l\,.
\end{align*}
Since, for distinct $a,b \in D_\xi(k)$, we can always find a clopen set~$C$ with
$a \in C$ and $b \notin C$, we can express~$H_f$ as the intersection of all sets
of the from
\begin{align*}
  \bigl(\mu_k^{-1}[(Df)^{-1}[C]] \cap \mu_l^{-1}[C]\bigr)
  \cup \bigl(\mu_k^{-1}[(Df)^{-1}[C']]
        \cap \mu_l^{-1}[C']\bigr)\,,
\end{align*}
where $C,C'$~range over all partitions of $D_\xi(k)$ into two clopen classes.
It follows that the sets~$H_f$ are all closed.
By the Theorem of Tychonoff, the product $\prod_i D_\xi(i)$ is compact.
Consequently, $B_\xi = \bigcap_f H_f$ is a closed subset of a compact space and, therefore,
also compact.

To show that the topology is Priestley, consider two distinct elements $a \nleq b$ in~$B$.
By the definition of the ordering of a limit in $\Pos^\Xi$, there exists an index $i \in I$
with $\mu_i(a) \nleq \mu_i(b)$.
Therefore we can find a clopen, upwards-closed set $C \subseteq D(i)$
such that $\mu_i(a) \in C$ and $\mu_i(b) \notin C$.
The preimage $C' := \mu_i^{-1}[C]$ is clopen in~$B$ and satisfies
$a \in C'$ and $b \notin C'$.
Suppose that $C'$~is not upwards-closed.
Then there are elements $c \leq d$ with $c \in C'$ and $d \notin C'$.
Consequently, $\mu_i(c) \leq \mu_i(d)$ and $\mu_i(c) \in C$ and $\mu_i(d) \notin C$.
This contradicts the fact that $C$~is upwards-closed.

We have shown that $B$~with the cone topology belongs to $\PSp^\Xi$.
Since $B$~is the limit in $\Pos^\Xi$, there exists a unique map
$f : A \to B$ (in $\Pos^\Xi$) such that $\lambda_i = \mu_i \circ f$, for all~$i$.
Similarly, there exists a unique morphism $g : B \to A$ of $\PSp^\Xi$
such that $\mu_i = \lambda_i \circ g$.
We can see that the function~$f$ is continuous as follows.
Let $C = \mu_i^{-1}[K]$ for a basic closed set $K \subseteq B$.
Then $f^{-1}[C] = (\mu_i \circ f)^{-1}[K] = (\lambda_i)^{-1}[K]$.
Hence, continuity of~$\lambda_i$ implies that the preimage $f^{-1}[C]$ is closed.

Consequently, we can applying the same universality argument two more times
to obtain $f \circ g = \id$ and $g \circ f = \id$.
Therefore, $B$~and $A$ with the cone topology are isomorphic as topological space.
\end{proof}

The following result is Corollary~1.1.6. of~\cite{RibesZalesskii10}.
It will be our key topology-based argument in the proof of Theorem~\ref{Thm: hat F for subclass}
below.
(In \cite{RibesZalesskii10} the result is proved for Stone spaces.
This is sufficient since Priestley spaces are Stone and, by the previous lemma,
the limits in both categories are the same.)
\begin{lem}\label{Lem: surjective map to cofiltered diagram}
Let $D : I \to \PSp^\Xi$ be a cofiltered diagram and $(\mu_i)_i$ a cone from
$A \in \PSp^\Xi$ to~$D$ where each $\mu_i : A \to D(i)$ is surjective.
The induced morphism $\varphi : A \to \lim D$ is surjective.
\end{lem}

Our main technical tool in the next section is the following natural transformation
relating the functors $\widehat\bbM_\calA$ and $\widehat\bbM_\calB$, for different classes
$\calA$~and~$\calB$. The important case below will be where $\calA$~is the pseudo-variety
under consideration and $\calB$~the class of all finitary $\bbM$-algebras.
\begin{thm}\label{Thm: hat F for subclass}
Let $\calA \subseteq \calB \subseteq \Alg(\bbM)$.
\begin{enuma}
\item There exists a unique morphism $\varrho : \widehat\bbM_\calB \Rightarrow \widehat\bbM_\calA$
  of monads that makes the following diagram commute, for all morphisms $\beta : \bbM X \to \frakA$
  where $\frakA \in \calA$ and $X$~is an unsorted set.

{\centering
\includegraphics{Abstract-8.mps}
%
%
%
%
\par}

\item If $\calA$~and~$\calB$ are closed under subalgebras and~$X$ is finite,
  then the induced morphism $\varrho_X : \widehat\bbM_\calB X \to \widehat\bbM_\calA X$ is surjective.
\end{enuma}
\end{thm}
\begin{proof}
(a) For a given set~$X$,
the family $(\val_\calB({-};\beta))_{\beta \in (\bbM X \downarrow \calA)}$ forms
a cone from $\widehat\bbM_\calB X$ to the diagram defining $\widehat\bbM_\calA X$.
As the cone $(\val_\calA({-};\beta))_{\beta \in (\bbM X \downarrow \calA)}$ is limiting,
there exists a unique map $\varrho_X : \widehat\bbM_\calB X \to \widehat\bbM_\calA X$
such that
\begin{align*}
  \val_\calA({-};\beta) \circ \varrho_X = \val_\calB({-};\beta)\,,
  \quad\text{for all } \beta : \bbM X \to \frakA\,.
\end{align*}

As the equation $\val_\calA({-};\beta) \circ \iota_\calA = \beta$ was already proved in
Lemma~\ref{Lem: basic properties of val}\,(a),
it therefore remains to prove that the family $\varrho := (\varrho_X)_X$ forms a morphism
of monads.
To see that it is a natural transformation, consider a function $f : X \to Y$.
Then
\begin{align*}
  \val_\calB({-};\beta) \circ \widehat\bbM_\calB f \circ \varrho
  &= \val_\calB({-};\beta \circ \bbM f) \circ \varrho \\
  &= \val_\calA({-};\beta \circ \bbM f) \\
  &= \val_\calA({-};\beta) \circ \widehat\bbM_\calA f
   = \val_\calB({-};\beta) \circ \varrho \circ \widehat\bbM_\calA f\,.
\end{align*}
By Corollary~\ref{Lem: val maps jointly monomorphic}, it follows that
$\widehat\bbM_\calB f \circ \varrho = \varrho \circ \widehat\bbM_\calA f$,
as desired.

To check the two axioms of a morphism of monads,
let $\mu_\calA$~and~$\varepsilon_\calA$ be the multiplication and unit map of~$\widehat\bbM_\calA$,
and $\mu_\calB$~and~$\varepsilon_\calB$ those of~$\widehat\bbM_\calB$.
For every $\beta : \bbM X \to \frakA$ with $\frakA \in \calA$, we have
\begin{align*}
  \val_\calA({-};\beta) \circ \varrho \circ \mu_\calB
  &= \val_\calB({-};\beta) \circ \mu_\calB \\
  &= \val_\calB\bigl({-};\pi \circ \bbM\val_\calB({-};\beta)\bigr) \\
  &= \val_\calA\bigl({-};\pi \circ \bbM\val_\calB({-};\beta)\bigr) \circ \varrho \\
  &= \val_\calA\bigl({-};\pi \circ \bbM\val_\calA({-};\beta) \circ \bbM\varrho\bigr) \circ \varrho \\
  &= \val_\calA\bigl({-};\pi \circ \bbM\val_\calA({-};\beta)\bigr) \circ \widehat\bbM\varrho
        \circ \varrho \\
  &= \val_\calA({-}; \beta) \circ \mu_\calA \circ \widehat\bbM\varrho \circ \varrho \\[1ex]
\prefixtext{and}
  \val_\calA({-};\beta) \circ \varrho \circ \varepsilon_\calB
  &= \val_\calA({-};\beta) \circ \varrho \circ \iota_\calB \circ \sing \\
  &= \val_\calB({-};\beta) \circ \iota_\calB \circ \sing \\
  &= \beta \circ \sing \\
  &= \val_\calA({-};\beta) \circ \iota_\calA \circ \sing
   = \val_\calA({-};\beta) \circ \varepsilon_\calA\,.
\end{align*}
By Corollary~\ref{Lem: val maps jointly monomorphic}, it follows that
$\varrho \circ \mu_\calB = \mu_\calA \circ \widehat\bbM\varrho \circ \varrho$
and $\varrho \circ \varepsilon_\calB = \varepsilon_\calA$.

(b)
To apply the topological machinery we have just set up, we translate the problem
into the category of Priestley spaces.
We equip each algebra in~$\calB$ with the discrete topology. As these algebras are
finitary, the resulting topologies are sort-wise Priestley.
According to Lemma~\ref{Lem: diagram cofiltered and surjective}\,(b),
we can define the limits $\widehat\bbM_\calA X$~and~$\widehat\bbM_\calB X$
in terms of only the surjective morphisms $\beta : \bbM X \to \frakA$
with $\frakA$~in $\calA$~or~$\calB$.
Furthermore, it follows by Lemma~\ref{Lem: limit is in KHaus} that
$\widehat\bbM_\calA X$~and~$\widehat\bbM_\calB X$
are also sort-wise Priestley spaces when equipped with the cone topology.
In addition, the limits in the category $\PSp^\Xi$ coincide with
$\widehat\bbM_\calA X$ and $\widehat\bbM_\calB X$.

Let $\Xi_0 \subseteq \Xi$ be the set of all sorts~$\xi$ such that $\bbM_\xi X \neq \emptyset$.
By Lemma~\ref{Lem: basic properties of val}\,(d), it follows that
these are exactly the same sorts~$\xi$ with $\widehat\bbM_{\calA,\xi} X \neq \emptyset$,
$\widehat\bbM_{\calB,\xi} X \neq \emptyset$,
and with $A_\xi \neq \emptyset$, for $\frakA \in \calA$.
Consequently, we can perform the rest of the proof in the category $\Pos^{\Xi_0}$.
By the definition of the cone topology, all the maps
$\val_\calA({-};\beta)$ and $\val_\calB({-};\beta)$ are continuous.
Furthermore, since we restricted the diagram to surjective maps~$\beta$,
$\val_\calB({-};\beta) \circ \iota = \beta$ implies that the value maps
$\val_\calB({-};\beta)$ are also surjective.
By Lemma~\ref{Lem: diagram cofiltered and surjective}\,(a),
$\widehat\bbM_\calA X$~is a cofiltered limit.
Consequently, we can use Lemma~\ref{Lem: surjective map to cofiltered diagram},
to show that $\varrho : \widehat\bbM_\calB X \to \widehat\bbM_\calA X$ is surjective.
\end{proof}

In the remainder of this section we will prove that
algebras of the form $\widehat\bbM_\calA|_\Delta X$ are what is called
\emph{finitely copresentable} (at least if $X$~and~$\Delta$ are finite).
This is another result requiring us to work with Priestly spaces.
Already the next lemma fails in~$\Set$ or~$\Pos$.
Unfortunately, it also does only hold for finitely many sorts.
\begin{defi}
An object~$A$ of a category~$\calC$ is \emph{finitely copresentable} if,
for every cofiltered diagram $D : I \to \calC$ with limit~$C$
and limiting cone $(\lambda_i)_{i \in I}$,
and for every morphism $f : C \to A$, there exists an index $k \in I$ and
an essentially unique morphism $g : D(k) \to A$ such that $f = g \circ \lambda_k$.
Essentially uniqueness here means that, if $g' : D(k) \to A$ is another morphism
with $f = g' \circ \lambda_k$, then there exists an $I$-morphisms $h : l \to k$ with
$g \circ Dh = g' \circ Dh$.
\markenddef
\end{defi}
\begin{lem}\label{Lem: finite sets finite copresentable}
Let $\Xi$~be a finite set of sorts.
Every finite Priestley space is finitely copresentable in $\PSp^\Xi$.
\end{lem}
\begin{proof}
First, note that the duality theorem implies that $\PSp^\Xi$ is equivalent
to $(\Dist^\Xi)^{\mathrm{op}}$.
Furthermore, the corresponding translation maps finite spaces to finite lattices.
Consequently it is sufficient to show that every finite lattice is
finitely presentable in $\Dist^\Xi$.
\end{proof}

It remains to transfer this result from $\PSp^\Xi$ to $\PAlg(\widehat\bbM)$.
We start with a variant of Remark~3.7\,(a) from~\cite{ChenAdamekMiliusUrbat16}.
As the proof is not included in the published version we have reproduced it here.
\begin{prop}\label{Prop: hat F preserves cofiltered limits}
Let $\calA$~be a class of finitary\/ $\bbM$-algebras.
The functor\/~$\widehat\bbM_\calA$ preserves cofiltered limits.
\end{prop}
\begin{proof}
We obtain a very concise proof if we employ a bit of category-theoretical machinery.
For the following material on \emph{right Kan extensions} and \emph{ends} we
refer the reader to Sections X.3~and~X.4 of~\cite{MacLane98}.
A~short introduction can also be found in Sections 1.1~and~1.2 of~\cite{Riehl14}.
We have tried to present the proof in a way that it should be intelligible
without knowledge of the actual definitions of these terms.

As already noted above one can define~$\widehat\bbM_\calA$ as the
\emph{codensity monad} associated with the forgetful functor
$I : \calA \to \PSp^\Xi$ maping an $\bbM$-algebra
$\frakA \in \calA$ to its universe~$A$ (equipped with the discrete topology).
By definition, this means that
\begin{align*}
  \widehat\bbM_\calA = \mathrm{Ran}_I\,I
\end{align*}
is the \emph{right Kan extension} of~$I$ along itself.
(This equation follows immediately from Section~X.3, Theorem~1 of~\cite{MacLane98}.)
Furthermore, (according to the dual of Section~X.4, Theorem~1 of~\cite{MacLane98})
we can compute a right Kan extension as
\begin{align*}
  (\mathrm{Ran}_I\,I)(X) = \int_{\frakA \in \calA} \PSp^\Xi(X, I\frakA) \pitchfork I\frakA
                         = \int_{\frakA \in \calA} \PSp^\Xi(X, A) \pitchfork A\,,
\end{align*}
where the integral sign is a certain kind of limit called an \emph{end}
and the \emph{power operator,} also called \emph{cotensor,}
${\pitchfork} : (\Set^\Xi)^\op \times \PSp^\Xi \to \PSp^\Xi$ is defined as iterated product
\begin{align*}
  X \pitchfork A = A^X = \prod_{x \in X} A\,.
\end{align*}
The characteristic property of the power is the equation
\begin{align*}
  \PSp^\Xi(B, X \pitchfork A) \cong \Set^\Xi(X, \PSp^\Xi(B,A))\,,
\end{align*}
for $X \in \Set^\Xi$ and $A,B \in \PSp^\Xi$.
For a fixed space~$A$ it follows that ${-} \pitchfork A : (\Set^\Xi)^\op \to \PSp^\Xi$
is the right adjoint of the hom-functor $\PSp^\Xi({-},A) : \PSp^\Xi \to (\Set^\Xi)^\op$.
This implies that ${-} \pitchfork Y$ preserves all limits
(cf. Proposition~3.2.2 of~\cite{Borceux94a}).
Furthermore, a space $A \in \PSp^\Xi$ is finitely copresentable if, and only if,
the hom-functor $\PSp^\Xi({-},A)$ preserves cofiltered surjective limits
(cf. Proposition~5.1.3 of~\cite{Borceux94b}).
As we have seen in Lemma~\ref{Lem: finite sets finite copresentable} that the universe
of a finitary $\bbM$-algebra is finitely copresentable in~$\PSp^\Xi$, it follows that
the composition $\PSp^\Xi({-},A) \pitchfork A$ preserves cofiltered limits,
for every $\frakA \in \calA$.

Given a cofiltered diagram $D : J \to \PSp^\Xi$, we therefore have
\begin{align*}
  \widehat\bbM_\calA(\lim D)
  &= (\mathrm{Ran}_I\,I)(\lim D) \\
  &= \int_{\frakA \in \calA} \PSp^\Xi(\lim_{j \in J} D(j), A) \pitchfork A \\
  &= \int_{\frakA \in \calA} \lim_{j \in J} {[\PSp^\Xi(D(j), A) \pitchfork A]} \\
  &= \lim_{j \in J} \int_{\frakA \in \calA} {[\PSp^\Xi(D(j), A) \pitchfork A]}
   = \lim_{j \in J} \widehat\bbM_\calA D(j)\,,
\end{align*}
where the fourth step follows by the fact that an end is a limit and limits commute.
\end{proof}

The following statement is the dual of Lemma~3.2 of~\cite{AdamekPorst04}.
\begin{lem}\label{Lem: algebras finitely copresentable}
Let $\calC$~be a category, $\bbM : \calC \to \calC$ a monad preserving cofiltered limits,
and\/ $\frakA$~an\/ $\bbM$-algebra with finitely copresentable domain~$A$.
Then\/ $\frakA$~is finitely copresentable in\/ $\Alg(\bbM)$.
\end{lem}
\begin{cor}\label{Cor: hat F_Delta X finitely copresentable}
Let $\Delta \subseteq \Xi$ be a finite set of sorts and $\calA$~a class of finitary $\bbM$-algebras.
For every finite set~$X \in \PSp^\Delta$,
the $\widehat\bbM_\calA|_\Delta$-algebra $\widehat\bbM_\calA|_\Delta X$
is finitely copresentable in $\PAlg(\widehat\bbM_\calA|_\Delta)$.
\end{cor}
\begin{proof}
By Lemma~\ref{Lem: finite sets finite copresentable}, the set~$X$ (with the discrete topology)
is finitely copresentable in $\PSp^\Delta$.
As we have shown in Propositon~\ref{Prop: hat F preserves cofiltered limits}
that $\widehat\bbM|_\Delta$~preserves cofiltered limits,
the claim therefore follows by Lemma~\ref{Lem: algebras finitely copresentable}.
\end{proof}

\section{Axiomatisations}   
\label{Sect:axioms}

After the preparations in the previous section we are now able to
define the type of inequalities we use to axiomatise pseudo-varieties
and to prove the characterisation theorem.
\begin{defi}
Let $X$~be a finite unordered set and $\calA$~a class of finitary $\bbM$-algebras.

(a)
An \emph{$\bbM$-inequality over~$X$} is a statement of the form $s \leq t$
with $s,t \in \widehat\bbM X$.

(b) A~finitary $\bbM$-algebra~$\frakA$ \emph{satisfies} an $\bbM$-inequality $s \leq t$
over~$X$ if
\begin{align*}
  \val(s;\beta) \leq \val(t;\beta)\,,
  \qquad\text{for all } \beta : \bbM X \to \frakA\,.
\end{align*}
We write $\frakA \models s \leq t$ to denote this fact.

(c) The \emph{$\bbM$-theory}~$\Th(\calA)$ of~$\calA$ is the set
of all $\bbM$-inequalities $s \leq t$ satisfied by every algebra in~$\calA$.
(We do \emph{not} fix the set~$X$ these inequalities are over.)

(d) A~set~$\Phi$ of $\bbM$-inequalities (possibly over several different sets~$X$)
\emph{axiomatises} the following subclass of~$\calA$.
\begin{align*}
  \Mod_\calA(\Phi) :=
    \bigset{ \frakA \in \calA }{ \frakA \models s \leq t \text{ for all } s \leq t \in \Phi }\,.
\end{align*}
\upqed
\markenddef
\end{defi}

Let us start with the following important property connecting
the theory of a class~$\calA$ to the morphism~$\varrho_\calA$
from Theorem~\ref{Thm: hat F for subclass}.
\begin{lem}\label{Lem: free algebra satisfies all axioms}
Let $\calA$~be a class of\/ $\bbM$-algebras, $X$~a finite set,
and $s \leq t$ an\/ $\bbM$-inequality over~$X$.
Then
\begin{align*}
  s \leq t \in \Th(\calA) \quad\iff\quad \varrho_\calA(s) \leq \varrho_\calA(t)\,,
\end{align*}
where $\varrho_\calA : \widehat\bbM \Rightarrow \widehat\bbM_\calA$ is the morphism
from Theorem~\ref{Thm: hat F for subclass}.
\end{lem}
\begin{proof}
By Lemma~\ref{Lem: basic properties of val}\,(e), we have
\begin{alignat*}{-1}
  &\frakA \models s \leq t\,, &&\quad\text{for all } \frakA \in \calA \\
\prefixtext{\iff\quad}
  &\val(s;\beta) \leq \val(t;\beta)\,, &&\quad\text{for all } \beta : \bbM X \to \frakA \in \calA \\
\prefixtext{\iff\quad}
  &\val_\calA(\varrho_\calA(s);\beta) \leq \val_\calA(\varrho_\calA(t);\beta)\,,
    &&\quad\text{for all } \beta : \bbM X \to \frakA \in \calA \\
\prefixtext{\iff\quad}
  &\varrho_\calA(s) \leq \varrho_\calA(t)\,.
\end{alignat*}
\upqed
\end{proof}

The easier direction is to show that every axiomatisable class is a pseudo-variety.
\begin{prop}\label{Prop: definable classes are pseudo-varieties}
Let $\calA$~be a pseudo-variety and $\Phi$~a set of\/ $\bbM$-inequalities.
Then\/ $\Mod_\calA(\Phi)$ is a pseudo-variety.
\end{prop}
\begin{proof}
We have to check three closure properties.
First, consider a finitary subalgebra~$\frakA$ of a product $\prod_{i \in I} \frakB^i$
with $\frakB^i \in \Mod_\calA(\Phi)$. Let $p_k : \prod_i B^i \to B^k$ be the projection.
For $s \leq t \in \Phi$ over~$X$ and $\beta : \bbM X \to \frakA$ it follows that
\begin{align*}
  p_k(\val(s;\beta))
  =    \val(s;p_k \circ \beta)
  \leq \val(t;p_k \circ \beta)
  =    p_k(\val(t;\beta))\,,
\end{align*}
where the second step follows from the fact that $\frakB^k \models s \leq t$.
As the ordering of the product is defined component-wise, this implies that
$\val(s;\beta) \leq \val(t;\beta)$. Consequently, $\frakA \in \Mod_\calA(\Phi)$.

Next, consider a quotient $q : \frakB \to \frakA$ with $\frakB \in \Mod_\calA(\Phi)$.
Fix $s \leq t \in \Phi$ over~$X$ and $\beta : \bbM X \to \frakA$.
Since $q$~is surjective,
we can use Lemma~\ref{Lem: FX projective}
to find some $\gamma : \bbM X \to \frakA$ with
$\beta = q \circ \gamma$.
Then
\begin{align*}
  \val(s;\beta)
   =    \val(s; q \circ \gamma)
   =    q(\val(s;\gamma))
   \leq q(\val(t;\gamma))
   =    \val(t; q \circ \gamma)
   =    \val(t; \beta)\,,
\end{align*}
where the third step follows by monotonicity of~$q$ and the fact that
$\frakB \models s \leq t$. Consequently, $\frakA \in \Mod_\calA(\Phi)$.

Finally, suppose that $\frakA$~is a sort-accumulation point of $\Mod_\calA(\Phi)$.
Fix $s \leq t \in \Phi$ over~$X$ and $\beta : \bbM X \to \frakA$.
We have to show that
\begin{align*}
  \val_\calA(s;\beta) \leq \val_\calA(t;\beta)\,.
\end{align*}
Suppose that $s,t \in \widehat\bbM_\xi X$ and let $\Delta \subseteq \Xi$ be a finite set of sorts
containing~$\xi$ and all sorts in~$X$.
By assumption, there is some algebra $\frakB \in \Mod_\calA(\Phi)$ and a surjective morphism
$\mu : \frakB|_\Delta \to \frakA|_\Delta$.
By Lemma~\ref{Lem: FX projective}, we can find a morphism
$\gamma : \bbM|_\Delta X \to \frakB|_\Delta$ with $\beta|_\Delta = \mu \circ \gamma$.
Since $\frakB \models s \leq t$ and $s,t \in \widehat\bbM|_\Delta X$,
we have (working in the category $\Pos^\Delta$)
\begin{align*}
  \val_\calA(s;\beta|_\Delta)
  &= \val_\calA(s;\mu \circ \gamma) \\
  &= \mu(\val_\calA(s;\gamma)) \\
  &\leq \mu(\val_\calA(t;\beta)) \\
  &= \val_\calA(t;\mu \circ \gamma)
  = \val_\calA(t;\beta|_\Delta)\,.
\end{align*}
Since
$\val_\calA({-};\beta|_\Delta) = \val_\calA({-};\beta) \restriction \widehat\bbM|_\Delta X$,
it follows that $\frakA \models s \leq t$.
\end{proof}

For the converse statement -- that every pseudo-variety is axiomatisable -- we start with
a proposition.
\begin{prop}\label{Prop: pseudo-variety consists of all quotients}
Let $\calV$~be a pseudo-variety. Then
\begin{align*}
  \calV =
    \set{ \frakA }
        { \frakA \text{ a finitary quotient of\/ } \widehat\bbM_\calV X
          \text{ for some finite set } X }\,.
\end{align*}
\end{prop}
\begin{proof}
$(\subseteq)$
Let $\frakA \in \calV$. As $\frakA$~is finitely generated, there exists
a surjective morphism $\beta : \bbM X \to \frakA$, for some finite set~$X$.
The claim follows since $\val({-};\beta) \circ \iota = \beta$ implies that that
$\val({-};\beta) : \widehat\bbM_\calV X \to \frakA$ is also surjective.

$(\supseteq)$
Let $\frakA$~be finitary and $\varphi : \widehat\bbM_\calV X \to \frakA$ surjective.
We have to show that $\frakA \in \calV$.
As $\calV$~is closed under sort-accumulation points,
it is sufficient to show that, for every finite set $\Delta \subseteq \Xi$
there is some algebra $\frakB \in \calV$ and a surjective morphism
$\frakB|_\Delta \to \frakA|_\Delta$.
Hence, fix $\Delta \subseteq \Xi$.
Note that, according to Lemma~\ref{Lem: limit is in KHaus}
we can define the set~$\widehat\bbM_\calV|_\Delta X$ as the limit of a cofiltered diagram
in $\PSp^\Delta$.
Furthermore, we have seen in Corollary~\ref{Cor: hat F_Delta X finitely copresentable} that
the $\widehat\bbM|_\Delta$-algebra $\widehat\bbM|_\Delta X$ is finitely copresentable in
$\PAlg(\widehat\bbM|_\Delta)$.
Therefore, there exists an algebra $\frakB \in \calV$ and morphisms
$\beta : \bbM X \to \frakB$ and $\mu : \frakB|_\Delta \to \frakA|_\Delta$ such that
$\varphi|_\Delta = \mu \circ \val({-};\beta)|_\Delta$.
Since $\varphi|_\Delta$~is surjective, so is~$\mu$. Consequently, $\frakA|_\Delta$~is
a quotient of~$\frakB|_\Delta$ and $\frakB \in \calV$.
\end{proof}
\begin{cor}\label{Cor: Mod(Th(V)) = V}
Let $\calV$~and~$\calW$ be pseudo-varieties.
\begin{enuma}
\item $\calV \subseteq \calW \quad\iff\quad \Th(\calV) \supseteq \Th(\calW)\,.$
\item $\Mod(\Th(\calV)) = \calV\,.$
\end{enuma}
\end{cor}
\begin{proof}
(a) $(\Rightarrow)$ follows immediately by definition. For $(\Leftarrow)$,
let $\varrho_{\calV,X} : \widehat\bbM X \to \widehat\bbM_\calV X$
and $\varrho_{\calW,X} : \widehat\bbM X \to \widehat\bbM_\calW X$
be the morphisms from Theorem~\ref{Thm: hat F for subclass}.
It follows by Lemma~\ref{Lem: free algebra satisfies all axioms} that
\begin{align*}
  \Th(\calW) \subseteq \Th(\calV)
  \qtextq{implies}
  \ker \varrho_{\calW,X} \subseteq \ker \varrho_{\calV,X}\,.
\end{align*}
Hence, we can use the Factorisation Lemma to find a morphism
$q_X : \widehat\bbM_\calW X \to \widehat\bbM_\calV X$
such that $\varrho_{\calV,X} = q_X \circ \varrho_{\calW,X}$.
By Theorem~\ref{Thm: hat F for subclass}, the morphism $\varrho_{\calV,X}$ is surjective.
Hence, so is~$q_X$. That means that $\widehat\bbM_\calV X$~is a quotient of~$\widehat\bbM_\calW X$.
Consequently, every quotient of~$\widehat\bbM_\calV X$ is also a quotient of~$\widehat\bbM_\calW X$
and it follows by Proposition~\ref{Prop: pseudo-variety consists of all quotients}
that $\calV \subseteq \calW$.

(b)
We have seen in Proposition~\ref{Prop: definable classes are pseudo-varieties}
that the class $\calW := \Mod(\Th(\calV))$ is a pseudo-variety.
We have to show that $\calV = \calW$.

$(\subseteq)$
Let $\frakA \in \calV$. Then we have $\frakA \models s \leq t$, for every $s \leq t$ in $\Th(\calV)$.
This implies that $\frakA \in \Mod(\Th(\calV)) = \calW$.

$(\supseteq)$
By~(a) it is sufficient to prove that $\Th(\calW) \supseteq \Th(\calV)$.
Hence, let $s \leq t$ be in $\Th(\calV)$.
Then $\frakA \models s \leq t$, for all $\frakA \in \Mod(\Th(\calV)) = \calW$,
which implies that $s \leq t$ belongs to $\Th(\calW)$.
\end{proof}

We are finally able to state our Reiterman theorem for pseudo-varieties of $\bbM$-algebras.
\begin{thm}
Let $\calF$~be the class of all finitary $\bbM$-algebras.
A class~$\calV$ is a pseudo-variety if, and only if,
it is of the form $\calV = \Mod_\calF(\Phi)$, for some set~$\Phi$ of $\bbM$-inequalities.
\end{thm}
\begin{proof}
$(\Leftarrow)$ was already proved in Proposition~\ref{Prop: definable classes are pseudo-varieties},
and $(\Rightarrow)$ follows by Corollary~\ref{Cor: Mod(Th(V)) = V}
since $\calV = \Mod_\calF(\Th(\calV))$.
\end{proof}

\section{Logics}   
\label{Sect:logic}

A major application of algebraic language theory consists in deriving criteria
for when a language is definable in a given logic.
In this section we will introduce an abstract framework covering
a large number of the logics used in practice.
Our focus will be on isolating some abstract properties of a logic
ensuring that the corresponding language family forms a variety and,
thus, fits into our framework.
In the next section we will then investigate what it means for a language
to be definable in a given logic.

Over the years several abstract logical frameworks have been in use, most of
them not developed enough to be ever published.
Among the major ones are the framework for abstract model theory proposed by
Barwise (see, e.g.,~\cite{BarwiseFeferman85}),
the notion of an \emph{abstract elementary class} introduced by Shelah
(see, e.g.,~\cite{Baldwin09}),
and the theory of \emph{institutions} developed by
Goguen and Burstall (see, e.g.,~\cite{Diaconescu08}).
The framework presented here is somewhat similar to the latter,
the main difference being that, in the following definition, we do not equip our class
of models with the structure of a category.
\begin{defi}
(a) A \emph{logic} is a triple $\langle L,\calM,{\models}\rangle$
consisting of a ($\Xi$-sorted, unordered) set~$L$ of \emph{formulae,}
a ($\Xi$-sorted, unordered) class~$\calM$ of \emph{models,}
and a \emph{satisfaction relation} ${\models} \subseteq \calM \times L$.
To keep notation light, we usually identify a logic with its set of formulae~$L$.

(b) A \emph{morphism} of logics
$\langle\lambda,\mu\rangle : \langle L,\calM,{\models}\rangle \to \langle L',\calM',{\models'}\rangle$
consists of two functions $\lambda : L \to L'$ and $\mu : \calM' \to \calM$ such that
\begin{align*}
  M' \models' \lambda(\varphi) \quad\iff\quad \mu(M') \models \varphi\,,
  \quad\text{for all } \varphi \in L_\xi \text{ and } M' \in \calM'_\xi\,.
\end{align*}
We denote the category of all logics and their morphisms by $\Log$.

(c) The \emph{$L$-theory} of a model $M \in \calM_\xi$ is
\begin{align*}
  \Th_L(M) := \set{ \varphi \in L_\xi }{ M \models \varphi }\,.
\end{align*}
For two models $M$~and~$N$, we define
\begin{alignat*}{-1}
  M &\sqsubseteq_L N &&\quad\defiff\quad &\Th_L(M) &\subseteq \Th_L(N)\,, \\
  M &\equiv_L N &&\quad\defiff\quad &\Th_L(M) &= \Th_L(N)\,.
\end{alignat*}

(d) The class of \emph{models} of a formula $\varphi \in L_\xi$ is the set
\begin{align*}
  \Mod(\varphi) := \set{ M \in \calM_\xi }{ M \models \varphi }\,.
\end{align*}
A~class $\calC \subseteq \calM_\xi$ is \emph{$L$-definable} if
$\calC = \Mod(\varphi)$, for some $\varphi \in L_\xi$.

(e) A logic~$L$ is \emph{lattice closed} if the collection of all $L$-definable
classes is closed under finite intersections and unions.
\markenddef
\end{defi}

As an example, for a given signature~$\Sigma$, a set~$X_1$ of first-order variables,
and a set $X_2$~of set variables, we can define \emph{monadic second-order logic} as
\begin{align*}
  \bigl\langle \MSO[\Sigma,X_1,X_2],\, \Alg[\Sigma,X_1,X_2],\, {\models}\bigr\rangle
\end{align*}
where $\MSO[\Sigma,X_1,X_2]$ is the set of all monadic second-order formulae over the
signature~$\Sigma$ with free first-order variables in~$X_1$ and
free monadic second-order variables in~$X_2$\?;
and $\Alg[\Sigma,X_1,X_2]$ is the set of all triples $\langle\frakA,\beta_1,\beta_2\rangle$
where $\frakA$~is a $\Sigma$-structure and $\beta_1 : X \to A$ and $\beta_2 : X \to \PSet(A)$
are variable assignments.

Every $\MSO$-interpretation~$\tau$ (from the signature~$\Sigma$ to~$\Gamma$) gives rise
to a morphism $\MSO[\Gamma,\emptyset,\emptyset] \to \MSO[\Sigma,\emptyset,\emptyset]$
since, according to the Interpretation Lemma, we can construct, for every formula
$\varphi \in \MSO[\Gamma,\emptyset,\emptyset]$,
some formula $\varphi^\tau \in \MSO[\Sigma,\emptyset,\emptyset]$ with
\begin{align*}
  \tau(\frakA) \models \varphi \quad\iff\quad \frakA \models \varphi^\tau,
  \quad\text{for all $\Sigma$-structures } \frakA\,.
\end{align*}

Let us isolate a few simple conditions for when a class of models is definable.
\begin{lem}\label{Lem: L-definable iff closed under L-equivalence}
Let $\langle L,\calM,{\models}\rangle$ and $\langle L',\calM',{\models}\rangle$ be lattice-closed
logics.
\begin{enuma}
\item A class $\calC \subseteq \calM_\xi$ is $L$-definable if, and only if,
  there exists a finite subset $\Delta \subseteq L_\xi$ such that
  \begin{align*}
    M \in \calC \qtextq{and} M \sqsubseteq_\Delta N \qtextq{implies} N \in \calC\,.
  \end{align*}
\item For sort-wise finite sets $\Delta \subseteq L$ and $\Delta' \subseteq L'$, and a function
  $f : \calM \to \calM'$ the following two statements are equivalent\?:
  \begin{enum1}
  \item $M \sqsubseteq_\Delta N \qtextq{implies} f(M) \sqsubseteq_{\Delta'} f(N)\,.$
  \item The preimage $f^{-1}[\calC']$ of a $\Delta'$-definable class $\calC' \subseteq \calM'$
    is $\Delta$-definable.
  \end{enum1}
\end{enuma}
\end{lem}
\begin{proof}
(a)
$(\Rightarrow)$
Let $\varphi \in L_\xi$ be a formula defining~$\calC$ and set $\Delta := \{\varphi\}$.
Suppose that $M \in \calC$ and $M \sqsubseteq_\Delta N$.
Then $M \models \varphi$, which implies that $N \models \varphi$.
Hence, $N \in \calC$.

$(\Leftarrow)$
Set
\begin{align*}
  \varphi := \Lor {\bigset{ \Land\Th_\Delta(M) }{ M \in \calC }}\,.
\end{align*}
Note that this disjunction is finite since there are only finitely many subsets of~$\Delta_\xi$.
For $N \in \calM$, it follows that
\begin{align*}
  N \models \varphi
  &\quad\iff\quad
  N \models \Land\Th_\Delta(M)\,, \quad\text{for some } M \in \calC \\
  &\quad\iff\quad
  M \sqsubseteq_\Delta N\,, \quad\text{for some } M \in \calC \\
  &\quad\iff\quad
  N \in \calC\,.
\end{align*}

(b) $(1) \Rightarrow (2)$
Suppose that $\calC' \subseteq \calM'$ is $\Delta'$-definable.
By~(a) (applied to the logic~$\Delta$ instead of~$L$), it is sufficient to show
that $M \in f^{-1}[\calC']$ and $M \sqsubseteq_\Delta N$ implies $N \in f^{-1}[\calC']$.
Hence, let $M \in f^{-1}[\calC']$ and $M \sqsubseteq_\Delta N$.
Then we have $f(M) \in \calC'$ and $f(M) \sqsubseteq_{\Delta'} f(N)$, by~(1).
Consequently, (a)~implies that $f(N) \in \calC'$, that is, $N \in f^{-1}[\calC']$.

$(2) \Rightarrow (1)$
Suppose that $M \sqsubseteq_\Delta N$.
To show that $f(M) \sqsubseteq_{\Delta'} f(N)$, we consider the class
$\calC'_M := \set{ H \in \calM' }{ f(M) \sqsubseteq_{\Delta'} H }$.
Note that $\calC'_M$~is $\Delta'$-definable by~(a).
By~(2), we know that $f^{-1}[\calC'_M]$ is $\Delta$-definable.
Consequently, it follows by~(a) that
$M \in f^{-1}[\calC'_M]$ and $M \sqsubseteq_\Delta N$ implies $N \in f^{-1}[\calC'_M]$,
that is, $f(M) \sqsubseteq_{\Delta'} f(N)$.
\end{proof}

\begin{lem}\label{Lem: morphisms of logics preserve definability}
Let $\langle L,\calM,{\models}\rangle$ and $\langle L',\calM',{\models}\rangle$ be logics and
$\mu : \calM' \to \calM$ a function.
The following statements are equivalent.
\begin{enum1}
\item There exists a function $\lambda : L \to L'$ such that
  $\langle\lambda,\mu\rangle : \langle L,\calM,{\models}\rangle \to \langle L',\calM',{\models}\rangle$
  is a morphism of logics.
\item If $\calC \subseteq \calM$ is $L$-definable, then $\mu^{-1}[\calC]$~is $L'$-definable.
\end{enum1}
If $L'$~is lattice closed, the following statement is also equivalent to those above.
\begin{enum1}[start=3]
\item For every sort-wise finite $\Delta \subseteq L$, there exists a sort-wise finite
  $\Delta' \subseteq L'$ such that
  \begin{align*}
    M \sqsubseteq_{\Delta'} N \qtextq{implies} \mu(M) \sqsubseteq_\Delta \mu(N)\,.
  \end{align*}
\end{enum1}
\end{lem}
\begin{proof}
(1)~$\Rightarrow$~(2)
Let $\varphi \in L$ be a formula defining~$\calC$. For $M \in \calM$, it follows that
\begin{align*}
  M \in \mu^{-1}[\calC]
  \quad\iff\quad
  \mu(M) \in \calC
  \quad\iff\quad
  \mu(M) \models \varphi
  \quad\iff\quad
  M \models \lambda(\varphi)\,.
\end{align*}
Thus, $\lambda(\varphi)$ defines $\mu^{-1}[\calC]$.

(2)~$\Rightarrow$~(1)
We define $\lambda : L \to L'$ as follows. For each $\varphi \in L$,
the class $\Mod(\varphi)$ is obviously $L$-definable.
By assumption it follows that the preimage $\mu^{-1}[\Mod(\varphi)]$
is defined by some formula $\varphi' \in L'$.
We set $\lambda(\varphi) := \varphi'$.

To see that $\langle\lambda,\mu\rangle$ is a morphism of logics, fix $M \in \calM'$ and
$\varphi \in L$. Then
\begin{align*}
  \mu(M) \models \varphi
  &\quad\iff\quad
  \mu(M) \in \Mod(\varphi) \\
  &\quad\iff\quad
  M \in \mu^{-1}[\Mod(\varphi)]
  \quad\iff\quad
  M \models \lambda(\varphi)\,.
\end{align*}

(1)~$\Rightarrow$~(3)
Given $\Delta \subseteq L$, we set $\Delta' := \lambda[\Delta]$.
Suppose that $M \sqsubseteq_{\Delta'} N$.
For every $\varphi \in \Delta$, we then have the implications
\begin{align*}
  \mu(M) \models \varphi
  \quad\Rightarrow\quad
  M \models \lambda(\varphi)
  \quad\Rightarrow\quad
  N \models \lambda(\varphi)
  \quad\Rightarrow\quad
  \mu(N) \models \varphi\,.
\end{align*}
Consequently, $\mu(M) \sqsubseteq_\Delta \mu(N)$.

(3)~$\Rightarrow$~(2)
Let $\calC \subseteq \calM$ be defined by the formula $\varphi \in L$.
By assumption, there is some sort-finite set $\Delta' \subseteq L'$ such that
\begin{align*}
  M \sqsubseteq_{\Delta'} N \qtextq{implies} \mu(M) \sqsubseteq_{\varphi} \mu(N)\,.
\end{align*}
We use Lemma~\ref{Lem: L-definable iff closed under L-equivalence}\,(a) to show
that $\mu^{-1}[\calC]$ is $\Delta'$-definable.
Hence, suppose that $M \sqsubseteq_{\Delta'} N$ and $M \in \mu^{-1}[\calC]$.
Then $\mu(M) \sqsubseteq_{\varphi} \mu(N)$ and, therefore,
\begin{align*}
  M \in \mu^{-1}[\calC]
  \quad\Rightarrow\quad
  \mu(M) \models \varphi
  \quad\Rightarrow\quad
  \mu(N) \models \varphi
  \quad\Rightarrow\quad
  N \in \mu^{-1}[\calC]\,.
\end{align*}
\upqed
\end{proof}

Here, we are mainly interested in logics whose set of models is of the form $\bbM\Sigma$
with $\Sigma \in \Alp$, as these can be used to define languages.
As with families of languages, we also need to consider families of logics indexed by
the alphabet used.
\begin{defi}
(a) A logic~$L$ is \emph{over} an alphabet~$\Sigma$ if its class of models is~$\bbM\Sigma$.

(b) A \emph{family of logics} is a functor $L : \Alp \to \Log$ such that
\begin{itemize}
\item for every alphabet~$\Sigma$, the image~$L[\Sigma]$ is a logic over~$\Sigma$,
\item for every function $f : \Sigma \to \Gamma$, the image~$L[f]$
  is a morphism $\langle\lambda,\mu\rangle$ of logics with $\mu = \bbM f$.
\end{itemize}

(c) Let $L$~be a family of logics.
A family of languages~$\calK$ is \emph{$L$-definable} if
\begin{align*}
  \calK_\xi[\Sigma] \subseteq \set{ \Mod(\varphi) }{ \varphi \in L_\xi[\Sigma] }\,,
  \quad\text{for all } \Sigma \in \Alp \text{ and } \xi \in \Xi\,.
\end{align*}

(d) Let $L$~be a family of logics and $A$~a finite ordered set.
We call a subset $K \subseteq \bbM A$ \emph{$L$-definable,}
if its unordered version $(\bbM\iota)^{-1}[K] \subseteq \bbM\bbV A$ is $L$-definable.

(e) A family~$L$ of logics is \emph{varietal} if the class of all $L$-definable
languages forms a variety of languages.

(f) We call a family of logics~$L$ \emph{(sort-wise) finite} if, for every alphabet~$\Sigma$,
the set of formulae $L[\Sigma]$ is (sort-wise) finite (up to logical equivalence).

(g) To keep notation light we will drop the signature from the notation
in cases where it is understood. Thus, we frequently write $L$ instead of $L[\Sigma]$.
\markenddef
\end{defi}

For the word monad $\bbM A := A^+$ and monadic second-order logic, we can define
a~family $\MSO$ that maps an alphabet~$\Sigma$ to the logic
$\MSO[\hat\Sigma,\emptyset,\emptyset]$ where
\begin{align*}
  \hat\Sigma := \{E,{\leq}\} \cup \set{ P_a }{ a \in \Sigma }
\end{align*}
is the signature consisting of the successor relation~$E$,
the ordering~$\leq$, and predicates~$P_a$ for all letters in~$\Sigma$.

As the notion of a logic is very general,
there is not much one can prove for an arbitrary logic.
To get non-trivial statements we need some kind of restriction.
As languages come equipped with a monadic composition operation,
it is natural to require our logics to be well-behaved under this form of composition.
This leads to the following definition.
\begin{defi}
A family~$L$ of logics is \emph{$\bbM$-compositional} if, for every finite subfamily
$\Phi \subseteq L$, there exists some sort-wise finite subfamily $\Phi \subseteq \Delta \subseteq L$
such that, for all alphabets~$\Sigma$,
the relation~$\sqsubseteq_{\Delta[\Sigma]}$ is a congruence ordering on~$\bbM\Sigma$.
\markenddef
\end{defi}
For instance, for words $u,u',v,v' \in \Sigma^+$ we have
\begin{align*}
  u \equiv_{\MSO_m} u' \qtextq{and} v \equiv_{\MSO_m} v'
  \qtextq{implies}
  uv \equiv_{\MSO_m} u'v'\,,
\end{align*}
where $\MSO_m$ denotes the set of $\MSO$-formulae of quantifier rank at most~$m$.
Consequently, $\MSO$~is $\bbM$-compositional for the word monad $\bbM A = A^+$.

The importance of $\bbM$-compositionality stems from the fact
the set of theories of such a logic forms an $\bbM$-algebra.
\begin{prop}\label{Prop: theory algebras}
A family of logics~$L$ is $\bbM$-compositional if, and only if, for every finite subfamily
$\Phi \subseteq L$, there exist
\begin{itemize}
\item a sort-wise finite subfamily $\Phi \subseteq \Delta \subseteq L$,
\item a functor $\Theta_\Delta : \Alp \to \Alg(\bbM)$, and
\item a surjective natural transformation
  $\theta_\Delta : (\bbM \restriction \Alp) \Rightarrow \Theta_\Delta$
\end{itemize}
such that
\begin{align*}
  s \sqsubseteq_{\Delta[\Sigma]} t \quad\iff\quad \theta_\Delta(s) \leq \theta_\Delta(t)\,,
  \quad\text{for all } s,t \in \bbM\Sigma\,.
\end{align*}
\end{prop}
\begin{proof}
$(\Leftarrow)$
Given $\Phi \subseteq L$, choose $\Phi \subseteq \Delta \subseteq L$ such that
$s \sqsubseteq_{\Delta[\Sigma]} t$ is equivalent to $\theta_\Delta(s) \leq \theta_\Delta(t)$.
This implies that the relation~$\sqsubseteq_{\Delta[\Sigma]}$ is equal to the kernel
of~$\theta_\Delta$, which is a congruence ordering.

$(\Rightarrow)$
Given $\Phi \subseteq L$, choose $\Phi \subseteq \Delta \subseteq L$ such that
$\sqsubseteq_\Delta$~is a congruence ordering.
We set $\Theta_\Delta\Sigma := \bbM\Sigma/{\sqsubseteq_{\Delta[\Sigma]}}$ and choose for
$\theta_\Delta : \bbM\Sigma \to \bbM\Sigma/{\sqsubseteq_{\Delta[\Sigma]}}$ the quotient map.
Given a function $f : \Sigma \to \Gamma$,
we define the morphism $\Theta_\Delta f : \Theta_\Delta\Sigma \to \Theta_\Delta\Gamma$
as follows.
We will show that $\ker \theta_\Delta \subseteq \ker {(\theta_\Delta \circ \bbM f)}$.
Then we can apply the Factorisation Lemma to obtain a unique morphism
$\psi : \Theta_\Delta\Sigma \to \Theta_\Delta\Gamma$ with
\begin{align*}
  \psi \circ \theta_\Delta = \theta_\Delta \circ \bbM f\,,
\end{align*}
and we set $\Theta_\Delta f := \psi$.

To prove the above claim, note that,
by definition of a family of logics, $L[f] = \langle\lambda,\bbM f\rangle$ is a morphism of logics.
Given $s,t \in \bbM\Sigma$ with $s \sqsubseteq_{\Delta[\Sigma]} t$ and a formula
$\varphi \in \Delta$, it therefore follows that
\begin{align*}
  \bbM f(s) \models \varphi
  \quad\Rightarrow\quad
  s \models \lambda(\varphi)
  \quad\Rightarrow\quad
  t \models \lambda(\varphi)
  \quad\Rightarrow\quad
  \bbM f(t) \models \varphi\,.
\end{align*}
Hence, $\bbM f(s) \sqsubseteq_{\Delta[\Gamma]} \bbM f(t)$, as desired.

From the definition, it immediately follows that $\theta_\Delta$~is a natural transformation
$\bbM \Rightarrow \Theta_\Delta$ since
$\Theta_\Delta f \circ \theta_\Delta = \theta_\Delta \circ \bbM f$.
Hence, it remains to show that $\Theta_\Delta$~is a functor.
Consider two functions $f : \Sigma \to \Gamma$ and $g : \Gamma \to \Upsilon$.
By the equation we have just established, we have
\begin{align*}
  \Theta_\Delta(g \circ f) \circ \theta_\Delta
   = \theta_\Delta \circ \bbM(g \circ f)
  &= \theta_\Delta \circ \bbM g \circ \bbM f \\
  &= \Theta_\Delta g \circ \theta_\Delta \circ \bbM f
   = \Theta_\Delta g \circ \Theta_\Delta f \circ \theta_\Delta\,.
\end{align*}
By surjectivity of~$\theta_\Delta$, this implies that
$\Theta_\Delta(g \circ f) = \Theta_\Delta g \circ \Theta_\Delta f$.
\end{proof}
It follows immediately from the definition that the theory algebra~$\Theta_\Delta\Sigma$
recognises every $\Delta$-definable language.
\begin{lem}\label{Lem: Theta recognises Delta-definable languages}
The morphism $\theta_\Delta : \bbM\Sigma \to \Theta_\Delta\Sigma$ recognises
every $\Delta$-definable language $K \subseteq \bbM\Sigma$.
\end{lem}
\begin{proof}
We claim that $K = \theta_\Delta^{-1}[P]$ where
\begin{align*}
  P := \set{ \sigma }{ \theta_\Delta(s) \leq \sigma \text{ for some } s \in K }\,.
\end{align*}
Clearly, $\theta_\Delta[K] \subseteq P$. Conversely, we have
\begin{alignat*}{-1}
  \theta_\Delta(t) \in P
  &\quad\Rightarrow\quad
  \theta_\Delta(s) \leq \theta_\Delta(t)\,, &&\quad\text{for some } s \in K\,, \\
  &\quad\Rightarrow\quad
  s \sqsubseteq_\Delta t\,, &&\quad\text{for some } s \in K\,, \\
  &\quad\Rightarrow\quad
  t \in K\,,
\end{alignat*}
where the last step follows by Lemma~\ref{Lem: L-definable iff closed under L-equivalence}\,(a).
\end{proof}

Next, let us take a look at the closure properties of definable languages.
Our first observation concerns closure under inverse relabellings,
which holds for every logic~$L$.
Then we show that $\bbM$-compositionality implies, but is slightly stronger than,
closure under derivatives.
\begin{lem}\label{Lem: closure under inverse relabellings}
Let $L$~be a family of logics. The class of $L$-definable languages is closed
under inverse relabellings.
\end{lem}
\begin{proof}
If $f : \Sigma \to \Gamma$ is a morphism of $\Alp$,
it follows by the definition of a family of logics that there is some function~$\lambda$ such that
$L[f] = \langle \lambda,\bbM f\rangle$ is a morphism of logics.
Consequently, we can use Lemma~\ref{Lem: morphisms of logics preserve definability}
to show that $(\bbM f)^{-1}[K]$ is $L$-definable, for every
$L$-definable language $K \subseteq \bbM\Gamma$.
\end{proof}
\begin{lem}\label{Lem: closure under adding contexts}
Let $L$~be an $\bbM$-compositional family of logics, and let $\Delta \subseteq L$ be a subfamily
such that $\sqsubseteq_\Delta$~is a congruence ordering.
Then
\begin{align*}
  s \sqsubseteq_\Delta t \qtextq{implies} p[s] \sqsubseteq_\Delta p[t]\,,
  \quad\text{for all } s,t \in \bbM\Sigma \text{ and } p \in \bbM(\Sigma+\Box)\,.
\end{align*}
\end{lem}
\begin{proof}
Set $p' := \bbM(\theta_\Delta + 1)(p)$.
By Lemma~\ref{Lem: applying context is monotonous}\,(a),
\begin{align*}
  a \leq b \qtextq{implies}
  p'[a] \leq p'[b]\,, \quad\text{for } a,b \in \Theta_\Delta\Sigma\,.
\end{align*}
Consequently,
\begin{align*}
  s \sqsubseteq_\Delta t
  &\quad\Rightarrow\quad
  \theta_\Delta(s) \leq \theta_\Delta(t) \\
  &\quad\Rightarrow\quad
  \theta_\Delta(p[s]) = p'(\theta_\Delta(s)) \leq p'(\theta_\Delta(t)) = \theta_\Delta(p[t]) \\
  &\quad\Rightarrow\quad
  p[s] \sqsubseteq_\Delta p[t]\,.
\end{align*}
\upqed
\end{proof}

Usually, the theory algebras~$\Theta_\Delta\Sigma$ from Proposition~\ref{Prop: theory algebras}
are not very well understood. (Otherwise, we would not need to introduce
a special algebraic framework to study definability questions.)
To shed a bit more light on what these algebras look like, we present an alternative
construction for the theory functor~$\Theta$.
\begin{defi}
Let~$L$ be a family of logics such that every $L$-definable language has a syntactic algebra.
The \emph{syntactic theory morphism} (for an alphabet~$\Sigma$) is
\begin{align*}
  \tilde\theta_L := \langle\syn_{\Mod(\varphi)}\rangle_{\varphi \in L[\Sigma]} :
  \bbM\Sigma \to \prod_{\varphi \in L[\Sigma]} \Syn(\Mod(\varphi))\,.
\end{align*}
\upqed
\markenddef
\end{defi}

\begin{lem}\label{Lem: kernel of tilde theta}
Let $L$~be a family of lattice-closed logics such that every $L$-definable language has
a syntactic algebra, and let $\Delta \subseteq L$ be sort-wise finite.
The following statements are equivalent.
\begin{enum1}
\item The class of $\Delta$-definable languages is closed under derivatives.
\item $s \sqsubseteq_\Delta t \quad\iff\quad \tilde\theta_\Delta(s) \leq \tilde\theta_\Delta(t)\,.$
\item $s \sqsubseteq_\Delta t \quad\Rightarrow\quad p[s] \sqsubseteq_\Delta p[t]\,,
  \quad\text{for all contexts } p\,.$
\end{enum1}
\end{lem}
\begin{proof}
$(1) \Leftrightarrow (3)$ follows by
Lemma~\ref{Lem: L-definable iff closed under L-equivalence}\,(b).

$(2) \Leftrightarrow (3)$
First, note that in~(3) we can replace the implication by an equivalence
since $p[s] \sqsubseteq_\Delta p[t]$ implies $s \sqsubseteq_\Delta t$,
if we choose for~$p$ the empty context~$\Box$.
Consequently, the equivalence of $(2)$~and~$(3)$ follows from the fact that,
for $s,t \in \bbM\Sigma$,
\begin{alignat*}{-1}
  &\tilde\theta_\Delta(s) \leq \tilde\theta_\Delta(t) \\
  \iff\quad& s \preceq_{\Mod(\varphi)} t\,,
     &&\quad\text{for all } \varphi \in \Delta \\
  \iff\quad& p[s] \models \varphi \ \Rightarrow\ p[t] \models \varphi\,,
     &&\quad\text{for all contexts } p \text{ and all } \varphi \in \Delta \\
  \iff\quad& p[s] \sqsubseteq_\Delta p[t]\,,
    &&\quad\text{for all contexts } p\,.
\end{alignat*}
\upqed
\end{proof}
\begin{thm}\label{Thm: characterisation of compositional logics}
Let $L$~be a family of lattice-closed logics such that every $L$-definable language has
a syntactic algebra. The following statements are equivalent.
\begin{enum1}
\item $L$~is $\bbM$-compositional.
\item For every finite $\Phi \subseteq L$, there exists a sort-wise finite
  $\Phi \subseteq \Delta \subseteq L$ such that the class of $\Delta$-definable languages
  is closed under derivatives.
\end{enum1}
\end{thm}
\begin{proof}
$(1)\Rightarrow(2)$
This follows immediately from Lemma~\ref{Lem: L-definable iff closed under L-equivalence}\,(b)
together with Lemma~\ref{Lem: closure under adding contexts}.

$(2)\Rightarrow(1)$
Given a subfamily $\Delta \subseteq L$ with the above closure properties,
it follows by Lemma~\ref{Lem: kernel of tilde theta}
that ${\sqsubseteq_\Delta} = \ker \tilde\theta_\Delta$.
In particular, ${\sqsubseteq_\Delta}$~is a congruence ordering.
\end{proof}

Apart from a criterion for $\bbM$-compositionality, this theorem also gives us an
explicit construction of the theory algebra $\Theta_\Delta\Sigma$ in
language-theoretic terms.
It therefore provides a more direct link between properties of a logic~$L$ and
properties of the class of $L$-definable languages.

\section{Definable algebras}   
\label{Sect:definable}

We have finally arrived at the central part of this article where we combine algebra
and logic.
It follows from Theorem~\ref{Thm: variety theorem} that, to every varietal logic~$L$,
there corresponds a unique pseudo-variety~$\calV$ of $\bbM$-algebras recognising the
family of $L$-definable languages.
We would like to use these $\bbM$-algebras to study the expressive power of our logic~$L$.
To do so, we need to know as much as possible about how the algebras in~$\calV$ look like.
Unfortunately, Theorem~\ref{Thm: variety theorem} does not tell us very much about that.
The following definition provides a slightly more concrete description.
\begin{defi}
Let $\frakA$~be an $\bbM$-algebra and $L$~a family of logics.

(a) A~finite subset $C \subseteq A$ is \emph{$L$-definably embedded} in~$\frakA$ if,
for every element $a \in A$, the preimage
\begin{align*}
  \pi^{-1}(\Aboveseg a) \cap \bbM C \text{ is $L$-definable.}
\end{align*}

(b) $\frakA$~is \emph{locally $L$-definable} if every finite subset $C \subseteq A$
is $L$-definably embedded in~$\frakA$.

(c) $\frakA$~is \emph{$L$-definable} if it is finitary and locally $L$-definable.
\markenddef
\end{defi}
\begin{Rem}
For the word functor $\bbM A = A^+$, every finite algebra (i.e., every finite semigroup)
is $\MSO$-definable since we can evaluate products in $\MSO$.
(Just guess a labelling that associates with every position the product of the corresponding
prefix.)

The same is true for the functor
$\bbM\langle A_1,A_\infty\rangle = \langle A_1^+,\ A_1^+A_\infty \cup A_1^\omega\rangle$
for infinite words and for the functor for finite trees.
(For the former, one can use a reduction to the semigroup case
via a simple application of the Theorem of Ramsey\?; for the latter,
one can compute the product of a tree bottom-up similarly to the semigroup case.)

For infinite trees the situation is more complicated\?:
Boja\'nczyk and Klin~\cite{BojanczykKlin18} have constructed an example of a finitary
algebra that is not $\MSO$-definable.
\end{Rem}

If our logic~$L$ is sufficiently well-behaved, it immediately follows from this definition
that $L$-definable algebras only recognise $L$-definable languages.
(The converse, that every $L$-definable language is recognised by some $L$-definable algebra,
is harder to prove. We will do so later in this section.)
Note that this correspondence, besides being trivial, is also not that useful for
understanding the expressive power of~$L$, as the definition makes essential use of
$L$-definability. But the above definition can serve as a starting point for deriving
more useful descriptions -- that of course will be specific to the logic in question.

Before proving that the $L$-definable algebras are exactly those that only recognise
$L$-definable languages, let us start by looking at definably embedded sets.
\begin{lem}\label{Lem: characterisation of regular embedded}
Let\/ $\frakA$~be a finitary\/ $\bbM$-algebra and $L$~a family of lattice-closed logics.
A~finite set $C \subseteq A$ is $L$-definably embedded in\/~$\frakA$ if, and only if,
there exists a sort-wise finite set $\Delta \subseteq L[C]$ such that
\begin{align*}
  s \sqsubseteq_\Delta t \qtextq{implies} \pi(s) \leq \pi(t)\,,
  \quad\text{for all } s,t \in \bbM C\,.
\end{align*}
\end{lem}
\begin{proof}
$(\Rightarrow)$
Choose a function (of unordered sets) $\vartheta : A \to L$ such that the
formula~$\vartheta(a)$ defines the set $\pi^{-1}(\Aboveseg a) \cap \bbM C$.
We claim that the set $\Delta := \rng \vartheta$ has the desired properties.
Consider $s,t \in \bbM C$ with $s \sqsubseteq_\Delta t$.
Then
\begin{align*}
  s \models \vartheta(\pi(s)) \qtextq{implies} t \models \vartheta(\pi(s))\,.
\end{align*}
Hence, $\pi(t) \geq \pi(s)$.

$(\Leftarrow)$
Let $\Delta$~be a sort-wise finite set such that
\begin{align*}
  s \sqsubseteq_\Delta t \qtextq{implies} \pi(s) \leq \pi(t)\,,
  \quad\text{for } s,t \in \bbM C\,.
\end{align*}
Given $a \in A_\xi$, we set $K := \pi^{-1}(\Aboveseg a) \cap \bbM C$.
It follows that
\begin{align*}
  s \in K \qtextq{and} s \sqsubseteq_{\Delta_\xi} t \qtextq{implies} t \in K\,.
\end{align*}
Hence, we can use Lemma~\ref{Lem: L-definable iff closed under L-equivalence}\,(a) to show
that $K$~is $L$-definable.
\end{proof}

In general, the closure properties of definably embedded sets are rather weak.
To make them better behaved we have to impose some restriction on the logic~$L$.
\begin{lem}\label{Lem: generated subsets are definable embedded}
Let $\frakA$~be an $\bbM$-algebra, $L$~a family of logics, and $C \subseteq A$ finite set that is
$L$-definably embedded in~$\frakA$.
\begin{enuma}
\item Every subset of~$C$ is $L$-definably embedded in~$\frakA$.
\item If the class of $L$-definable languages is closed under inverse morphisms,
  every finite subset $D \subseteq \langle C\rangle$ is $L$-definably embedded in~$\frakA$,
  where $\langle C\rangle$ denotes the subalgebra generated by~$C$.
\end{enuma}
\end{lem}
\begin{proof}
(a) Fix $D \subseteq C$ and let $i : D \to C$ be the inclusion map. Then
\begin{align*}
  \pi^{-1}(\Aboveseg a) \cap \bbM D
  = \bigl(\pi^{-1}(\Aboveseg a) \cap \bbM C\bigr) \cap \bbM D
  = (\bbM i)^{-1}\bigl(\pi^{-1}(\Aboveseg a) \cap \bbM C\bigr)
\end{align*}
is the image of an $L$-definable set under an inverse relabelling
and, therefore, $L$-definable by Lemma~\ref{Lem: closure under inverse relabellings}.

(b)
By~(a) it is sufficient to consider the case where $D = \langle C\rangle$.
For every $d \in D$, we can find an element $f(d) \in \bbM C$ such that $\pi(f(d)) = d$.
This defines a function~$f$ with $\pi \circ f = \id_D$.
But note that, in general, $f$~is not monotone.
Thus, we only obtain a function $f : \bbV D \to \bbM C$.
Let $\varphi : \bbM\bbV D \to \bbM C$ be the (unique) extension of
$f : \bbV D \to \bbM C$ to~$\bbM\bbV D$.
Let $\delta : \bbM\circ\bbV \Rightarrow \bbV\circ\bbM$ be the natural isomorphism
obtained by the fact that $\bbM$~is order agnostic.
Then

\noindent
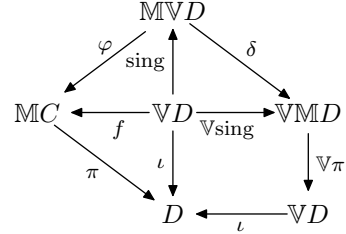
\begin{minipage}[t]{0.6\textwidth}
\vspace*{0pt}%
\begin{align*}
  \iota \circ \bbV\pi \circ \delta \circ \sing
  &= \iota \circ \bbV\pi \circ \bbV\sing \\
  &= \iota \\
  &= \pi \circ f \\
  &= \pi \circ \varphi \circ \sing\,.
\end{align*}%
\end{minipage}%
\begin{minipage}[t]{0.4\textwidth}
\vspace*{\abovedisplayskip}%
\includegraphics{Abstract-9.mps}
%
%
%
%
\end{minipage}%
\par\bigskip

\noindent
Note that $\iota \circ \bbV\pi \circ \delta$ and $\pi \circ \varphi$ are both morphisms
of $\bbM$-algebras. The latter is a composition of two morphism $\varphi : \bbM\bbV D \to \bbM C$
and $\pi : \bbM C \to D$.
Concerning the former, we have
\begin{align*}
  (\iota \circ \bbV\pi \circ \delta) \circ \pi
  &= (\iota \circ \bbV\pi \circ \delta) \circ \Flat \\
  &= \pi \circ \iota \circ \delta \circ \Flat \\
  &= \pi \circ \bbM\iota \circ \Flat \\
  &= \pi \circ \Flat \circ \bbM\bbM\iota \\
  &= \pi \circ \Flat \circ \bbM(\iota \circ \delta) \\
  &= \pi \circ \bbM\pi \circ \bbM\iota \circ \bbM\delta
   = \pi \circ \bbM(\iota \circ \bbV\pi \circ \delta)\,.
\end{align*}
As morphisms from a free $\bbM$-algebra are uniquely determined by their restriction
to $\rng \sing$, we therefore have $\iota \circ \bbV\pi \circ \delta = \pi \circ \varphi$.
For $a \in A$, it follows that
\begin{align*}
  (\bbM\iota)^{-1}[\pi^{-1}[\Aboveseg a] \cap \bbM D]
  &= (\pi \circ \bbM\iota)^{-1}[\Aboveseg a] \cap \bbV\bbM D \\
  &= (\pi \circ \iota \circ \delta)^{-1}[\Aboveseg a] \cap \bbV\bbM D \\
  &= (\iota \circ \bbV\pi \circ \delta)^{-1}[\Aboveseg a] \cap \bbV\bbM D \\
  &= (\pi \circ \varphi)^{-1}[\Aboveseg a] \cap \bbV\bbM D
   = \varphi^{-1}[\pi^{-1}[\Aboveseg a] \cap \bbM C]\,,
\end{align*}
which is the image of an $L$-definable language under an inverse morphism.
The way we defined $L$-definability for ordered sets, this implies that
$\pi^{-1}[\Aboveseg a] \cap \bbM D$ is also $L$-definable.
\end{proof}

It follows immediately from the definition that an $L$-definable algebra only recognises
$L$-definable languages. We start with a slightly more precise statement.
\begin{thm}\label{Thm: L-definable algebras recognise L-definable languages}
Let $L$~be a varietal family of logics.
A~finitary $\bbM$-algebra~$\frakA$ is $L$-definable if, and only if,
every language recognised by~$\frakA$ is $L$-definable.
\end{thm}
\begin{proof}
$(\Leftarrow)$
If some finite subset $C \subseteq A$ is not $L$-definably embedded,
we can find an element $a \in A$ such that
the preimage $K := \pi^{-1}[\Aboveseg a] \cap \bbM C$ is not $L$-definable.
Thus, the restriction $\pi \restriction \bbM C : \bbM C \to \frakA$ of the
product is a morphism recognising the non-$L$-definable language~$K$.

$(\Rightarrow)$
Let $\varphi : \bbM\Sigma \to \frakA$ be a morphism and $P \subseteq A_\xi$ an upwards closed set.
By assumption, the set $C := \rng (\varphi \circ \sing)$ is $L$-definably embedded in~$\frakA$.
For every $a \in A_\xi$, we can therefore fix an $L$-formula~$\vartheta_a$ defining the set
$\pi^{-1}(\Aboveseg a) \cap \bbM C$. Setting $\varphi_0 := \varphi \circ \sing$
it follows that
\begin{align*}
  t \in \varphi^{-1}[P]
  \quad\iff\quad
  \varphi(t) \in P
  \quad\iff\quad
  \pi(\bbM\varphi_0(t)) \in P
  \quad\iff\quad
  \bbM\varphi_0(t) \models \Lor_{a \in P} \vartheta_a\,.
\end{align*}
(For the last step, note that $\bbM\varphi_0(t) \in \bbM C$.)
As the $L$-definable languages are closed under inverse morphisms,
we can find a formula $\chi \in L$ such that
\begin{align*}
  \bbM\varphi_0(t) \models \Lor_{a \in P} \vartheta_a
  \quad\iff\quad
  t \models \chi\,.
\end{align*}
Consequently, $\chi$~defines $\varphi^{-1}[P]$.
\end{proof}

Obvious candidates for $L$-definable algebras are those of the form~$\Theta_\Delta\Sigma$
from Proposition~\ref{Prop: theory algebras}.
Below we will characterise under which conditions these are $L$-definable.
The proof rests on the following technical result.
\begin{lem}
Let $L$~be an $\bbM$-compositional family of lattice-closed logics.
For every sort-wise finite set $\Delta \subseteq L$ such that $\sqsubseteq_\Delta$~is
a congruence ordering,
the set $\rng (\theta_\Delta \circ \sing)$ is $L$-definably embedded in~$\Theta_\Delta\Sigma$.
\end{lem}
\begin{proof}
Set $f := \theta_\Delta \circ \sing : \Sigma \to \Theta_\Delta\Sigma$ and $C := \rng f$.
Choose a right inverse $g : \bbV C \to \bbV\Sigma$ of $\bbV f : \bbV\Sigma \to \bbV C$ and
let $\pi_0 : \bbM C \to \Theta_\Delta\Sigma$ be the restriction of the product
of~$\Theta_\Delta\Sigma$ to~$\bbM C$.
Recall the natural transformations $\delta$~and~$\iota$ from Definition~\ref{Def: bbV and iota}.
Then
\begin{align*}
  \bbV\pi_0 \circ \delta
  &= \bbV\pi_0 \circ \delta \circ \bbM(\bbV f \circ g) \\
  &= \bbV\pi \circ \delta \circ \bbM(\bbV\theta_\Delta \circ \bbV\sing \circ g) \\
  &= \bbV\pi \circ \delta \circ \bbM\bbV\theta_\Delta \circ \bbM\bbV\sing \circ \bbM g \\
  &= \bbV\pi \circ \bbV\bbM\theta_\Delta \circ \bbV\bbM\sing \circ \delta \circ \bbM g \\
  &= \bbV\theta_\Delta \circ \bbV\pi \circ \bbV\bbM\sing \circ \delta \circ \bbM g
   = \bbV\theta_\Delta \circ \delta \circ \bbM g\,.
\end{align*}
To show that $C$~is $L$-definably embedded in~$\Theta_\Delta\Sigma$,
we fix an element $a \in \Theta_\Delta\Sigma$. Then
\begin{align*}
  (\bbM\iota)^{-1}[\pi_0^{-1}[\Aboveseg a]]
  &= (\pi_0 \circ \bbM\iota)^{-1}[\Aboveseg a] \\
  &= (\pi_0 \circ \iota \circ \delta)^{-1}[\Aboveseg a] \\
  &= (\iota \circ \bbV\pi_0 \circ \delta)^{-1}[\Aboveseg a] \\
  &= (\iota \circ \bbV\theta_\Delta \circ \delta \circ \bbM g)^{-1}[\Aboveseg a] \\
  &= (\theta_\Delta \circ \iota \circ \delta \circ \bbM g)^{-1}[\Aboveseg a]
   = (\bbM\iota \circ \bbM g)^{-1}[\theta_\Delta^{-1}[\Aboveseg a]]\,.
\end{align*}
We can use Lemma~\ref{Lem: L-definable iff closed under L-equivalence} to show that
the language $K := \theta^{-1}_\Delta[\Aboveseg a]$ is $L$-definable\?:
given $s \sqsubseteq_\Delta t$ and $s \in K$, we have
$\theta_\Delta(s) \leq \theta_\Delta(t)$ and $\theta_\Delta(s) \geq a$.
Thus, $\theta_\Delta(t) \geq a$, i.e., $t \in K$.

To conclude the proof, note that we have shown in
Lemma~\ref{Lem: closure under inverse relabellings} that the class of $L$-definable languages
is closed under inverse relabellings.
Consequently, the language
$(\bbM\iota)^{-1}[\pi_0^{-1}[\Aboveseg a]] = (\bbM(\iota \circ g))^{-1}[K]$ is $L$-definable and,
therefore, so is
$\pi_0^{-1}[\Aboveseg a] = \pi^{-1}[\Aboveseg a] \cap \bbM C$.
\end{proof}

\begin{thm}\label{Thm: theory algebras are L-definable}
Let $L$~be an $\bbM$-compositional family of lattice-closed logics.
The following statements are equivalent.
\begin{enum1}
\item $L$~is varietal.
\item Every algebra of the form $\Theta_\Delta\Sigma$ is $L$-definable.
\item The class of $L$-definable languages is closed under inverse morphisms.
\end{enum1}
\end{thm}
\begin{proof}
(1)~$\Rightarrow$~(3) is trivial.

(3)~$\Rightarrow$~(1) Closure under inverse morphisms and under finite unions and intersections
holds by assumption, while closure under derivatives
was shown in Theorem~\ref{Thm: characterisation of compositional logics}.

(3)~$\Rightarrow$~(2)
First, note that $\Theta_\Delta\Sigma$ is finitary\?:
it is generated by the finite set $\rng (\theta_\Delta \circ \sing)$
and, for every sort $\xi \in \Xi$,
there are only finitely many elements in $(\Theta_\Delta\Sigma)_\xi$,
since there are only finitely many subsets of $\Delta_\xi[\Sigma]$.

Hence, it remains to show that every finite subset is $L$-definably embedded.
Let $D \subseteq \Theta_\Delta\Sigma$ be finite.
We have shown in the preceding lemma that the set $C := \rng (\theta_\Delta \circ \sing)$
is $L$-definably embedded in $\Theta_\Delta\Sigma$.
As~$C$~generates $\Theta_\Delta\Sigma$, we have $D \subseteq \langle C\rangle$.
Consequently, it follows by Lemma~\ref{Lem: generated subsets are definable embedded}\,(b)
that $D$~is also $L$-definably embedded.

(2)~$\Rightarrow$~(3)
Fix a morphism $\varphi : \bbM\Sigma \to \bbM\Gamma$ and an $L$-definable language
$K \subseteq \bbM_\xi\Gamma$.
We will construct two sort-wise finite sets $\Delta,\Delta' \subseteq L$ such that
$K$~is $\Delta[\Gamma]$-definable and
\begin{align*}
  s \sqsubseteq_{\Delta'[\Sigma]} t \qtextq{implies}
  \varphi(s) \sqsubseteq_{\Delta[\Gamma]} \varphi(t)\,,
  \quad\text{for all } s,t \in \bbM_\xi\Sigma\,.
\end{align*}
Then it follows by Lemma~\ref{Lem: L-definable iff closed under L-equivalence}\,(b)
that $\varphi^{-1}[K]$ is $L$-definable.

Hence, it remains to find the sets $\Delta$~and~$\Delta'$.
As $L$~is $\bbM$-compositional, we can
choose a sort-wise finite subset $\Delta \subseteq L$ such that $K$~is $\Delta[\Gamma]$-definable
and $\sqsubseteq_\Delta$~is a congruence ordering.
Set
\begin{align*}
  f := \theta_\Delta \circ \varphi \circ \sing : \Sigma \to \Theta_\Delta\Gamma
  \qtextq{and}
  C := \rng f\,.
\end{align*}
By assumption, $C$~is $L$-definably embedded in~$\Theta_\Delta\Gamma$.
We can therefore use Lemma~\ref{Lem: characterisation of regular embedded} to find a sort-wise 
finite subset $\Psi \subseteq L$ such that
\begin{align*}
  u \sqsubseteq_\Psi v
  \qtextq{implies}
  \pi(u) \leq \pi(v)\,,
  \quad\text{for all } u,v \in \bbM C\,.
\end{align*}
Let $\Delta_0 \subseteq \Delta$ be the (finite) subset of all formulae whose sort is equal to the
sort of some element of~$C$.
We have shown in Lemma~\ref{Lem: closure under inverse relabellings}
that $L$-definable languages are closed under inverse relabellings.
Therefore, we can use Lemma~\ref{Lem: morphisms of logics preserve definability}
to find a sort-wise finite set $\Psi_\xi \cup \Delta_0 \subseteq \Delta' \subseteq L$ such that
\begin{align*}
  s \sqsubseteq_{\Delta'[\Sigma]} t \qtextq{implies}
  \bbM f(s) \sqsubseteq_\Psi \bbM f(t)\,.
\end{align*}
For $s,t \in \bbM_\xi\Sigma$, it follows that
\begin{align*}
  s \sqsubseteq_{\Delta'[\Sigma]} t
  &\quad\Rightarrow\quad
  \bbM f(s) \sqsubseteq_\Psi \bbM f(t) \\
  &\quad\Rightarrow\quad
  \theta_\Delta(\varphi(s)) = \pi(\bbM f(s)) \leq \pi(\bbM f(t)) = \theta_\Delta(\varphi(t)) \\
  &\quad\Rightarrow\quad
  \varphi(s) \sqsubseteq_\Delta \varphi(t)\,.
\end{align*}
\upqed
\end{proof}

As a consequence we obtain the following counterpart to
Theorem~\ref{Thm: L-definable algebras recognise L-definable languages}.
\begin{cor}\label{Cor: L-definable iff recognised by L-definable algebra}
Let $L$~be an $\bbM$-compositional, varietal family of logics.
A~language $K \subseteq \bbM\Sigma$ is $L$-definable if, and only if, it is recognised
by an $L$-definable algebra.
\end{cor}
\begin{proof}
$(\Leftarrow)$ follows from Theorem~\ref{Thm: L-definable algebras recognise L-definable languages}.
For $(\Rightarrow)$, fix some sort-wise finite $\Delta \subseteq L$ such that
$K$~is $\Delta$-definable and $\Theta_\Delta\Sigma$~exists.
The claim follows as we have seen in Lemma~\ref{Lem: Theta recognises Delta-definable languages}
that the morphism $\theta_\Delta : \bbM\Sigma \to \Theta_\Delta\Sigma$
recognises~$L$ and in the preceding theorem that the algebra $\Theta_\Delta\Sigma$ is $L$-definable.
\end{proof}

For syntactic algebras, we obtain similar statements.
\begin{lem}\label{Lem: image of syntactic morphism is definably embedded}
Let $L$~be a varietal family of logics.
If $K \subseteq \bbM_\xi\Sigma$ is an $L$-definable language with a syntactic algebra,
then $\Syn(K)$ is $L$-definable.
\end{lem}
\begin{proof}
Clearly, $\Syn(K)$ is finitary. Hence, it remains to prove that it is locally $L$-definable.
Let $C \subseteq \Syn(K)$ be finite.
Then $N := \pi^{-1}(\Aboveseg a) \cap \bbM C$ is recognised by the restriction
$\pi \restriction \bbM C : \bbM C \to \Syn(K)$.
By Proposition~\ref{Prop: languages recognised by Syn(K)} it therefore follows that
$N$~is of the form
\begin{align*}
  N = \varphi^{-1}\bigl[\bigcup_{i<m} \bigcap_{k < n_i} p_{ik}^{-1}[K]\bigr]\,,
\end{align*}
for some morphism $\varphi : \bbM C \to \bbM\Sigma$ and contexts~$p_{ik}$.
By the assumed closure properties, all such languages are $L$-definable.
Consequently, it follows by
Theorem~\ref{Thm: L-definable algebras recognise L-definable languages}
that $C$~is $L$-definably embedded in $\Syn(K)$.
\end{proof}

\begin{thm}
Let $L$~be a family of lattice-closed logics such that every $L$-definable language
has a syntactic algebra. The following statements are equivalent.
\begin{enum1}
\item $L$~is varietal.
\item For every $L$-definable language $K \subseteq \bbM\Sigma$,
  the syntactic algebra $\Syn(K)$ is $L$-definable.
\end{enum1}
\end{thm}
\begin{proof}
(1)~$\Rightarrow$~(2) follows by
Lemma~\ref{Lem: image of syntactic morphism is definably embedded}.
For (2)~$\Rightarrow$~(1), fix an $L$-definable language $K \subseteq \bbM_\xi\Gamma$.
Then $K \subseteq \syn_K^{-1}[P]$ where $P := \syn_K[K]$.
For closure under inverse morphisms, consider $\varphi : \bbM\Sigma \to \bbM\Gamma$. Then
\begin{align*}
  \varphi^{-1}[K] = \varphi^{-1}[\syn_K^{-1}[P]]
  = (\syn_K \circ \varphi)^{-1}[P]\,,
\end{align*}
which is $L$-definable by Theorem~\ref{Thm: L-definable algebras recognise L-definable languages}.

For closure under derivatives, consider a context $p \in \bbM(\Gamma + \Box)$.
By Proposition~\ref{Prop: comparing syntactic congruences},
there exists an upwards closed set $Q \subseteq \Syn(K)$ such that $p^{-1}[K] = \syn_K^{-1}[Q]$.
Consequently, $p^{-1}[K]$ is recognised by $\Syn(K)$ and, hence, $L$-definable by
Theorem~\ref{Thm: L-definable algebras recognise L-definable languages}.
\end{proof}

Next, let us take a look at the closure properties of $L$-definable algebras.
\begin{prop}\label{Prop: L-definable algebras form pseudo-variety}
Let $L$~be an $\bbM$-compositional logic that is lattice closed.
The class of $L$-definable $\bbM$-algebras is a pseudo-variety.
\end{prop}
\begin{proof}
We start by proving that the class of locally $L$-definable $\bbM$-algebras
is closed under arbitrary (a)~subalgebras and (b)~products.
This implies that that the class of $L$-definable $\bbM$-algebras is closed
under finitary subalgebras of finite products.
Having done so, it is then sufficient to show that the latter class is also
closed under (c)~sort-accumulation points.
Hence, we have three statements to prove.

(a) Suppose that $\frakA \subseteq \frakB$ where $\frakB$~is locally $L$-definable.
Given a finite set $C \subseteq A$ and an element $a \in A$, note that the set
$K := \pi^{-1}[\Aboveseg a] \cap \bbM C$ has the same value when evaluated in~$\frakA$
and in~$\frakB$.
(To see this, note that $\pi[\bbM C] \subseteq A$ as $A$~is closed under~$\pi$.
Hence, $\Aboveseg a \cap \pi[\bbM C]$ has the same value in both algebras.)
By our assumption on~$\frakB$ it thus follows that $K$~is $L$-definable.

(b) First, note that the empty product~$\frakA$ has exactly one element~$1_\xi$ of each sort~$\xi$.
Consequently, $\pi^{-1}(\Aboveseg 1_\xi) \cap \bbM C = \bbM C$,
which is $L$-definable (by the empty conjunction).

It remains to consider the case of a non-empty product $\frakA = \prod_{i \in I} \frakB^i$.
Given a finite set $C \subseteq A$, we choose finite sets $D^i \subseteq B^i$, for $i \in I$,
such that $C \subseteq \prod_i D^i$.
Let $p_k : \prod_i B^i \to B^k$ be the projections.
For $t \in \bbM\prod_i D^i$ and $a \in A$, we have
\begin{align*}
  \pi(t) \geq a
  \quad\iff\quad
  \pi(\bbM p_i(t)) = p_i(\pi(t)) \geq p_i(a) \quad\text{for all } i\,.
\end{align*}
Thus,
\begin{align*}
  \pi^{-1}(\Aboveseg a) \cap \bbM C
  = \bigcap_{i \in I}
    {(\bbM p_i)^{-1}\bigl[\pi^{-1}(\Aboveseg p_i(a)) \cap \bbM D^i\bigr] \cap \bbM C}\,.
\end{align*}
For every pair of distinct elements $c,d \in C$, we fix one index $i \in I$ with
$p_i(c) \neq p_i(d)$. Let $H \subseteq I$ be the (finite) set of these indices.
Then we have
\begin{align*}
  c \neq d \quad\iff\quad p_i(c) \neq p_i(d)\,, \text{ for some } i \in H\,.
\end{align*}
It follows that
\begin{align*}
  \pi^{-1}(\Aboveseg a) \cap \bbM C
  &= \bigcap_{i \in I}
     {(\bbM p_i)^{-1}\bigl[\pi^{-1}(\Aboveseg p_i(a)) \cap \bbM D^i\bigr] \cap \bbM C} \\
  &= \bigcap_{h \in H}
     {(\bbM p_h)^{-1}\bigl[\pi^{-1}(\Aboveseg p_h(a)) \cap \bbM D^h\bigr] \cap \bbM C} \\
  &= \bigcap_{h \in H}
     (\bbM j)^{-1}\bigl[(\bbM p_h)^{-1}\bigl[\pi^{-1}(\Aboveseg p_h(a)) \cap \bbM D^h\bigr]\bigr]\,,
\end{align*}
where $j : C \to \prod_i D^i$ is the inclusion map.
Note that we have seen in Lemma~\ref{Lem: closure under inverse relabellings}
that the $L$-definable languages are closed under inverse relabellings.
As the~$\frakB^i$ are $L$-definable and $L$-is closed under finite conjunctions,
the above set is therefore also $L$-definable.

(c)
Let $\frakA$~be a sort-accumulation point of the class of $L$-definable algebras.
To show that $\frakA$~is also locally $L$-definable, fix a finite set $C \subseteq A$ and
an element $a \in A$. Let $\Delta \subseteq \Xi$ be a finite set of sorts such that
$C \cup \{a\} \subseteq A|_\Delta$.
By assumption, we can find an $L$-definable algebra~$\frakB$ such that
$\frakA|_\Delta$~is a quotient of~$\frakB|_\Delta$.
Let $\mu : \frakB|_\Delta \to \frakA|_\Delta$ be a surjective morphism.

Since $\mu$~is surjective, there exists a function $f : \bbV A|_\Delta \to B|_\Delta$ such that
$\mu \circ f = \iota$.
Setting $D := f[C]$, it follows that
\begin{align*}
  (\bbM\iota)^{-1}[\pi^{-1}[\Aboveseg a] \cap \bbM C]
  &= (\pi \circ \bbM\iota)^{-1}[\Aboveseg a] \cap \bbM\bbV C \\
  &= (\pi \circ \bbM\mu \circ \bbM f)^{-1}[\Aboveseg a] \cap \bbM\bbV C \\
  &= (\mu \circ \pi \circ \bbM f)^{-1}[\Aboveseg a] \cap \bbM\bbV C \\
  &= (\bbM f)^{-1}\bigl[\pi^{-1}[\mu^{-1}[\Aboveseg a]]\bigr] \cap \bbM\bbV C \\
  &= \bigcup_{b \in \mu^{-1}[\Aboveseg a]} (\bbM f)^{-1}\bigl[\pi^{-1}[\Aboveseg b]\bigr] \cap \bbM\bbV C \\
  &= \bigcup_{b \in \mu^{-1}[\Aboveseg a]} (\bbM f)^{-1}\bigl[\pi^{-1}[\Aboveseg b] \cap \bbM D\bigr] \cap \bbM\bbV C\,.
\end{align*}
This set is $L$-definable since
each language of the form $\pi^{-1}[\Aboveseg b] \cap \bbM D$ is $L$-definable and
the class of $L$-definable languages is closed under finite unions and inverse relabellings.
Consequently, $\pi^{-1}[\Aboveseg a] \cap \bbM C$ is also $L$-definable.
\end{proof}
The next theorem provides a more concrete description of this pseudo-variety\?:
it is generated by the theory algebras $\Theta_\Delta\Sigma$.

\begin{thm}\label{Thm: pseudo-variety generated by theory algebras}
Let $L$~be a varietal $\bbM$-compositional logic.
An $\bbM$-algebra~$\frakA$ is $L$-definable if, and only if,
it belongs the the pseudo-variety~$\calV$ generated by all theory algebras of
the form~$\Theta_\Delta X$ where $X$~is some finite set and $\Delta \subseteq L$
a sort-wise finite subfamily such that $\sqsubseteq_\Delta$~is a congruence ordering.
\end{thm}
\begin{proof}
$(\Leftarrow)$
We have seen in Proposition~\ref{Prop: L-definable algebras form pseudo-variety}
that the class of all $L$-definable algebras forms a pseudo-variety~$\calW$,
and in Theorem~\ref{Thm: theory algebras are L-definable}
that $\calW$~contains all theory algebras. Consequently, $\calV \subseteq \calW$.

$(\Rightarrow)$
Let $\frakA$~be $L$-definable and fix a finite set $C \subseteq A$ of generators.
It is sufficient to prove that $\frakA$~is a sort-accumulation point of theory algebras.
Hence, fix a finite set $X \subseteq \Xi$ of sorts.
W.l.o.g.\ we may assume that $X$~contains all the sorts in~$C$.
Let $\Delta \subseteq L$ be a sort-wise finite set such that every language
$\pi^{-1}[\Aboveseg a] \cap \bbM C$ with $a \in A|_X$ is $\Delta$-definable
and $\Theta_\Delta$~is defined.
For $s,t \in \bbM_\xi C$ with $\xi \in X$, we have
\begin{align*}
  \theta_\Delta(s) \leq \theta_\Delta(t)
  &\quad\Rightarrow\quad
  s \sqsubseteq_\Delta t \\
  &\quad\Rightarrow\quad
  [\pi(s) \geq a \Rightarrow \pi(t) \geq a]\,, \quad\text{for all } a \in A_\xi\,, \\
  &\quad\Rightarrow\quad
  \pi(s) \leq \pi(t)\,.
\end{align*}
Consequently, $\ker \theta_\Delta|_X \subseteq \ker \pi|_X$ and we can use the Factorisation Lemma
to find a morphism $\mu : \Theta_\Delta C|_X \to \frakA|_X$ such that
$\pi = \mu \circ \theta_\Delta$.
In particular, $\frakA|_X$~is a quotient of~$\Theta_\Delta C|_X$
and $\Theta_\Delta C \in \calV$.
\end{proof}
\begin{cor}\label{Cor: L-definable algebras satisfy theory of theory algebras}
Let $L$~be a varietal $\bbM$-compositional logic and
$\calT$~the class of all theory algebras~$\Theta_\Delta\Sigma$.
A~finitary\/ $\bbM$-algebra is $L$-definable if, and only if, it satisfies every
$\bbM$-inequality in $\Th(\calT)$.
\end{cor}
\begin{proof}
Let $\calV$~be the pseudo-variety of all $L$-definable algebras.
By Theorem~\ref{Thm: pseudo-variety generated by theory algebras}, $\calV$~is the
smallest pseudo-variety containing~$\calT$.
The class $\calW := \Mod(\Th(\calT))$ is also a pseudo-variety containing~$\calT$.
Consequently, $\calV \subseteq \calW$.
Furthermore, $\calT \subseteq \calV$ implies $\Th(\calT) \supseteq \Th(\calV)$.
Hence, it follows by Corollary~\ref{Cor: Mod(Th(V)) = V} that
\begin{align*}
  \calW = \Mod(\Th(\calT)) \subseteq \Mod(\Th(\calV)) = \calV\,.
\end{align*}
\upqed
\end{proof}

The following theorem summarises our various characterisations of when a language
is definable in a given logic.
It can be considered the main result of this article.
Of these characterisations, (8)~and~(9) are the most useful.
(9)~mainly when trying to prove that a language is $L$-definable
and (8)~when devising a decision procedure for $L$-definability.
Of course, for the latter one has to first determine the set of inequalities in question.
Depending on the logic~$L$ this can be a highly non-trivial task.
\begin{thm}\label{Thm: definability theorem}
Let $L$~be an $\bbM$-compositional varietal family of logics and let
$K \subseteq \bbM_\xi\Sigma$ be a language with a syntactic algebra.
The following statements are equivalent.
\begin{enum1}
\item $K$~is $L$-definable.
\item $K$~is recognised by some $L$-definable algebra.
\item $\Syn(K)$ is $L$-definable.
\item $\Syn(K)$ is a quotient of~$\Theta_\Delta\Gamma$, for some $\Delta$~and~$\Gamma$.
\item $\syn_K = \varrho \circ \theta_\Delta$, for some $\Delta$ and a surjective morphism
  $\varrho : \Theta_\Delta\Sigma \to \Syn(K)$.
\item $K$~is recognised by $\Theta_\Delta\Gamma$, for some $\Delta$~and~$\Gamma$.
\item $\theta_\Delta : \bbM\Sigma \to \Theta_\Delta\Sigma$ recognises~$K$, for some~$\Delta$.
\item $\Syn(K)$ satisfies all $\bbM$-inequalities $s \leq t$ that hold in every theory algebra
  $\Theta_\Delta\Gamma$.
\item There is some~$\Delta$ such that
  \begin{align*}
    s \sqsubseteq_\Delta t \qtextq{implies} s \in K \Rightarrow t \in K\,,
    \quad\text{for all } s,t \in \bbM_\xi\Sigma\,.
  \end{align*}
\item ${\sqsubseteq_\Delta} \subseteq {\preceq_K}\,,$ for some $\Delta$.
\end{enum1}
(Here $\Delta$~ranges over sort-wise finite subsets of~$L$ and $\Gamma$~ranges over alphabets.)
\end{thm}
\begin{proof}
(5)~$\Rightarrow$~(4) is trivial.

(4)~$\Rightarrow$~(6)
Since $\syn_K : \bbM\Sigma \to \Syn(K)$ recognises~$K$, the claim follows
by Lemma~\ref{Lem: languages recognised by quotients}.

(6)~$\Rightarrow$~(2) follows by Theorem~\ref{Thm: theory algebras are L-definable}.

(2)~$\Leftrightarrow$~(1)
was shown in Corollary~\ref{Cor: L-definable iff recognised by L-definable algebra}.

(1)~$\Leftrightarrow$~(9)
was proved in Lemma~\ref{Lem: L-definable iff closed under L-equivalence}.

(9)~$\Rightarrow$~(10)
Fix a finite set $\Phi \subseteq L$ such that
\begin{align*}
  s \sqsubseteq_\Phi t \qtextq{implies} s \in K \Rightarrow t \in K\,,
\end{align*}
and choose a finite set $\Phi \subseteq \Delta \subseteq L$ such that $\Theta_\Delta$~is defined.
We claim that ${\sqsubseteq_\Delta} \subseteq {\preceq_K}$.
Hence, suppose that $s \sqsubseteq_\Delta t$.
Note that we have shown in Lemma~\ref{Lem: closure under adding contexts} that
$s \sqsubseteq_\Delta t$ implies $p[s] \sqsubseteq_\Delta p[t]$, for every context~$p$.
By choice of~$\Delta$, it follows that
\begin{align*}
  p[s] \in K \qtextq{implies} p[t] \in K\,, \quad\text{for all contexts } p\,.
\end{align*}

(10)~$\Leftrightarrow$~(5)
Note that ${\sqsubseteq_\Delta} = \ker \theta_\Delta$ and ${\preceq_K} = \ker \syn_K$.
For a finite set $\Delta \subseteq L$, it therefore follows that
\begin{align*}
  {\sqsubseteq_\Delta} \subseteq {\preceq_K}
  \quad\iff\quad
  \ker \theta_\Delta \subseteq \ker \syn_K
  \quad\iff\quad
  \syn_K = \varrho \circ \theta_\Delta\,, \quad\text{for some } \varrho\,,
\end{align*}
where the last equivalence holds by the Factorisation Lemma
and Lemma~\ref{Lem: factorisation is morphism}.

(5)~$\Rightarrow$~(7)
Since $\syn_K$ recognises~$K$, there exists an upwards closed set $P \subseteq \Syn(K)$ such that
$K = \syn_K^{-1}[P]$. Setting $Q := \varrho^{-1}[P]$, it follows that
\begin{align*}
  K = \syn_K^{-1}[P] = \theta_\Delta^{-1}[\varrho^{-1}[P]] = \theta_\Delta^{-1}[Q]\,.
\end{align*}

(7)~$\Rightarrow$~(9)
Suppose that $K = \theta_\Delta^{-1}[P]$ for some upwards closed set~$P$.
If $s \sqsubseteq_\Delta t$ and $s \in K$, then
\begin{align*}
  \theta_\Delta(s) \leq \theta_\Delta(t) \qtextq{and} \theta_\Delta(s) \in P\,,
\end{align*}
which implies that $\theta_\Delta(t) \in P$, i.e., $t \in K$.

(4)~$\Rightarrow$~(3)
We have seen in Theorem~\ref{Thm: theory algebras are L-definable} that every theory algebra
is $L$-definable,
and in Proposition~\ref{Prop: L-definable algebras form pseudo-variety}
that the class of $L$-definable algebras forms a pseudo-variety.
In particular, it is closed under quotients.

(3)~$\Rightarrow$~(2) holds as $\syn_K : \bbM\Sigma \to \Syn(K)$ recognises~$K$.

(3)~$\Leftrightarrow$~(8) follows by
Corollary~\ref{Cor: L-definable algebras satisfy theory of theory algebras}.
\end{proof}

\section{Applications}   
\label{Sect:applications}

As an example, let us see how this abstract framework performs in the case
of languages of infinite trees. In this case, the functor~$\bbM$ maps a given set~$A$
to the set of all (finite and infinite) $A$-labelled trees.
There are several possible ways to chose the precise definition for~$\bbM$.
We will present two of them denoted $\bbT$~and~$\bbT^\times$.
The latter is the more general one, while the former is a subfunctor.
Both operate on the category $\Pos^\omega$ with sorts $\Xi := \omega$.
We interpret a sort $n < \omega$ as the \emph{arity} of an element.
Given a set $A \in \Pos^\omega$, the set $\bbT^\times A$ consists of all
(finite or infinite) trees~$t$ where each vertex~$v$ is labelled by an element~$a$ of~$A$
such that the arity of~$a$ matches the number of successors of~$v$.
Hence, the elements of~$A_0$ appear at the leaves of~$t$,
those of~$A_2$ at internal vertices with exactly two successors, and so on.
In addition to the elements of~$A$ we also allow as labels special variable symbols
$x_0,x_1,x_2,\dots$, which are treated as elements of arity~$0$ and which are supposed
to be distinct from all elements of~$A$.
Thus $\bbT^\times A$ can be interpreted as the set of all (possibly infinite) non-closed terms over
the signature~$A$.
The set $\bbT^\times_nA$ consists of all trees~$t$ that use only the variable symbols
$x_0,\dots,x_{n-1}$.
Furthermore, we assume that the root is labelled by an element of~$A$, not a variable.
(This is needed to define the flattening function below.)
Formally, we consider such a tree $t \in \bbT^\times_nA$ as a function
$t : \dom(t) \to A + \{x_0,\dots,x_{n-1}\}$ where $\dom(t)$~is the set of vertices of~$t$.

Note that each variable~$x_i$ can be used once, several times, or not at all.
The subset $\bbT A \subseteq \bbT^\times A$ consists of all those trees that use each variable
at most once. (Such trees are sometimes called \emph{linear} in the literature.)

$\bbT$~and~$\bbT^\times$ are clearly polynomial functors. We turn them into monads as follows.
The singleton map $\sing : A \to \bbT^\times A$ maps an element $a \in A_n$ to the tree
$a(x_0,\dots,x_{n-1})$ consisting of a root with label~$a$ to which are attached~$n$
leaves with labels  $x_0,\dots,x_{n-1}$, respectively.
The flattening map $\Flat : \bbT^\times\bbT^\times A \to \bbT^\times A$ works as follows.
Given a tree $T \in \bbT^\times\bbT^\times A$ where each vertex~$v$ is labelled by some tree
$T(v) \in \bbT^\times A$,
we build a large tree assembled from the trees~$T(v)$ by
\begin{itemize}
\item taking the disjoint union of all trees $T(v)$\?;
\item replacing each occurrence of a variable~$x_i$ in $T(v)$ by an edge to the root
  of $T(u)$, where $u$~is the $(i+1)$-th successor of~$v$\?;
\item unravelling the resulting directed acyclic graph into a tree.
\end{itemize}
For details, we refer the reader to~\cite{Blumensath20,BlumensathZZ}.
It is now straightforward but a bit tedious to check that $\bbT^\times$
together with these two operations forms a monad.
Hence, so is the restriction to~$\bbT$.

The logics we are looking at are \emph{first-order logic} $\FO$ and
\emph{monadic second-order logic} $\MSO$.
To define the satisfaction relation between formulae of these logics and
elements of $\bbT^\times\Sigma$, we encode a tree $t \in \bbT^\times_n\Sigma$ as the structure
\begin{align*}
  \frakT = \bigl\langle T,\,{\preceq},\,(S_i)_{i<\omega},\,(P_c)_{c \in \Sigma},\, (Q_i)_{i<n},\,R\bigr\rangle
\end{align*}
where $T = \dom(t)$, $\preceq$~is the tree ordering (with the root as least element),
$S_i$~is the $i$-th successor relation, $P_c = t^{-1}(c)$ the set of all vertices $v \in T$
with label $c \in \Sigma$, $Q_i = t^{-1}(x_i)$ the vertices labelled by the variable~$x_i$,
and $R$~is an unary relation that only contains the root.
By necessity the proofs below assume some familiarity with tree automata
(more precisely, non-deterministic parity automata)
and back-and-forth arguments. Readers who want to refresh their knowledge
we refer to \cite{Thomas97}~and~\cite{EbbinghausFlum95} for an introduction.

We start by proving that $\MSO$-definable languages have syntactic algebras.
\begin{thm}
$\bbT^\times$ is essentially finitary over the class of all $\MSO$-definable $\bbT^\times$-algebras.
\end{thm}
\begin{proof}
Let $\bbT^\reg A \subseteq \bbT^\times A$ the set of all \emph{regular} trees in~$\bbT^\times A$,
i.e., those that, up to isomorphism, have only finitely many distinct subtrees.
We claim that the inclusion morphism $\bbT^\reg \Rightarrow \bbT^\times$ is dense
over the class of all finite products of $\MSO$-definable $\bbT^\times$-algebras.

Let $\frakA_0,\dots,\frakA_{n-1}$ be $\MSO$-definable,
$B \subseteq A_0 \times\dots\times A_{n-1}$, and $t \in \bbT^\times B$ a tree with $\pi(t) = \bar a$.
We have to find a regular tree $t^\circ \in \bbT^\reg B$ with $\pi(t^\circ) = \bar a$.
Let $C_i \subseteq A_i$ be a finite set of generators of~$\frakA_i$ and let
$\calA_i$~be a parity automaton recognising $\pi^{-1}(a_i) \cap \bbT^\times C_i$.
Suppose that $Q_i$~is the set of states of~$\calA_i$, $K_i$~the set of priorities
used by it, and $\Omega_i : Q_i \to K_i$ the corresponding priority function.
For every $\bar b \in B$, we fix trees $\sigma_i(\bar b) \in \bbT^\times C_i$ with
$\pi(\sigma_i(\bar b)) = b_i$, for $i < n$.
This defines a function $\sigma_i : \bbV B \to \bbT^\times \bbV C_i$, which we can extend to a
morphism $\hat\sigma_i : \bbT^\times \bbV B \to \bbT^\times \bbV C_i$.

We construct the desired tree~$t^\circ$ by the following variant of the
usual Automaton--Pathfinder game (see, e.g.,~\cite{Thomas97}).
In this game Automaton tries to construct a tree $s \in \bbT^\times B$ such that,
for every $i < n$, $\hat\sigma_i(s)$~is accepted by~$\calA_i$,
while Pathfinder tries to prove that such a tree does not exist.
We will define the game in such a way that there is a correspondence between
winning strategies for Automaton and such trees~$s$.
Note that these are exactly the trees~$s$ with $\pi(s) = \bar a$,
since
\begin{align*}
  \pi(\hat\sigma_i(s))
  = \pi(\Flat(\bbT^\times\sigma_i(s)))
  = \pi(\bbT^\times\pi(\bbT^\times\sigma_i(s)))
  = \pi(\bbT^\times p_i(s))
  = p_i(\pi(s))\,,
\end{align*}
where $p_i : A_0\times\dots\times A_{n-1} \to A_i$ is the $i$-th projection.
As $\pi(t) = \bar a$, it follows that Automaton indeed has a winning strategy for the game.
Furthermore, the winning condition of our game is regular. Therefore, it follows by
the B\"uchi-Landweber Theorem~\cite{BuchiLandweber69} that Automaton even has a winning strategy
that uses only a finite amount of memory.
As the trees~$s$ corresponding to finite-memory strategies via the above correspondence
are regular, the claim follows.

To conclude the proof, it therefore remains to define a regular game with the above properties.
In each round, Automaton picks the label $\bar b \in B$ for the next vertex~$v$ of~$s$
and Pathfinder responds by choosing one of the successors of~$v$.
While doing so, we have to keep track of all the states of the various automata from which
we want to accept the remaining subtree.

The positions for Automaton are of the form
$\bar U \in \prod_{i<n} \PSet(K_i \times Q_i)$,
while those for Pathfinder are tuples $\langle\bar V_0,\dots,\bar V_{m-1}\rangle$
where each component~$\bar V_i$ is a position for Automaton.
The initial position belongs to Automaton and consists of the tuple
$\bigl\langle\bigl\{\langle 0,q^i_0\rangle\bigr\}\bigr\rangle_{i<n}$,
where $q^i_0$~is the initial state of~$\calA_i$.

In a position~$\bar U$, Automaton chooses an element $\bar b \in B$ and,
for every $i < n$ and every pair $\langle k,q\rangle \in U_i$,
a partial run~$\varrho_q$ of~$\calA_i$ on the tree~$\sigma_i(\bar b)$ that starts in the state~$q$.
(It will turn out that Automaton can choose this run independently of~$k$.
So we drop it to keep the notation light. We also assume that the sets~$Q_i$ are disjoint,
so we do not need to specify the index~$i$.)
Suppose that $\bar b \in B_m$ has arity~$m$. For $i < n$ and $j < m$,
let $H_{ij}$~be the set of all vertices of~$\sigma_i(\bar b)$ labelled by the variable~$x_j$.
We denote by~$W_{ij}(q)$ the set of all pairs
$\langle k',q'\rangle \in K_i \times Q_i$ such that there is some $v \in H_{ij}$ with
\begin{align*}
  \varrho_q(v) = q'
  \qtextq{and}
  k' := \min {\set{ \Omega_i(\varrho_q(w)) }{ w \preceq v }}\,.
\end{align*}
The new position is $\langle\bar V^0,\dots,\bar V^{m-1}\rangle$ where
\begin{align*}
  V^j_i := \bigcup_{\langle k,q\rangle \in U_i} W_{ij}(q)\,.
\end{align*}
Pathfinder responds by choosing some $j < m$ after which the game proceeds to position~$\bar V^j$.

Automaton wins a play of this game if either the play ends
in the position $\langle\rangle$ where Pathfinder cannot make a move,
or if the play is infinite and satisfies the following variant of the parity condition.
Suppose that the play is $\bar U^0,\bar V^0,\bar U^1,\bar V^2,\dots$ and let $W^l_{ij}(q)$
be the sets used in the $l$-th turn by Automaton to determine the next position
$\bar V^l = \langle \bar V^l_0,\dots,\bar V^l_{m-1}\rangle$.
We call a sequence $k_0,q_0,k_1,q_1,k_2,q_0,\dots$ an \emph{$i$-trace} of this play if
$\langle k_0,q_0\rangle \in U^0_i$ and, for all $l < \omega$,
\begin{align*}
  \langle k_{l+1},q_{l+1}\rangle \in W^l_{ij}(q_l)\,,
  \quad\text{for some } j \text{ with } \bar U^{l+1} = V^l_j\,.
\end{align*}
We say that the play satisfies the parity condition if,
for all $i < n$,
\begin{align*}
  \liminf_{l < \omega} k_l \text{ is even,}
  \quad\text{for all $i$-traces } k_0,q_0,k_1,q_1,k_2,q_0,\dots.
\end{align*}
Note that this is a regular winning condition.
Furthermore, it is straightforward to check that Automaton wins this game if, and only if,
there exists some tree $s \in \bbT^\times B$ such that,
for every $i <n $, the tree $\hat\sigma_i(s)$ is accepted by~$\calA_i$.
\end{proof}
Consequently, we can use Theorem~\ref{Thm: essentially finitary monads have syntactic algebras}
to prove the existence of syntactic algebras.
\begin{cor}
Every $\MSO$-definable language has a syntactic algebra.
\end{cor}

Our next goal is to show that $\MSO$ and $\FO$ are varietal and compositional.
We start with $\MSO$.
\begin{thm}\label{Thm: MSO bbT*-compositional}
The logic\/ $\MSO$ is\/ $\bbT^\times$-compositional and, therefore, also $\bbT$-com\-pos\-i\-tional.
\end{thm}
\begin{proof}
We start with a bit of terminology.
A \emph{partial run} of a tree automaton~$\calA$ on some tree $t \in \bbT^\times\Sigma$ is a function
$\varrho$~assigning states to the vertices of~$t$ in such a way that
\begin{itemize}
\item $\varrho$~satisfies the transition relation of~$\calA$ at every vertex with a label
  in~$\Sigma$,
\item there is no restriction on $\varrho(v)$ if $v$~is the root or a leaf labelled by a variable,
\item every infinite branch of~$\varrho$ satisfies the parity condition.
\end{itemize}
The \emph{profile} of a partial run~$\varrho$ on a tree~$t$ is the tuple
$\tau = \langle p,\bar U\rangle$ where $p$~is the state at the root of~$t$
and $U_i$~is the set of all pairs $\langle k,q\rangle$ such that there exists
some leaf~$v$ of~$t$ labelled~$x_i$ with state $q := \varrho(v)$ and such that
the least priority seen along the path from the root to~$v$ is equal to~$k$.

Because of the translations between formulae and automata, there exists,
for every automaton~$\calA$ and each profile~$\tau$ of~$\calA$,
an $\MSO$-formula~$\varphi_{\calA,\tau}$ stating that there is a partial run of~$\calA$ on the
given tree with profile~$\tau$.
Furthermore, every $\MSO$-formula is equivalent to some formula of this kind.

For $m < \omega$, let $\MSO_{(m)}$ denote the set of all $\MSO$-formulae equivalent to a formula
of the form~$\varphi_{\calA,\tau}$ where $\calA$~is an automaton with at most~$m$ states.
Since there are only finitely many such automata and each of them has only finitely many profiles
of partial runs, it follows that $\MSO_{(m)}$ is finite (up to logical equivalence).
Let $\equiv_{(m)}$~be the equivalence relation
which holds for two trees if they satisfy the same $\MSO_{(m)}$-formulae.
We claim that $\equiv_{(m)}$~is a congruence ordering.
This means that, if $S,T \in \bbT^\times\bbT^\times\Sigma$ are trees with the same `shape', i.e.,
$\dom(S) = \dom(T)$, then
\begin{align*}
  S(v) \equiv_{(m)} T(v)\,, \quad\text{for all } v\,,
  \qtextq{implies}
  \Flat(S) \equiv_{(m)} \Flat(T)\,.
\end{align*}

For the proof, fix a formula $\varphi_{\calA,\tau} \in \MSO_{(m)}$ with $\Flat(S) \models \varphi$.
We have to show that $\Flat(T)$ also satisfies~$\varphi_{\calA,\tau}$,
i.e., that there is a partial run of~$\calA$ on $\Flat(T)$ with profile~$\tau$.
To do so, we introduce the following variant of the Automaton--Pathfinder game.
For a given tree $T \in \bbT^\times\bbT^\times\Sigma$,
Player Automaton tries to prove that there is a partial run of~$\calA$ on $\Flat(T)$
with profile~$\tau$, while Pathfinder tries to disprove him.
The game starts in the position $\langle r,\tau\rangle$ where $r$~is the root of~$T$.
In a position $\langle v,\upsilon\rangle$ where $v \in \dom(T)$ and $\upsilon$~is a profile,
Automaton tries to show that there exists a partial run~$\varrho$ on the subtree rooted at~$v$
with profile~$\upsilon$.
He starts by choosing a partial run~$\varrho$ of~$\calA$ on the tree $T(v)$ starting in the same
state as~$\upsilon$. Then he has to choose profiles~$\bar\lambda$ for all the subtrees attached
to the copy of $T(v)$ in~$\Flat(T)$ such that the `composition' of the profile of~$\varrho$
and~$\bar\lambda$ is equal to~$\upsilon$.
This is done as follows.

Let $\mu = \langle p,\bar U\rangle$ be the profile of~$\varrho$.
For each component~$U_i$, Automaton chooses a set~$W_i$ of triples $\langle k,q,\lambda\rangle$
where $k$~is a priority, $q$~a state, and $\lambda$~a profile.
These sets must satisfy the following conditions.
\begin{itemize}
\item $U_i$~is the projection of~$W_i$ to the first two components.
\item For each $\langle k,q,\lambda\rangle \in W_i$, the state~$q$ is equal to the starting state
  of~$\lambda$.
\item $\upsilon = \langle p,\bar V\rangle$ is the composition of~$\mu$ and the profiles~$\lambda$.
  Formally,
  \begin{align*}
    V_i = \biglset \langle l,q'\rangle \bigmset {}
                 & \langle k,q,\lambda\rangle \in W_i\,,\
                   \lambda = \langle q,\bar L\rangle\,,\
                   \langle k',q'\rangle \in L\,, \\
                 & l = \min {\{k,k'\}} \bigrset\,.
  \end{align*}
\end{itemize}
Let $u_0,\dots,u_{n-1}$ be the successors of~$v$.
Given $\bar W$, Pathfinder responds by choosing a successor~$u_i$ of~$v$ and a triple
$\langle k,q,\lambda\rangle \in W_i$.
Then the game continues in the position $\langle u_i,\lambda\rangle$.

If the game reaches a leaf of~$T$, it ends with a win for one of the players.
If the leaf is labelled by a variable~$x_i$ and the current position is $\langle v,\upsilon\rangle$,
then Automaton wins if, and only if, $\upsilon$~is of the form $\langle q,\bar U\rangle$
with $U_i = \{q\}$ and $U_j = \emptyset$, for $j \neq i$. Otherwise, Pathfinder wins.
If the leaf is not labelled by a variable, then Automaton wins if he can choose
$\mu = \langle p,\bar U\rangle$ such that $U_i = \emptyset$, for all~$i$.

In the case where the game is infinite, Automaton wins
if the sequence of pairs $\langle k_0,q_0,\lambda_0\rangle,\langle k_1,q_1,\lambda_0\rangle,\dots$
chosen by Pathfinder satisfies the parity condition
\begin{align*}
  \liminf_{i < \omega} k_i \quad\text{is even}\,.
\end{align*}

It is straightforward to check that Automaton wins the game on a given tree~$T$
if, and only if, there exists a partial run of~$\calA$ on $\Flat(T)$ with profile~$\tau$.
(Every partial run of~$\calA$ on $\Flat(T)$ with this profile gives rise to a winning strategy
in the game and, conversely, every winning strategy can be used to construct a partial run
with the desired profile.)

To conclude the proof we have to show that,
if $T$~is a tree with $S(v) \equiv_{(m)} T(v)$, for all~$v$,
then Automaton has a winning strategy in the game on~$T$.
By construction, Automaton has a winning strategy~$\sigma$ in the game on~$S$.
We use it to define a winning strategy~$\sigma'$ in the game on~$T$ as follows.
If $\sigma$~tells Automaton to choose a partial run~$\varrho$ on $S(v)$,
$\sigma'$~returns some partial run~$\varrho'$ on~$T(v)$ with the same profile as~$\varrho$.
(This is possible since $S(v) \equiv_{(m)} T(v)$.)
As only the profile of the chosen run is used by the game and $\sigma$~is winning,
it follows that the resulting strategy~$\sigma'$ is also winning.
\end{proof}
\begin{Rem}
Note that in the above proof we have chosen a rather strange stratification of $\MSO$.
It might be nice if we could use the usual stratification in terms of the quantifier-rank instead,
but this does not seem to work for~$\bbT^\times$.
For the monad~$\bbT$ on the other hand, there is an alternative proof consisting
of a simple inductive back-and-forth argument based on the quantifier-rank.
\end{Rem}

To show that $\MSO$ is varietal it suffices, by Theorem~\ref{Thm: theory algebras are L-definable},
to prove that the theory algebras are $\MSO$-definable,
\begin{prop}\label{Lem: MSO closed under inverse morphisms}
Let $\Sigma$~be an alphabet and $\Delta_m := \MSO_{(m)}$ the fragment of\/ $\MSO$ used in
the proof of Theorem~\ref{Thm: MSO bbT*-compositional}.
The theory algebra $\Theta_{\Delta_m}\Sigma$ is\/ $\MSO$-definable.
\end{prop}
\begin{proof}
The set $C := \theta_{\Delta_m}[\Sigma]$ is a finite set of generators of~$\Theta_{\Delta_m}\Sigma$.
Given a $\Delta_m$-theory $\sigma \in \Theta_{\Delta_m}\Sigma$, we have to find an
$\MSO$-formula~$\varphi$ defining the set
\begin{align*}
  \pi^{-1}(\sigma) \cap \bbM C\,.
\end{align*}
Each formula $\chi \in \sigma$ is a statement of the form\?:
`there exists a partial run of the automaton~$\calA$ with profile~$\tau$'.
Let us write $\chi = \chi_{\calA,\tau}$ to mark the relevant parameters.
For $t \in \bbM C$, it follows that $\pi(t) = \sigma$ if, and only if,
for every tree $s \in \bbM\Sigma$ with $\bbM\theta_{\Delta_m}(s) = t$
and every $\chi \in \sigma$, there exists a partial run of the corresponding automaton on~$s$
with the corresponding profile.
Consequently, to define the above preimage it is sufficient to express, for a given automaton~$\calA$
and a profile~$\tau$, that every preimage of the given tree~$t$ under $\bbM\theta_{\Delta_m}$ has
a partial run of~$\calA$ with profile~$\tau$.
This can be done by saying that, for every vertex~$v$ there is some formula
$\chi_{\calA,\upsilon_v} \in t(v)$ such that the `composition' of the profiles~$\upsilon_v$
yields~$\tau$.
For this composition, we have to check that the states at the borders match and to compute
the minimal priorities on each branch. All of this can easily be expressed in $\MSO$.
\end{proof}

Let us turn to $\FO$ next. Again, we start with compositionality.
\begin{thm}\label{Thm: FO bbT*-compositional}
The logic\/ $\FO$ is\/ $\bbT^\times$-compositional and, therefore, also $\bbT$-com\-pos\-i\-tional.
\end{thm}
\begin{proof}
Let $\FO_m$~denote the set of all first-order formulae of quantifier-rank at most~$m$
and denote by~$\equiv_m$ equivalence with respect to such formulae.
We claim that $\equiv_m$~is a congruence on $\bbT^\times\Sigma$.
Hence, consider two trees $S,T \in \bbT^\times\bbT^\times\Sigma$ with $\dom(S) = \dom(T)$
satisfying
\begin{align*}
  S(v) \equiv_m T(v)\,, \quad\text{for all vertices } v\,.
\end{align*}
We have to show that $\Flat(S) \equiv_m \Flat(T)$.

The proof is by induction on~$m$. To make the inductive step go through we have to
prove a slightly stronger statement involving parameters.
Given a tuple~$\bar a$ of vertices of $\Flat(S)$ and a copy~$s$ of~$S(v)$ in~$\Flat(S)$,
we denote by~$\bar a^s$ the tuple
\begin{align*}
  a^s_i := \begin{cases}
             a_i &\text{if } a_i \in \dom(s)\,, \\
             v   &\text{if } v \text{ is a leave of } s \text{ labelled by a variable and }
                  a_i \text{ is a descendant} \\
                 &\quad\text{of } v \text{ in } \Flat(S)\,, \\
             u   &\text{if } v \text{ is the root of } s \text{ and } a_i
                  \text{ is not a descendent of } v \text{ in } \Flat(S)\,.
           \end{cases}
\end{align*}
We use the same notation for parameters in~$\Flat(T)$.
For a tuple~$\bar a$ of vertices of some tree~$s$, we write $\langle s,\bar a\rangle$
for the expansion of~$s$ by constants for the vertices~$\bar a$.
The claim we prove is that, for trees $S,T \in \bbT^\times\bbT^\times\Sigma$
with the same `shape' and with parameters~$\bar a$
in~$\Flat(S)$ and $\bar b$~in $\Flat(T)$,
\begin{align*}
  \langle s,\bar a^s\rangle \equiv_m \langle t,\bar b^t\rangle\,,
  \quad\text{for all } v, \text{ copies } s \text{ of } S(v),
  \text{ and copies } t \text{ of } T(v),
\end{align*}
implies that
\begin{align*}
  \langle\Flat(S),\bar a\rangle \equiv_m \langle\Flat(T),\bar b\rangle\,.
\end{align*}

For $m = 0$, the proof is straightforward. For the inductive step, suppose that
\begin{align*}
  \langle s,\bar a^s\rangle \equiv_{m+1} \langle t,\bar b^t\rangle\,,
  \quad\text{for all } v, \text{ copies } s \text{ of } S(v),
  \text{ and copies } t \text{ of } T(v).
\end{align*}
We use a back-and-forth argument to show that
$\langle\Flat(S),\bar a\rangle \equiv_{m+1} \langle\Flat(T),\bar b\rangle$.
Let $c \in \dom(\Flat(S))$ be a new parameter.
Suppose that $c$~belongs to a copy~$s$ of the tree $S(v)$.
When we want to apply the inductive hypothesis, we now face the problem that,
if $\Flat(S)$ contains several copies of $S(v)$, only one of them contains the new parameter.
To solve this issue, we have to modify the trees $S$~and~$T$ to make sure this does not happen.

Let $v_0,\dots,v_n$ be the path from the root~$v_0$ of~$S$ to $v = v_n$
and let $s_i$~be the copy of $S(v_i)$ in $\Flat(S)$ such that $c$~is a descendent of the
root of~$s_i$.
We construct new trees $S_0,\dots,S_n$ and $T_0,\dots,T_n$ as follows.
We start with $S_0 := S$ and $T_0 := T$.
For the inductive step, suppose we have already defined $S_i$~and~$T_i$ for some $i < n$
and that there is a unique copy~$t_i$ of $T_i(v_i)$ in $\Flat(T_i)$.
We choose a vertex~$d_i$ of~$t_i$ such that
\begin{align*}
  \langle s_i,\bar a^{s_i}c^{s_i}\rangle \equiv_m \langle t_i,\bar b^{t_i}d_i\rangle\,.
\end{align*}
Note that the vertex~$c^{s_i}$ is a leaf labelled by some variable~$x_j$.
Hence, so is~$d_i$.
If there is no other occurrence of~$x_j$ in~$s_i$, we set $S_{i+1} := S_i$.
Otherwise, we choose some variable~$x_k$ that does not appear in~$s_i$
and we replace every occurrence of~$x_j$ in~$s_i$ by~$x_k$, except for the one at~$c^{s_i}$.
Let $S_{i+1}$~be the tree obtained from~$S_i$ by
\begin{itemize}
\item changing $S(v_i) = s_i$ in this way and
\item duplicating the subtree attached to~$v_i$ that corresponds to the variable~$x_j$
  in such a way that the new copy corresponds to the variable~$x_k$.
\end{itemize}
This ensures that $\Flat(S_{i+1}) = \Flat(S_i)$
and that $S_{i+1}$ contains a unique copy of $S(v_{i+1})$.
The tree $T_{i+1}$~is obtained from~$T_i$ in exactly the same way.

Having constructed $S_n$ and $T_n$, we choose some element $d_n \in \dom(T_n(v_n))$ such that
\begin{align*}
  \langle s_n,\bar a^{s_n}c^{s_n}\rangle \equiv_m \langle t_n,\bar b^{t_n}d_n\rangle\,.
\end{align*}
Setting $d := d_n$, it follows that
$d^{t_i} = d_i$, for all $i \leq n$, which implies that
\begin{align*}
  \langle s_i,\bar a^{s_i}c^{s_i}\rangle \equiv_m \langle t_i,\bar b^{t_i}d^{t_i}\rangle\,,
  \quad\text{for all } i \leq n\,.
\end{align*}
Note that,
if $u$~is a vertex different from $v_0,\dots,v_n$, $s$~a copy of $S_n(u)$ and $t$~a copy
of $T_n(u)$, then $c^s$~is the root of~$s$ and $d^s$~is the root of~$t$.
Consequently, we also have
\begin{align*}
  \langle s,\bar a^sc^s\rangle \equiv_m \langle t,\bar b^td^t\rangle\,.
\end{align*}
Hence, the trees $S_n$~and~$T_n$ together with the parameters $\bar a,c$ and $\bar b,d$
satisfy our inductive hypothesis and it follows that
\begin{align*}
  \langle\Flat(S_n),\bar a,c\rangle \equiv_m \langle\Flat(T_n),\bar b,d\rangle\,.
\end{align*}
Since $\Flat(S_n) = \Flat(S)$ and $\Flat(T_n) = \Flat(T)$, the claim follows.

In the same way we can show that, for every choice of~$d$ in $\Flat(T)$,
we find a matching vertex~$c$ in $\Flat(S)$.
\end{proof}

It remains to show that $\FO$ is varietal. It turns out that this is only the case
for the monad~$\bbT$, but not for $\bbT^\times$.
\begin{prop}
$\FO$~is closed under inverse morphisms of\/ $\bbT$-algebras.
\end{prop}
\begin{proof}
Let $\varphi : \bbT\Sigma \to \bbT\Gamma$ be a morphism of $\bbT$-algebras
and let $\varphi_0 := \varphi \circ \sing : \Sigma \to \bbT\Gamma$ be its restriction to~$\Sigma$.
For $s,t \in \bbT\Sigma$, we will prove that
\begin{align*}
  s \equiv_m t \qtextq{implies} \varphi(s) \equiv_m \varphi(t)\,,
\end{align*}
where $\equiv_m$~denotes equivalence with respect to $\FO$-formulae of quantifier-rank at most~$m$.
For the induction we again need to prove a more general statement involving parameters.
We start with setting up a bit of notation.

Note that a tree of the form $\varphi(s) = \Flat(\bbT\varphi_0(s))$
is obtained from~$s$ by replacing each vertex~$u$ by a tree $\varphi_0(s(u))$.
For $s \in \bbT\Sigma$, we denote by $g_s : \dom(\varphi(s)) \to \dom(s)$
the function mapping a vertex~$u$ of $\varphi(s)$ to the vertex $v := g_s(u)$ such that
the copy of the tree $\varphi_0(s(v))$ replacing~$v$ in $\varphi(s)$ contains~$u$.
(Note that this copy of $\varphi_0(s(v))$ is unique, since we are dealing with
the monad~$\bbT$.)

For an $n$-tuple $\bar a$ of vertices of~$\varphi(s)$ and a vertex~$u$ of~$s$, we set
\begin{align*}
  I_u := \set{ i < n }{ g_s(a_i) = u }
  \qtextq{and}
  \bar a^u := (a_i)_{i \in I_u}\,,
\end{align*}
where we consider $\bar a^u$~as a tuple of vertices of $\varphi_0(s(u))$.
\begin{figure}
\centering
\includegraphics{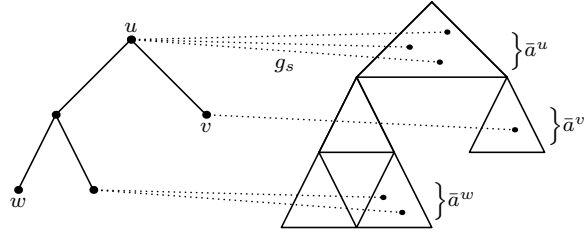}
\caption{Definition of $g_s$, $I_u$, and $\bar a^u$}
\end{figure}

The statement we will prove by induction on~$m$ is the following.
Let $s,t \in \bbT\Sigma$ be trees and $\bar a$~and~$\bar b$ $n$-tuples of parameters of,
respectively, $\varphi(s)$~and~$\varphi(t)$. Then
\begin{align*}
  &\bigl\langle s,g_s(\bar a)\bigr\rangle \equiv_m \bigl\langle t,g_t(\bar b)\bigr\rangle \\[1ex]
\prefixtext{and}
  &\bigl\langle \varphi_0(s(u)),\bar a^u\bigr\rangle \cong
   \bigl\langle \varphi_0(t(v)),\bar b^v\bigr\rangle\,,
  \quad\text{for all } u, v \text{ with } I_u = I_v \neq \emptyset\,,
\end{align*}
implies
\begin{align*}
  \langle\varphi(s),\bar a\rangle \equiv_m \langle\varphi(t),\bar b\rangle\,.
\end{align*}

For $m = 0$, this is immediate.
Hence, suppose that $m > 0$.
We have to check the back-and-forth properties.
Thus, let $c \in \dom(\varphi(s))$ and set $u := g_s(c)$.
Then there is some $v \in \dom(t)$ such that
\begin{align*}
  \langle s,g_s(\bar a),g_s(c)\rangle \equiv_{m-1} \langle t,g_t(\bar b),v\rangle\,.
\end{align*}
Note that
\begin{align*}
  I_u = \set{ i }{ g_s(a_i) = u } = \set{ i }{ g_t(b_i) = v } = I_v\,.
\end{align*}
We distinguish two cases.
If $I_u = I_v \neq \emptyset$, then there exists an isomorphism
\begin{align*}
  \sigma : \bigl\langle \varphi_0(s(u)),\bar a^u\bigr\rangle \to
           \bigl\langle \varphi_0(t(v)),\bar b^v\bigr\rangle
\end{align*}
and we can set $d := \sigma(c)$.

Otherwise, $I_u = I_v = \emptyset$ and $s(u) = t(v)$ implies that
$\varphi_0(s(u)) \cong \varphi_0(t(v))$.
Hence, can choose some element~$d$ of $\varphi_0(t(v))$ such that
\begin{align*}
  \bigl\langle \varphi_0(s(u)),c\bigr\rangle \cong
  \bigl\langle \varphi_0(t(v)),d\bigr\rangle\,.
\end{align*}

In both cases, it now follows that
\begin{align*}
  &\bigl\langle s,g_s(\bar a),g_s(c)\bigr\rangle \equiv_m
    \bigl\langle t,g_t(\bar b),g_t(d)\bigr\rangle \\[1ex]
\prefixtext{and}
  &\bigl\langle \varphi_0(s(u)),\bar a^uc^u\bigr\rangle \cong
   \bigl\langle \varphi_0(t(v)),\bar b^vd^v\bigr\rangle\,,
  \quad\text{for all } u, v \text{ with } I_u = I_v \neq \emptyset\,,
\end{align*}
which, by inductive hypothesis, implies that
\begin{align*}
  \langle \varphi(s),\bar ac\rangle \equiv_{m-1} \langle \varphi(t),\bar bd\rangle\,.
\end{align*}
The other direction follows by symmetry.
\end{proof}

As already noted by Boja\'nczyk and Michalewski~\cite{BojanczykMiXX},
$\FO$~is not closed under inverse morphisms of $\bbT^\times$-algebras.
Their counterexample rests on the following lemma.
Recall that a tree is \emph{complete binary} if every non-leaf has exactly two successors.
\begin{lem}[Potthoff~\cite{Potthoff94}]\label{Lem: even depth FO-definable}
There exists a first-order formula~$\varphi$ such that a finite complete binary tree\/
$\frakT = \langle T,S_0,S_1,{\preceq}\rangle$ satisfies~$\varphi$ if, and only if,
every leaf of\/~$\frakT$ has an even distance from the root.
\end{lem}
\begin{proof}
The basic idea is as follows.
If every leaf is at an even distance from the root, we can determine whether a vertex~$x$
belongs to an even level of the tree by walking a zig-zag path from~$x$ downwards until
we hit a leaf. For such a path it is trivial to check that its length is even.
Hence, our formula only needs to express
that the level parities computed in this way are consistent
and that the root is on an even level.

To express all this in first-order logic, we first define a few auxiliary formulae.
\begin{align*}
  \mathrm{suc}(x,y) := {}& S_0(x,y) \lor S_1(x,y) \\
  \mathrm{zigzag}(x,y;u,v) := {}& [S_0(x,y) \land S_1(u,v)] \lor [S_1(x,y) \land S_0(u,v)] \\
  \mathrm{probe}(x,y) := {}& x \preceq y \land \neg\exists z[\mathrm{suc}(y,z)] \\
  {} \land {}&\begin{aligned}[t]
                \forall u\forall v\forall w[
                &x \preceq u \land \mathrm{suc}(u,v) \land \mathrm{suc}(v,w) \land w \preceq y \\
                &\quad{} \lso \mathrm{zigzag}(u,v;v,w)]\,.
              \end{aligned}
\end{align*}
The first one just states that $y$~is a successor of~$x$\?;
the second one says that $\langle x,y\rangle$ and $\langle u,v\rangle$ are two edges
that go into different directions, one to the left and one to the right\?;
and the last one states that $y$~is one of the two leaves below~$x$ that are reached by
a zig-zag path consisting of alternatingly taking left and right successors.

Using these formulae we can express that a vertex~$x$ has an even distance from some leaf by
\begin{align*}
  \mathrm{even}(x) := \exists y[
    &\mathrm{probe}(x,y) \land {} \\
    &\quad\exists u\exists v[
        \begin{aligned}[t]
        x = y \lor [\mathrm{suc}(x,u) &\land u \preceq v \land \mathrm{suc}(v,y) \\
                    &\land \mathrm{zigzag}(x,u;v,y)]]]\,.
        \end{aligned}
\end{align*}
Consequently, we can write the desired formula as
\begin{align*}
  \forall x\forall y[\mathrm{suc}(x,y) \lso [\mathrm{even}(x) \liff \neg\mathrm{even}(y)]]
  \land \exists x\forall y[x \preceq y \land \mathrm{even}(x)]\,.
\end{align*}
\upqed
\end{proof}
\begin{cor}[Boja\'nczyk, Michalewski \cite{BojanczykMiXX}]%
\label{Cor: FO not closed under inverse morphisms}
$\FO$~is not closed under inverse morphisms of\/ $\bbT^\times$-algebras.
\end{cor}
\begin{proof}
Let $\Sigma := \{a,c\}$ and $\Gamma := \{b,c\}$ where $a$~is unary, $b$~binary,
and $c$~a constant, and let $\varphi := \bbT^\times\Sigma \to \bbT^\times\Gamma$ be the
morphism mapping~$a(x_0)$ to $b(x_0,x_0)$ and $c$~to~$c$.
Let $K \subseteq \bbT^\times\Gamma$ be the set of all trees where every leave is at an even depth.
By Lemma~\ref{Lem: even depth FO-definable}, $K$~is $\FO$-definable.
But $\varphi^{-1}[K]$~is the set of all paths $a^n(c)$ where $n$~is even,
which is not $\FO$-definable.
\end{proof}

\begin{thm}
\begin{enuma}
\item $\MSO$ is\/ varietal with respect to the functors\/ $\bbT$~and\/~$\bbT^\times$.
\item $\FO$ is\/ varietal with respect to the functor\/~$\bbT$, but not with respect to\/~$\bbT^\times$.
\end{enuma}
\end{thm}
\begin{proof}
Both claims follow by Theorem~\ref{Thm: theory algebras are L-definable}.
\end{proof}

It follows that the framework we have set up applies to $\MSO$ and $\FO$\?:
(i)~$\MSO$-definable languages have syntactic algebras which, furthermore, are $\MSO$-definable\?;
(ii)~the class of all such languages forms a variety of languages\?;
(iii)~every subvariety can be axiomatised by a set of inequalities.
In particular, we can use Theorem~\ref{Thm: definability theorem} to study
the expressive power of these two logics.

When the functors $\bbT$~and~$\bbT^\times$ were introduced, it was not quite clear which variant
was the right one. The preceding proposition is an indication that $\bbT$~is to be
preferred over~$\bbT^\times$.
For instance, it follows from the general results above that
the syntactic $\bbT$-algebra of every first-order definable language
is first-order definable. Furthermore, we know that there must exist
a set of $\bbT$-inequalities axiomatising first-order definability,
even it is still unknown at the moment how it might look like.
The hope is that such a set of inequalities can be used to devise a decision procedure
for first-order definability of a given tree language.
For simpler logics, the algebraic methods developed in this article have already
sucessfully be used to obtain such decision procedures~\cite{BlumensathX3}.

\bibliographystyle{alpha}
\bibliography{Abstract}

\end{document}